\documentclass[12pt]{article}
\usepackage[T1]{fontenc}
\usepackage[utf8]{inputenc}
\usepackage{amssymb,amsmath}
\usepackage{amsfonts}
\usepackage{mathtools}
\usepackage{amsthm}
\usepackage{upgreek}
\usepackage{enumitem}
\usepackage{amssymb}
\usepackage{amsfonts}

\usepackage{graphicx,psfrag}
\usepackage{amsfonts}
\usepackage{mathtools}

\usepackage{geometry}
\usepackage[colorlinks=true,linkcolor=blue, citecolor=blue, urlcolor = blue]{hyperref}
\usepackage{lmodern}
\usepackage{avant}
\usepackage[round,sort,authoryear]{natbib}
\usepackage{comment}
\usepackage{setspace}
\usepackage{mathptmx}

\counterwithin*{equation}{section}

\theoremstyle{definition}
\newtheorem{theorem}{Theorem}
\newtheorem{definition}{Definition}
\newtheorem{assumption}{Assumption}
\newtheorem{lemma}{Lemma}
\newtheorem{example}{Example}

\newtheorem{proposition}{Proposition}
\newtheorem{remark}{Remark}
\newtheorem{corollary}{Corollary}

\newcommand\norm[1]{\left\lVert#1\right\rVert}

\DeclarePairedDelimiter{\floor}{\lfloor}{\rfloor}

\usepackage{amsmath,amssymb}
\newcommand{\vertiii}[1]{{\left\vert\kern-0.25ex\left\vert\kern-0.25ex\left\vert #1 
    \right\vert\kern-0.25ex\right\vert\kern-0.25ex\right\vert}}

\newif\ifshow 
\showfalse 

\ifshow
  \includecomment{wrap}
\else
  \excludecomment{wrap} 
\fi

\newcommand{\nc}{\newcommand}
\nc{\mbb}{\mathbb}\nc{\bb}{\mathbb}
\nc{\mbf}{\mathbf}\nc{\mb}{\mathbf}
\nc{\mc}{\mathcal}
\nc{\msf}{\mathsf}\nc{\ms}{\mathsf}
\nc{\acc}{\ms{acc}}
\nc{\ack}{\ms{ack}}
\nc{\alp}{\alpha}\nc{\al}{\alpha}\nc{\gka}{\alpha}
\nc{\ap}{\ms{ap}}
\nc{\apd}{\ms{apd}}
\nc{\base}{\ms{base}}\nc{\ba}{\ms{base}}
\nc{\bet}{\beta}\nc{\gkb}{\beta}
\nc{\boucle}{\ms{loop}}\nc{\Loop}{\ms{loop}}\nc{\lo}{\ms{loop}}
\nc{\bu}{\bullet}
\nc*{\cc}{\raisebox{-3pt}{\scalebox{2}{$\cdot$}}}
\nc{\centre}{\ms{center}}\nc{\Center}{\ms{center}}\nc{\cen}{\ms{center}}\nc{\ce}{\ms{center}}
\nc{\ci}{\circ}
\nc{\code}{\ms{code}}\nc{\cod}{\ms{code}}\nc{\decode}{\ms{decode}}\nc{\encode}{\ms{encode}}
\nc{\de}{:\equiv}
\nc{\dr}{\right}\nc{\ga}{\left}
\nc{\ds}{\displaystyle}
\nc{\ep}{\varepsilon}
\nc{\eq}{\equiv}
\nc{\ev}{\ms{ev}}
\nc{\fib}{\ms{fib}}
\nc{\funext}{\ms{funext}}\nc{\fu}{\ms{funext}}
\nc{\gam}{\gamma}
\nc{\glue}{\ms{glue}}\nc{\gl}{\ms{glue}}
\nc{\happly}{\ms{happly}}\nc{\ha}{\ms{happly}}
\nc{\id}{\ms{id}}
\nc{\ima}{\ms{im}}
\nc{\inc}{\subseteq}
\nc{\ind}{\ms{ind}}
\nc{\inl}{\ms{inl}}
\nc{\inr}{\ms{inr}}
\nc{\isContr}{\ms{isContr}}\nc{\co}{\ms{isContr}}\nc{\iC}{\ms{isContr}}\nc{\ic}{\ms{isContr}}
\nc{\isequiv}{\ms{isequiv}}\nc{\iseq}{\ms{isequiv}}\nc{\ieq}{\ms{isequiv}}
\nc{\ishae}{\ms{ishae}}\nc{\ish}{\ms{ishae}}\nc{\ih}{\ms{ishae}}
\nc{\isProp}{\ms{isProp}}\nc{\prop}{\ms{isProp}}\nc{\iP}{\ms{isProp}}\nc{\ip}{\ms{isProp}}
\nc{\isSet}{\ms{isSet}}\nc{\isS}{\ms{isSet}}\nc{\iss}{\ms{isSet}}\nc{\iS}{\ms{isSet}}\nc{\is}{\ms{isSet}}
\nc{\lam}{\lambda}
\nc{\LEM}{\ms{LEM}}\nc{\lem}{\ms{LEM}}\nc{\LE}{\ms{LEM}}
\nc{\lv}{\lvert}\nc{\rv}{\rvert}\nc{\lV}{\lVert}\nc{\rV}{\rVert}
\nc{\Map}{\ms{Map}}
\nc{\merid}{\ms{merid}}\nc{\meri}{\ms{merid}}\nc{\mer}{\ms{merid}}\nc{\me}{\ms{merid}}
\nc{\N}{\bb N}
\nc{\na}{\ms{nat}}
\nc{\nn}{\noindent}
\nc{\one}{\mb1}
\nc{\oo}{\operatorname}
\nc{\pd}{\prod}
\nc{\ps}{\mc P}
\nc{\pa}{\ms{pair}^=}
\nc{\ph}{\varphi}
\nc{\ppmap}{\ms{ppmap}}
\nc{\pr}{\ms{pr}}
\nc{\Prop}{\ms{Prop}}
\nc{\qinv}{\ms{qinv}}\nc{\qin}{\ms{qinv}}\nc{\qi}{\ms{qinv}}
\nc{\rec}{\ms{rec}}
\nc{\refl}{\ms{refl}}
\nc{\seg}{\ms{seg}}
\nc{\Set}{\ms{Set}}
\nc{\sm}{\scriptstyle}
\nc{\sms}{\ms s}
\nc{\sq}{\square}
\nc{\suc}{\ms{succ}}\nc{\su}{\ms{succ}}
\nc{\tb}{\textbf}
\nc{\then}{\Rightarrow}
\nc{\tms}{\ms t}
\nc{\tx}{\text}
\nc{\transport}{\ms{transport}}\nc{\tr}{\ms{transport}}
\nc{\two}{\mb2}
\nc{\Type}{\text-\ms{Type}}\nc{\type}{\text-\ms{Type}}\nc{\ty}{\text-\ms{Type}}
\nc{\U}{\mc U}
\nc{\ua}{\ms{ua}}
\nc{\uniq}{\ms{uniq}}
\nc{\univalence}{\ms{univalence}}
\nc{\vide}{\varnothing}
\nc{\ws}{\ms{sup}}
\nc{\zero}{\mb0}

\title{\textbf{Quantile Time Series Regression Models Revisited}\footnote{Part of these lecture notes were prepared during my Ph.D. studies at the Department of Economics, University of Southampton. I am grateful to Xiaohui Lui, from the School of Statistics at the Jiangxi University of Finance and Economics (Nanchang, China), and to Tassos Magdalinos, Jean-Yves Pitarakis, Jose Olmo and Peter W. Smith  from the School of Economic, Social and Political Sciences as well as Zudi Lu and Chao Zheng from the School of Mathematical Sciences, University of Southampton;  for helpful and stimulating discussions on the study of related econometric and statistical techniques. Lastly, I am grateful to Ji Hyung Lee and Marcelo C. Medeiros from the Department of Economics, University of Illinois at Urbana-Champaign for useful conversations as well as Giuseppe Cavaliere and Sebastian Kripfganz from the University of Exeter Business School.  \\ 

Dr. Christis Katsouris is a Lecturer in Economics at the University of Exeter Business School, Exeter EX4 4PU, United Kingdom. \textit{Email Address}: \textcolor{blue}{c.katsouris@exeter.ac.uk} }}

\author{\textbf{Christis Katsouris}\\ Department of Economics, University of Southampton\\ University of Exeter Business School}

\date{\today}

\begin{document}

\maketitle

\begin{abstract}
\vspace*{-0.8 em}
This article discusses recent developments in the literature of quantile time series models in the cases of stationary and nonstationary underline stochastic processes. 

\end{abstract}


\maketitle

\newpage 
   
\begin{small}
\begin{spacing}{0.9}
\tableofcontents
\end{spacing}
\end{small}

\newpage

\section{Introduction}

The study of systemic risk remains a complex issue which occurs due to financial connectedness and interdependence in markets. Specifically, the increased level of connectivity and interdependence (see for example, \cite{billio2012econometric}, \cite{diebold2012better}, \cite{diebold2014network} among others), can lead to the phenomenon of correlated defaults of financial institutions (see, \cite{duffie2009frailty}). 
Specifically, the dependent variable of interest represents a portfolio loss (e.g.,  large default losses on portfolios of corporate debt). According to \cite{duffie2009frailty}\footnote{The framework of \cite{duffie2009frailty} provide a robust statistical estimation procedure in which the practitioner can construct the distribution of default times and rates based on the frailty variable and unobserved heterogeneity in the model.}: 

\begin{quotation}
"\textit{Any uncertainty about the level of this variable, as well as, the joint exposure to future movements of this variable, can cause a substantial increase in the conditional probability of large portfolio defaults.}"
\end{quotation}
Risk measures such as the CoVaR provide a representation of systemic risk in financial markets. Thus, motivated by the above aspects in these lecture series, we focus on studying existing estimation and inference methodologies for quantile time series models. We begin by reviewing relevant theory for quantile processes (moderate deviation principles e.g., see \cite{mao2019moderate}) as well as estimation of quantile risk measures based on distribution functions. Therefore, we are interested to study of applications of time series regression models based on a conditional quantile functional form in both stationary (e.g., see \cite{he2020inference}, \cite{katsouris2021optimal, katsouris2023statistical}, \cite{escanciano2010specification}) and nonstationary (e.g., see \cite{qu2008testing}, \cite{lee2016predictive}, \cite{xiao2009quantile}, \cite{katsouris2023structural}) time series models. Although we focus on quantile time series regression models we discuss relevant estimation aspects to quantile regression (see, \cite{koenker1978regression}, \cite{koenker2002inference}) in general such as the studies of \cite{victor2005extremal}, \cite{portnoy2012nearly} and \cite{daouia2022extremile} among others. 

A directly relevant framework to both the statistics as well as the econometrics literature is concerned with M-estimation techniques. Specifically, the M-estimation approach is a robust statistical methodology (e.g., robust to outliers and heavy-tails in data) pioneered by \cite{huber1996robust}\footnote{In terms of M-tests these have been introduced in the literature for linear models such as in the paper of \cite{sen1982m}, \cite{sen1986asymptotic} and \cite{sen1987preliminary}. See also \cite{el2001asymptotic}.}. In the econometrics literature it is often used to refer to any estimator based on the maximization or minimization of a criterion function under the assumption of a pair of stationary time series. Thus, for the case of nonstationary time series (integrated, nearly-integrated or using the local-to-unity parametrization i.e., see \cite{phillips2007limit}), more recently the literature has seen growing attention (see, \cite{xiao2012robust}). Suppose that we are working with a parametrized model $( \mathbb{M}, \boldsymbol{\theta} )$. 
  The range of the parameter-defining mapping $\boldsymbol{\theta}$ will be a parameter space $\Theta \in \mathbb{R}^k$. Let $Q^n ( \boldsymbol{y}, \boldsymbol{\theta} )$ denote the value of some criterion function, where $\boldsymbol{y}^n$ is a sample of $n$ observations on one or more dependent variables, and $\boldsymbol{\theta} \in \Theta$. Usually, $Q^n$ will depend on exogenous or predetermined variables as well as dependent variables $\boldsymbol{y}^n$.

\newpage

\begin{definition}
A sequence of criterion functions $Q$ asymptotically identifies a parametrized model  $( \mathbb{M}, \boldsymbol{\theta} )$ if, for all $\mu \in \mathbb{M}$ and for all $\boldsymbol{\theta} \in \Theta$, it holds that 
\begin{align}
\bar{Q} \left( \mu, \boldsymbol{\theta} \right) \equiv \underset{ n \to \infty }{ \mathsf{plim}_{\mu} } Q^n \left( \boldsymbol{y}^n, \boldsymbol{\theta} \right)    
\end{align}
\end{definition}
Our main interest is to determine the asymptotic behaviour of the regression quantile process as defined below for different modelling environments. 
\begin{align}
\label{quanitle.functional}
\boldsymbol{Z}_n := \bigg\{ \boldsymbol{Z}_n ( \alpha ) = n^{1 / 2} \big( \hat{ \boldsymbol{\beta} }_n ( \alpha ) - \boldsymbol{\beta}_n ( \alpha ) \big), 0 < \alpha < 1 \bigg\} 
\end{align}

\subsection{Moderate Deviations for Quantile processes}

\begin{corollary}
Suppose that $f$ is the continuous density of a life distribution function $F( F(0) = 0 )$, and let $\tilde{\rho}_n$ be the PL-normed quantile process. Let $T < 1 \wedge T_{ G( F^{-1}) }$ and assume that $\underset{ 0 \leq y \leq T }{ \text{inf} } \ f \big( F^{-1} (y) \big) > 0$. Then, as $n \to \infty$, with the generalized Kiefer process we have that 
\begin{align}
\underset{ 0 \leq y \leq T }{ \text{sup} } \left| \tilde{\rho}_n - n^{1/2} K(y, n) \right| \to 0. 
\end{align}
\end{corollary}

\begin{proof}
We have that 
\begin{align*}
\underset{ 0 \leq y \leq T }{ \text{sup} } \left| \tilde{\rho}_n - n^{1/2} K(y, n) \right|
\leq \underset{ 0 \leq y \leq T }{ \text{sup} } \left| \tilde{u}_n (y) - n^{1 / 2} K(y,n) \right| + 
\underset{ 0 \leq y \leq T }{ \text{sup} } \left| \tilde{u}_n (y) \right| \underset{ 0 \leq y \leq T }{ \text{sup} } \frac{ \left| f \big( F^{-1} (y) \big) - f \big( F^{-1} ( \theta_{y,n} ) \big) \right| }{ f \big( F^{-1} ( \theta_{y,n} ) \big) } 
\end{align*}
By the continuity property of $f(.)$ it implies its uniform continuity over our assumed finite interval $0 \leq H^{-1} (T) \leq T_H = T_F \wedge T_G$. Moreover, we have that 
\begin{align}
\underset{ 0 \leq y \leq T }{ \text{sup} }  \left| \theta_{y,n} - y   \right| \leq \underset{ 0 \leq y \leq T }{ \text{sup} } \left| \tilde{U}_n (y) - y \right| = \mathcal{O} \left( n^{1 / 2} \big( \mathsf{log} \ \mathsf{log} n \big)^{1 / 2} \right)
\end{align} 
\end{proof}

\subsubsection{Moderate deviations for the empirical quantile processes}

Consider a nondecreasing function $G \in D [a,b]$ and define 
\begin{align*}
G^{-1} (p) = \mathsf{inf} \big\{ x : G(x) \geq p \big\}
\end{align*}
for any $p \in \mathbb{R}$. Moreover, suppose we are interested to the limit behaviour of a test statistic of the following form 
\begin{align}
T_n := \underset{ x \in [0, \uptau] }{ \mathsf{sup} } \left| \hat{F}_n(x) - F_0(x) \right|
\end{align}

\newpage

Then, the rejection region for testing the null hypothesis $H_0$ against $H_1$ is 
\begin{align*}
\left\{ \frac{ \sqrt{n} }{ a(n) } T_n  \geq c \right\}
\end{align*}
where $c$ is a positive constant. 

Moreover, the probability $\alpha_n$ of Type I Error and the probability $\beta_n$ of Type II Error are given by the following expressions 
\begin{align}
\alpha_n = \mathbb{P} \left( \frac{ \sqrt{n} }{ a(n) } T_n \geq c \bigg| F = F_0 \right) \ \ \ \text{and} \ \ \ \beta_n = \mathbb{P} \left( \frac{ \sqrt{n} }{ a(n) } T_n < c \bigg| F = F_1 \right)
\end{align}
such that $a_n$ is the probability of false rejection. 

Therefore, it holds that 
\begin{align}
\beta_n \leq \mathbb{P} \left( \frac{\sqrt{n}}{a(n)} \underset{ x \in [0, \uptau] }{ \mathsf{sup} } \left| \hat{F}_n(x) - F_1(x) \right| \geq \frac{\sqrt{n}}{a(n)} \underset{ x \in [0, \uptau] }{ \mathsf{sup} } \left| F_0(x) - F_1(x) \right| - c \bigg| F = F_1 \right),
\end{align}
Thus, we obtain the following moderate deviation rates 
\begin{align}
\underset{ n \to \infty }{ \mathsf{lim} } \ \frac{1}{ a^2 (n) } \mathsf{log} \left( \alpha_n \right) = - \frac{ c^2 }{ 2 \sigma^2_{km} }, \ \ \ \underset{ n \to \infty }{ \mathsf{lim} } \ \frac{1}{ a^2 (n) } \ \mathsf{log} \left( \beta_n \right) = - \infty.
\end{align}
The above definitions are particularly useful for establishing moderate and large deviation principles for expression \eqref{quanitle.functional} along with suitable rate functions depending on the properties of the underline stochastic process under consideration (see, \cite{victor2005extremal}, \cite{gao2011delta}, \cite{portnoy2012nearly}, \cite{mao2019moderate} and \cite{daouia2022extremile}). Furthermore, from the nonstationary time series perspective, \cite{katsouris2022asymptotic} develops a framework for moderate deviation principles from the unit boundary in nonstationary quantile autoregressions (see, also \cite{knight1991limit} and references therein).

\subsubsection{Conditional Quantiles as Operators}

More recently, \cite{de2023conditional} consider the connection between conditional quantiles and expectation operators which is quite useful for portfolio optimization purposes\footnote{Although for quantiles and conditional quantiles to be used in portfolio optimization problems further regularity conditions are needed to ensure that the statistical problem is well-defined (see, \cite{katsouris2021optimal}). The statistical elicitability and properties of such risk measures is discussed by \cite{fissler2016higher}, \cite{fissler2021elicitability}, \cite{fissler2023backtesting}. Moreover,   \cite{patton2019dynamic}, \cite{lin2023portfolio}  and \cite{corradi2023out} investigate the estimation of these systemic risk measures and their applications to forecasting and backtesting.} and other optimization problems in economics and finance. Given a random variable in a probability space $( \Omega, \mathcal{F}, \mathbb{P} )$ and a $\sigma-$algebra $\mathcal{G} \subset \mathcal{F}$, we want to define the conditional quantile map 
\begin{align}
Q_{\tau} [ X | \mathcal{G} ] : ( \Omega, \mathcal{G} ) \to ( \mathbb{R}, \mathcal{B} ( \mathbb{R} ) )
\end{align}

\newpage

\begin{definition}[\cite{de2023conditional}]
A conditional quantile map $Q_{\tau} \big( X | \mathcal{G} \big): ( \Omega, \mathcal{G} ) \to \big( \mathbb{R}, \mathcal{B}( \mathbb{R} ) \big)$, gives the conditional probability that satisfies $\mathbb{P} \big[ X \leq y | \mathcal{G} \big] ( \omega ) \geq \tau$ such that 
\begin{align}
 Q_{\tau} \big( X | \mathcal{G} \big) \equiv \mathsf{inf} \bigg\{ y \in \mathbb{R} : \mathbb{P} \big[ X \leq y | \mathcal{G} \big] ( \omega )   \geq \tau \bigg\}.   
\end{align}
\end{definition}

\begin{definition}[\cite{de2023conditional}]
Let $\mathbb{P} \big( X \in . | \mathcal{G} \big) : \Omega \times \mathcal{B} ( \Omega ) \to [0,1]$. Them, the $\tau-$quantile random set of $X$ conditional to $\mathcal{G}$ is a map $\Gamma_{\tau} [ X | \mathcal{G} ]: ( \Omega, \mathcal{G} ) \to \left(  \mathcal{K}, \mathcal{B} ( \mathcal{K} ) \right)$ satisfying:  
\begin{align}
\Gamma_{\tau} [ X | \mathcal{G} ]: ( \Omega, \mathcal{G} ) = \underset{ y \in \mathbb{R} }{ \mathsf{argmin} } \int \big( \rho_{\tau} ( x - y) - \rho_{\tau} (x) \big) \mathbb{P} \big( X \in dx | \mathcal{G} \big) ( \omega ), \ \ \ \forall \ \omega \in \Omega.    
\end{align}
\end{definition}
where $\rho_{\tau} ( \cdot )$ is the check function for some $\tau \in (0,1)$ (see, \cite{koenker1978regression}, \cite{koenker1987estimation}) defined as $\rho_{\tau} ( \cdot ) := ( \tau - 1 ) x \cdot \mathbf{1} \left\{ x < 0  \right\} + \tau x \cdot \mathbf{1} \left\{ x \geq 0 \right\}$.

\bigskip

\paragraph{Monotonization} 

Another application for modelling quantile functions when the information set includes of many covariates considers the monotonization property of quantile operators. Suppose that $\mathcal{U}$ is a bounded closed interval such as $\mathcal{U} = [ u_L, u_U ]$ with $0 < u_L < u_U < 1$. The conditional quantile function $\mathcal{Q}_{Y | X} (u|x)$ is monotonically nondecreasing in $u$. However, the plug-in estimator $\mathcal{ \hat{Q} }_{Y | X} (u|x)$ constructed is not necessarily so.  Let $F_{Y | X} ( y | X)$ to denote the conditional distribution of $Y$ given $X$. Then, for every realization of $X$, the map $y \mapsto F_{Y | X} ( y | X )$ is twice continuously differentiable with
\begin{align*}
f_{Y|X (y|X)} 
&=  \frac{ \partial F_{ Y|X } (y|X)  }{\partial y } 
\\
f^{\prime}_{Y|X (y|X)} 
&=
\frac{ \partial f_{ Y|X } (y|X)  }{  \partial y }
\end{align*}
In other words, quantile regression is a means of modelling the conditional quantile function. More precisely, the particular regression methodology has the ability to capture differentiated effects of the explanatory variables on various conditional quantiles of the dependent variable.  Generally speaking, conditional quantiles of $Y$ given $X$ can provide additional information not captured by conditional mean in describing the relationship between $Y$ and $X$. Another important aspect we consider is the classical result of Huber for models with nonstandard conditions such as quantile regression models. Thus, the first order condition (FOC) defined as the right derivative of the objective function plays a key role in deriving the asymptotic theory of estimators for the quantile model. In particular, we can show that the parameter vector estimator solves these FOC and then apply a Bahadur representation for the estimator.

Denote with $Z_n = o_p(1)$ the random sequence $Z_n \to_p 0$, where it denotes convergence in probability. Then, the monotonicity condition $\frac{\partial }{ \partial \tau } x^{\prime} \beta (\tau) =  \big[ f_{Y|X} \big( F^{-1}_{Y|X} (\tau) \big) \big]^{-1} \geq \frac{1}{b}$, implies that $x^{\prime} \beta (\tau)$ must grow by at least a rate of $\delta_n / b$ from $t_k$ to $t_{k+1}$. Using a Bahadur representation, $\big( \hat{\beta} ( \tau_k ) -  \beta ( \tau_k )  \big) = \mathcal{O}_p \left( n^{-1/2} \right)$ uniformly in $k$ and so $\left\{ \hat{\beta}_k \right\}_{k=1}^M$ must also be strictly monotonic with probability tending to one.

\newpage

\subsection{Risk Measures Identification}

The estimation of risk measures such as the Value at Risk, (VaR) and Conditional Value at Risk, (CoVaR) is an important aspect in risk management and portfolio allocation problems. The VaR is computed as the quantile of the loss distribution at a given confidence level, while the CoVaR (or expected shortfall) is the expected loss conditional on a certain level of VaR (see, \cite{acharya2012capital}, \cite{acemoglu2015systemic}). Various studies in the literature investigate the properties of these risk measures as well as related econometric applications (see, \cite{white2015var}, \cite{blasques2016spillover}, \cite{hardle2016tenet}, \cite{chen2019tail}). In particular, the CoVaR has recently been the preferred risk measure of examination due to the fact that it captures the conditional event of one financial institution being under stress given the event of another financial institution being at its Value at risk (see, \cite{tobias2016covar}). The challenge that the econometrician faces is that this risk measure is not elicitable (see, \cite{fissler2016higher},  \cite{patton2019dynamic}), which means that there is no direct loss function which  CoVaR is the solution to the minimization of the expected loss, and therefore a common practise is the use of a two-step estimation procedure for estimation and forecasting purposes.

Since risk measures depend on the state of the economy (see, \cite{tobias2016covar}, \cite{cai2008nonparametric}), then a common practise is to consider an information set, which contains both macroeconomic and financial variables. We shall denote with $Y_t$ the risk or loss of an asset portfolio, for the corresponding information set, denoted as $\mathcal{F}_{t-1}$. In particular, \cite{cai2008nonparametric} considers the conditional $\tau-$quantile of the adapted sequence $\left\{ Y_{t} | \mathcal{F}_{t-1} \right\}_{t=1}^n$ where $\tau \in (0,1)$. Specifically,  the one-period ahead VaR at confidence level $\tau$ is denoted by $Q_{ \tau }( Y_{t} | \mathcal{F}_{t-1} )$. Assume that we have available data $\{ X_t, Y_t \}$ for $t=1,...,n$ which are assumed to be generated from a stationary process. 

Let $Q_{\tau} (x)$ be the CVaR\footnote{Note that CVaR and CoVaR denote two different risk measures in the literature. CVaR implies the estimation of VaR given the information set, while CoVaR denotes the risk measure proposed by \cite{tobias2016covar}. Moreover, CES in the literature has similar definition as the CoVaR, however the CoVaR estimation requires to use the two-step quantile regression procedure as proposed by \cite{tobias2016covar} (see, also \cite{patton2019dynamic}).} which can be expressed as $Q_{\tau}(x) = S^{-1} ( \tau |x )$ where $S( \tau |x ) = 1 - F( y | x )$ and  $F( y | x )$ is the conditional CDF of $Y_t$ given $X_t = x$. Then, \cite{cai2008nonparametric} propose the nonparametric estimation of $Q_{\tau} (x)$ can be constructed as $\widehat{Q_{\tau}} ( x ) = \widehat{S}^{-1} ( \tau | x )$, where $\widehat{S}^{-1} ( \tau | x )$ is a nonparametric estimation of $S^{-1} ( \tau | x )$. Then the CES denoted as $\mu_{\tau} (x)$ is formulated as below
\begin{align}
\mu_{\tau} (x) = \mathbf{E} \big[ Y_t | Y_t \geq \nu_{\tau} (x), X_t = x  \big] = \int_{ \nu_{\tau} (x) }^{\infty} y f( y | x) dy
\end{align}
where $f( y | x)$ represents the conditional PDF of $Y_t$ given $X_t = x$. Therefore, to estimate $\mu_{\tau} (x)$, one can use the plugging-in method as below
\begin{align}
\widehat{\mu_{\tau}} (x) = \int_{ \widehat{\nu_{\tau}} (x) }^{\infty} y \widehat{f}( y | x) dy
\end{align}
Further details on a nonparametric framework suitable for estimating risk measures such as the expected shortfall are given in the study of \cite{martins2018nonparametric}.

\newpage 

In summary suppose that $Q( \tau | x ) := Q \left(  \tau | \underline{X}_j = x \right)$ denotes the conditional $\tau-$th quantile of $Y_j$ given $\underline{X}_j = x$, which is the main conditional function we are interested to estimate and forecast. An equivalent expression in terms of the probability space, is written as below
\begin{align}
\mathbb{P} \bigg[ y_t \leq Q \left( \tau | x \right) \Big| \mathcal{F}_{t-1}  \bigg] = \tau,  \ \ \ \text{where} \ \ \tau \in (0,1).
\end{align}
We estimate the conditional quantile function (CQF) via the following loss function
\begin{align}
\label{quantile.specs}
Q_{\tau} ( y_j | x_j ) = \underset{ q(x)}{\text{argmin}} \ \mathbb{E} \bigg[ \rho_{\tau} \big( y_j - q(x_j) \big) \bigg],
\end{align}
where $\tau \in (0,1)$, is a specific quantile level, and $\rho_{\tau}( u ) = u \big( \tau - \mathbf{1} { \left\{ u < 0 \right\} } \big)$ is the check function. Then conditional quantile function is estimated via
\begin{align}
Q_{y_j} ( \tau | x_j ) = F_{y_j}^{-1}( \tau | x_j )
\end{align}
Furthermore, a linear approximation to the CQF is provided by the QR parameter $\beta ( \tau )$, which solves the population minimization problem described below:
\begin{align}
\beta ( \tau ) =  \underset{ \beta \in \mathbb{R}^{p} }{  \text{arg min} } \ \mathbb{E} \big[ \rho_{\tau}( Y - X^{\prime} \beta ) \big]
\end{align}
under the assumption of integrability and uniqueness of the solution. Therefore, the QR parameter  $\beta ( \tau )$ provides a summary statistic for the CQF. Then, the corresponding QR estimator has the following form
\begin{align}
\hat{ \beta } ( \tau ) \in \underset{ \beta \in \mathbb{R}^{p} }{  \text{arg min} } \frac{1}{n} \sum_{t=1}^n \rho_{\tau} \left( Y_t - X_t^{\prime} b \right)
\end{align}
Note that equivalently, the QR estimator $\hat{ \beta } ( \tau )$ is also the generalized method of moments (GMM) estimator based on the unconditional moment restriction given by 
\begin{align}
\mathbb{E} \bigg[ \bigg( \tau - \mathbf{1} \bigg\{  Y \leq X^{\prime} \hat{ \beta } ( \tau ) \bigg\}  \bigg) X \bigg] = 0.
\end{align}
Moreover, when the CQF is modelled via a linear to the regressors function, such that $Q( Y | X ) = X^{\prime} \beta ( \tau )$ or $F_Y ( X^{\prime} \beta ( \tau ) \big| X ) = \tau$, then the coefficient $\beta ( \tau )$ satisfies the conditional moment restriction given by
\begin{align}
\mathbb{E} \bigg[ \tau - \mathbf{1} \bigg\{ Y \leq X^{\prime} \hat{ \beta } ( \tau ) \bigg\} \bigg| X \bigg] = 0. 
\end{align}
with almost surely convergence. Further relevant literature to the estimation of  risk measures such as the VaR and CoVaR include the framework of \cite{xu2016model}, who consider a  semi-parametric index model with multiple covariates for the joint estimation of the VaR and Expected Shortfall. Moreover, \cite{white2015var}) consider a multivariate specification for the joint estimation of quantile risk measures. 

\newpage 

Therefore, to assess the predictive accuracy of these risk measures, the literature has proposed the methodology of backtesting since the CoVaR is considered to be an unobservable quantity. In particular, some notable studies that propose  backtesting methodologies for the VaR are presented by \cite{escanciano2010backtesting} and \cite{escanciano2011robust}. Specifically, \cite{escanciano2010backtesting}, consider the forecast evaluation problem using backtesting methodologies robust to estimation risk. Therefore, it is of paramount importance to develop powerful tests for the correct specification of parametric conditional quantiles over a possibly continuous range of quantiles of interest and under general conditions on the underlying data-generating process.

\subsection{Properties of Quantile Processes}

\subsubsection{Nearly Root-n approximations for Quantile Processes}

Denote the conditional CDF of $Y$ given $X = x$ by $F_{Y|X} (. | x)$ and its conditional quantile at $\uptau \in (0,1)$ by $Q ( \uptau | x )$, that is, 
\begin{align}
Q ( \uptau | x ) = F_{Y|X}^{-1} ( \uptau | x ) = \mathsf{inf} \big\{ s : F_{Y|X} ( s | x ) \geq \uptau  \big\}
\end{align} 
where $Q( \uptau | x )$ is modelled as a general nonlinear function of $x$ and $\uptau$. We fix $x$ and treat $Q( \uptau | x )$ as a process in $\uptau$, where $\uptau \in \mathcal{T} = [ \lambda_1, \lambda_2 ]$ with $0 < \lambda_1 \leq \lambda_2 < 1$. In this section we follow the framework of \cite{portnoy2012nearly}. Denote with $\dot{\phi}(t)$ the conditional characteristic function of the random variable 
\begin{align}
\dot{x}_i \bigg( I \left( Y_i \leq x_i^{\prime} \beta ( \uptau ) + \delta / \sqrt{n} \right) - \uptau \bigg)
\end{align}
given $x_i$. Moreover, let $f_i(y)$ and $F_i(y)$ denote the conditional density and CDF of $Y_i$ given $x_i$. 
\begin{assumption}[\cite{portnoy2012nearly}]
Let $\norm{ x_i }$ to be uniformly bounded and there are positive definite $p \times p$ matrices $G = G( \uptau )$ and $H$ such that for any $\epsilon > 0$ as $n \to \infty$, 
\begin{align}
G_n( \uptau ) 
&:= 
\frac{1}{n} \sum_{i=1}^n f_i \big( x_i^{\prime} \beta ( \uptau ) \big) x_i^{\prime} x_i 
= 
G( \uptau ) \left( 1 + \mathcal{O} \left( n^{- 1 / 2} \right) \right),
\\
H_n
&:= 
\frac{1}{n} \sum_{i=1}^n  x_i x_i^{\prime} 
= 
H \left( 1 + \mathcal{O} \left( n^{- 1 / 2} \right) \right),
\end{align}
uniformly in $\epsilon \leq \uptau \leq 1 - \epsilon$.
\end{assumption}

\begin{assumption}[\cite{portnoy2012nearly}]
The derivative of $\mathsf{log} \big( f_i(y) \big)$ is uniformly bounded on the interval $
\big\{ y : \epsilon \leq F_i(y) \leq 1 - \epsilon \big\}$.Furthermore, the covariance matrix for $\mathsf{Cov} \big( B_n(\uptau_1),  B_n(\uptau_2) \big)$ has blocks as
\begin{align}
\mathsf{Cov} \big( B_n(\uptau_1),  B_n(\uptau_2) \big) =
\begin{bmatrix}
\uptau_1 ( 1 - \uptau_1 ) \Lambda_{11} \ & \ \uptau_1 ( 1 - \uptau_2) \Lambda_{12}
\\
\uptau_1 ( 1 - \uptau_2 ) \Lambda_{21} \ & \ \uptau_2 ( 1 - \uptau_2) \Lambda_{22}
\end{bmatrix}
\end{align}
where $\Lambda_{ij} = G_n^{-1} ( \uptau_i ) H_n G_n^{-1} ( \uptau_j )$ with $G_n$ and $H_n$ given above.
\end{assumption}

\newpage

Then, \cite{portnoy2012nearly} considers the proof for the Hungarian construction developed inductively. More precisely, it follows that the coverage probability may be computed using only two terms of the Taylor series expansion for the normal CDF:
\begin{align*}
\mathbb{P} \bigg( \sqrt{n} a^{\prime} \big( \hat{\beta} (\uptau)  - \beta ( \uptau ) \big) \leq z_{ \alpha } \sqrt{n} s_{\alpha} (  \hat{\delta} ) \bigg)
&= 
\mathbb{P} \bigg( a^{\prime} \left( W + R_n \right) \leq z_{ \alpha } \sqrt{n} s_{\alpha} (  \hat{\delta} ) + K \bigg)
\\
&= 
\mathbb{E} \left[ \Phi_{ \alpha^{\prime} W | Z } \left( z_{ \alpha } \sqrt{n} s_{\alpha} (  \hat{\delta} ) + K - \alpha^{\prime} R_n \right) \right].
\end{align*} 
In other words, \cite{portnoy2012nearly} employs the ``Hungarian'' construction of \cite{komlos1975approximation} who provide an alternative expansion for the one-sample quantile process with nearly the root-n error rate (see, 
\cite{portnoy2012nearly}). Notice that establishing nearly root$-n$ approximations of quantile-based processes is instrumental for avoiding nearly singular designs (see, \cite{randles1982asymptotic}, \cite{knight2001comparing, knight2008shrinkage} as well as \cite{bhattacharya2020quantile} among others). Thus, motivated by these considerations \cite{portnoy2012nearly} develops a framework to provide an increased accuracy for conditional inference beyond that provided by the traditional Bahadur representation. Specifically, the focus is to provide a theoretical justification for an error bound of nearly root-n order uniformly in $\uptau$. Define with 
\begin{align}
\hat{\delta}_n ( \uptau ) = \sqrt{n} \left( \widehat{\beta}( \uptau ) - \beta( \uptau ) \right).
\end{align}
where $\big[1 / f \big( F^{-1} ( \uptau ) \big) \big]$, corresponds to the sparsity function. Next, we consider a bivariate approximation for the joint density of one regression quantile and the difference between this one and a second regression quantile (properly normalized for the difference in $\uptau-$values). Let $\epsilon \leq \uptau_1 \leq 1 - \epsilon$ for some $\epsilon > 0$, and let $\uptau_2 = \uptau_1 + a_n$ with $a_n > c n^{-b}$ for some $b < 1$. Define with 
\begin{align}
B_n &\equiv B_n( \uptau_1 ) \equiv n^{1 / 2} \big( \hat{\beta}( \uptau_1 ) - \beta ( \uptau_1 ) \big), 
\\
R_n &\equiv R_n( \uptau_1, \uptau_2 ) \equiv \left( n a_n \right)^{1 / 2} \big[ \big( \hat{\beta}( \uptau_1 ) - \beta ( \uptau_1 ) \big) - \big( \hat{\beta}( \uptau_2 ) - \beta ( \uptau_2 ) \big) \big].
\end{align}
The following theorem provides the ``Hungarian'' construction:
\begin{theorem}[\cite{portnoy2012nearly}]
Define $B_n^{\star} ( \uptau )$ to be the piecewise linear interpolant of $\left\{ B_n ( \uptau_j ) \right\}$. Then, for any $\epsilon > 0$, there is a zero-mean Gaussian process, $\left\{ Z_n ( \uptau_j ) \right\}$, defined along the dyadic rationals $\left\{ \uptau_j \right\}$ and with the same covariance structure as $B_n^{\star} ( \uptau )$, along  $\left\{ \uptau_j \right\}$ such that its piecewise linear interpolant $\left\{ Z_n^{\star} ( \uptau_j ) \right\}$ satisfies
\begin{align}
\underset{ \epsilon \leq \tau \leq 1 - \epsilon }{ \mathsf{sup} } \ \big| B_n^{\star} ( \uptau_j ) - Z_n^{\star} ( \uptau_j )  \big| = \mathcal{O} \left( \frac{ (\mathsf{log}n )^{5 / 2} }{ \sqrt{n} }  \right), \ \ \textit{almost surely}.
\end{align}
\end{theorem}

\begin{remark}
Notice that for the development of the proof we focus on extending the density approximation to the joint density for $\hat{\beta}( \uptau_1 )$ and $\hat{\beta}( \uptau_2 )$. However, a major complication is that one needs $a_n \equiv \left| \uptau_2 - \uptau_1 \right| \to 0$, making the covariance matrix tend to singularity. Thus, we focus on the joint density for standardized versions of $\hat{\beta}( \uptau_1 )$ and $D_n \equiv \big( \hat{\beta}( \uptau_2 ) - \hat{ \beta} ( \uptau_1 ) \big)$. Clearly, this requires modification of the proof for the univariate case to treat the fact that $D_n$ converges at a rate depending on $a_n$. 
\end{remark}

\newpage 

\subsubsection{Quantile Dependent Processes Induced by Regression Models}

Quantile regression is a flexible and powerful approach which allows us to model the quantiles of the conditional distribution of a response variable given a set of covariates. Regression quantile estimators can be viewed as M-estimators and standard asymptotic inference is readily available based on likelihood-ratio, Wald and score-type test statistics. However, these statistics require the estimation of the sparsity function and this can lead to nonparametric density estimation. The most appealing feature of QR methods is that they allow estimation of the effect of covariates on many points of the outcome distribution, including the tails as well as the center of the distribution. On the other hand, these do not capture the full distribution impact of a variable unless the variable affects all quantiles of the outcome distribution in the same way. Overall QR methods have the ability to capture heterogeneous effects. In the framework  proposed by \cite{chernozhukov2009finite}, the authors show that valid finite sample confidence regions can be constructed for parameters of a model defined by quantile restrictions under minimal assumptions. The proposed approach makes use of the fact that the estimating equations that correspond to conditional quantile restrictions are are conditionally pivotal; that is, conditional on the exogenous regressors and instruments, the estimating equations are pivotal in finite samples.

\paragraph{Uniform Convergence Rates}

We first establish uniform convergence rates of $\hat{Q}_x ( \uptau )$. The following Bahadur representation of the linear quantile regression estimator $\hat{ \boldsymbol{\beta} }$ is employed in subsequent proofs. Recall that we denote with $Q_{ \boldsymbol{x} }( \uptau )$, where $\uptau \in (0,1)$, the conditional $\uptau-$quantile of Y given $\boldsymbol{x}$ (e.g., see \cite{ota2019quantile}). Notice that we have that the conditional quantile function with respect to the quantile index $\uptau$ coincides with the reciprocal of the conditional density at $Q_{ \boldsymbol{x} }( \uptau )$ such that 
\begin{align}
s_{ \boldsymbol{x} } \left( \uptau \right) := Q^{\prime}_{ \boldsymbol{x} }( \uptau ) \equiv \frac{ \partial Q^{\prime}_{ \boldsymbol{x} } }{ \partial \uptau } 
= 
\frac{1}{ f \big(  Q^{\prime}_{ \boldsymbol{x} } ( \uptau ) \big) }
\end{align}
where $s_{ \boldsymbol{x} } \left( \uptau \right)$ denotes the sparsity function. The estimation of the sparsity function will affect the finite-sample performance of the test statistics if not consistently estimated. A key quantity for establishing related asymptotic theory results is the estimation of the following matrix  
\begin{align}
J ( \uptau ) = \mathbb{E} \bigg[ f \bigg( \boldsymbol{X}^{\prime} \beta ( \uptau ) \bigg) \bigg| \boldsymbol{X}   \bigg] \big( \boldsymbol{X}^{\prime} \boldsymbol{X} \big) \equiv \mathbb{E} \bigg[ f \bigg( X^{\prime} \beta( \uptau)  | X \bigg) X X^{\prime} \bigg].
\end{align}

\begin{lemma}
(\textit{Bahadur representation}) Under Assumption 1, we have that 
\begin{align}
\bigg( \hat{\boldsymbol{\beta}} - \boldsymbol{\beta} \bigg) = J(\uptau)^{-1} \left[ \frac{1}{n} \sum_{t=1}^n \bigg( \uptau - \mathbf{1} \left\{ U_t \leq \uptau \right\} \bigg) \boldsymbol{X}_t \right] + o_p \left(  n^{- 3 / 4} \mathsf{log} (n)  \right)
\end{align}
where $U_1,...,U_n \sim U(0,1)$ \textit{i.i.d} that are independent of $\boldsymbol{X}_1,..., \boldsymbol{X}_n$. In addition, we assume that 
\begin{align}
\underset{ \uptau \in [ \epsilon / 2, - \epsilon / 2 ] }{ \mathsf{sup} } \ \norm{ \frac{1}{n} \sum_{t=1}^n \bigg\{ \uptau - \mathbf{1} \bigg( U_t \leq \uptau \bigg) X_t  \bigg\}  } = \mathcal{O}_p \left( n^{ - 1 / 2} \right).
\end{align}
\end{lemma}

\newpage

Notice that $U_t = F( Y_t | \boldsymbol{X}_t )$ where $F( y | \boldsymbol{X} )$ is the conditional distribution function of Y given $\boldsymbol{X}$. Therefore, to prove the technical lemma we begin by expanding the following expression
\begin{align}
\frac{1}{n} \sum_{t=1}^n \bigg\{ \uptau - \mathbf{1} \bigg( Y_t \leq X_t^{\prime} \hat{\beta}(\uptau) \bigg) X_t  \bigg\}
\end{align} 
Therefore, observe that 
\begin{align*}
\frac{1}{n} \sum_{t=1}^n \bigg\{ \uptau - \mathbf{1} &\bigg( Y_t \leq X_t^{\prime} \hat{\beta}(\uptau) \bigg) X_t  \bigg\} 
=
\frac{1}{n} \sum_{t=1}^n \bigg\{ \uptau - \mathbf{1} \bigg( Y_t \leq X_t^{\prime} \hat{\beta}(\uptau) \bigg) X_t  \bigg\} + \mathbb{E} \bigg[ \bigg\{ \uptau - \mathbf{1} \bigg( Y_t \leq X_t^{\prime} \hat{\beta}(\uptau) \bigg) X_t  \bigg\} \bigg] \bigg|_{ \beta = \widehat{\beta} (\uptau) }
\\
&+ 
\frac{1}{n} \sum_{t=1}^n \bigg\{ \mathbf{1} \bigg( Y_t \leq X_t^{\prime} \beta(\uptau) \bigg) - \mathbf{1} \bigg( Y_t \leq X_t^{\prime} \widehat{\beta}(\uptau) \bigg) \bigg\} X_t    
- \mathbb{E} \bigg[ \bigg\{ \uptau - \mathbf{1} \bigg( Y_t \leq X_t^{\prime} \beta(\uptau) \bigg) X_t  \bigg\} \bigg] \bigg|_{ \beta = \widehat{\beta} (\uptau) }
\end{align*}
Therefore, by the Taylor expansion we obtain that 
\begin{align}
\mathbb{E} \bigg[ \bigg\{ \uptau - \mathbf{1} \bigg( Y_t \leq X_t^{\prime} \beta(\uptau) \bigg) X_t  \bigg\} \bigg] \bigg|_{ \beta = \widehat{\beta} (\uptau) }
= 
- J( \uptau ) \bigg( \widehat{\beta} (\uptau) - \beta(\uptau)     \bigg) + \mathcal{O}_p \left( \frac{1}{n} \right)
\end{align}
uniformly in $\uptau \in [ \epsilon / 2, - \epsilon / 2 ]$. It remains to show that 
\begin{align*}
\norm{ \frac{1}{n} \sum_{t=1}^n \bigg\{ \mathbf{1} \bigg( Y_t \leq X_t^{\prime} \beta(\uptau) \bigg) - \mathbf{1} \bigg( Y_t \leq X_t^{\prime} \widehat{ \beta }(\uptau) \bigg) \bigg\} X_t - \mathbb{E} \bigg[ \bigg\{ \uptau - \mathbf{1} \bigg( Y_t \leq X_t^{\prime} \beta(\uptau) \bigg) X_t  \bigg\} \bigg] \bigg|_{ \beta = \widehat{\beta} (\uptau) }  } = \mathcal{O}_p \left( n^{- 3 / 4} \mathsf{log} (n) \right)
\end{align*}
within the interval $\big[ \epsilon / 2, - \epsilon /2 \big]$.  Thus, since we have $\sqrt{n}-$consistency within the region $\big[ \epsilon / 2, - \epsilon /2 \big]$, that is, $\norm{ \widehat{\beta} - \beta }_{ \big[ \epsilon / 2, - \epsilon /2 \big] }$, for any $M_n \to \infty$ sufficiently slowly such that 
\begin{align}
\mathbb{P} \left( \norm{ \widehat{\beta} - \beta }_{ \big[ \epsilon / 2, - \epsilon /2 \big] }  \leq M_n n^{- 1 / 2} \right) \to 1, 
\end{align} 
Therefore, we consider the following function class as below 
\begin{align*}
\mathcal{F}_n := \bigg\{ (y,x) \mapsto \bigg( \mathbf{1} \left\{ y \leq x^{\prime} \beta \right\} - \mathbf{1} \left\{ y \leq x^{\prime} \left( \beta + \delta \right) \right\} \bigg) a^{\prime} x : \beta \in \mathbb{R}^d , \norm{ \delta } \leq M_n n^{- 1 / 2} , \alpha \in \mathbb{S}^{d-1} \bigg\}, 
\end{align*}
where $\mathbb{S}^{d-1} = \left\{ x \in \mathbb{R}^d : \norm{ x } = 1 \right\}$.

\begin{remark}
Notice that when considering the conditional quantile function $Q_Y ( \tau | X ) = X^{\prime} \beta(\tau)$, for some $\tau \in (0,1)$, then this allow to capture the so-called distribution heterogeneity as the slope of $\beta ( \cdot )$ depends upon the position of the quantile level at the conditional distribution function. Then, optimization is performed at each quantile level. Now, in terms of checking a correct specification of the functional form there are usually two approaches: the first one considers testing the specification of the distribution function implied by the model against the empirical distribution, while the second one is concerned with testing whether the model is correctly specified at the particular quantile level (goodness-of-fit testing). 
\end{remark}

\newpage

\paragraph{Quantile Regression under Misspecification}

Specifically, the Theorem presented in \cite{angrist2006quantile} provides the asymptotic distribution for the stochastic sequence which involves the model parameter. Then, the mapping $\uptau \mapsto \boldsymbol{\beta}(\uptau)$ is continuous by the implicit function theorem and relevant assumptions. In fact, because $\boldsymbol{\beta}(\uptau)$ solves 
\begin{align}
\mathbb{E} \big[ \big( \uptau - \mathbf{1} \left\{ y \leq \boldsymbol{x}_{t-1}^{\prime} \boldsymbol{\beta}(\uptau) \right\}    \big) \boldsymbol{x}_{t-1}  \big] = \boldsymbol{0}
\end{align}
Moreover, it holds that $\frac{ d \boldsymbol{\beta}(\uptau) }{  d\uptau } = J(\uptau)^{-1} \mathbb{E} \left[ X \right]$. Hence, $\uptau \mapsto \mathbb{G}_n \big[ \phi_{\uptau} \big( Y - X^{\prime} \boldsymbol{\beta}(\uptau) \big) X \big]$, is stochastically equicontinuous over $\mathcal{T}_{\iota}$ for the metric given by
\begin{align}
\varrho ( \uptau_1, \uptau_2 ) := \varrho \big( ( \uptau_1, \boldsymbol{\beta}(\uptau_1) ),  ( \uptau_2, \boldsymbol{\beta}(\uptau_2) ) \big)
\end{align}
Thus, stochastic equicontinuity of $\uptau \mapsto \mathbb{G}_n \big[ \phi_{\uptau} \big( Y - X^{\prime} \boldsymbol{\beta}(\uptau) \big) X \big]$ and a multivariate central limit theorem \begin{align}
\mathbb{G}_n \big[ \phi \left( Y - X^{\prime} \boldsymbol{ \uptau }  \right) X  \big] 
\Rightarrow
 z(\uptau) \in \ell^{\infty}(\mathcal{T}_{\iota}),
\end{align}
where $z(\uptau)$ is a Gaussian process with covariance function $]=\Sigma(.,.)$ which implies that 
\begin{align}
\underset{ \uptau \in \mathcal{T}_{\iota} }{ \mathsf{sup} } \norm{ \sqrt{n} \left( \hat{\boldsymbol{\beta}}(\uptau) - \boldsymbol{\beta}(\uptau)  \right) } = O_p(1).
\end{align} 
Therefore, it holds that 
\begin{align}
J(\uptau) \sqrt{n} \left( \hat{\boldsymbol{\beta}}(\uptau) - \boldsymbol{\beta}(\uptau) \right)  
= 
- 
\mathbb{G}_n \big[ \phi \left( Y - X^{\prime} \boldsymbol{ \uptau }  \right) X  \big] + o_p(1)
\Rightarrow
z(\uptau)
\end{align}
Consider that the conditional density $f_Y( y | X = x )$ exists, and is bounded and uniformly continuous in $y$ and uniformly in $x$ over the support of $X$. Define with 
\begin{align}
J(\uptau):= \mathbb{E} \big[ f_Y \big( X^{\prime} \boldsymbol{\beta}(\uptau) | X \big) X X^{\prime} \big]
\end{align}
is positive definite for all $\uptau \in \mathcal{T}_{\iota}$. Then, the quantile regression process is uniformly consistent, 
\begin{align}
\mathsf{sup}_{ \uptau \in \mathcal{T}_{\iota} } \norm{ \hat{\boldsymbol{\beta}}(\uptau) - \boldsymbol{\beta}(\uptau) } = o_p(1), 
\end{align}
and $J(\uptau):= \sqrt{n} \left( \hat{\boldsymbol{\beta}}(\uptau) - \boldsymbol{\beta}(\uptau) \right)$, converges in distribution to a zero mean Gaussian process $z(\uptau)$, where $z(\uptau)$ is defined by its covariance function $\Sigma( \uptau_1, \uptau_2 ) := \mathbb{E} \big[ z(\uptau_1), z(\uptau_2)^{\prime} \big]$ with 
\begin{align}
\Sigma( \uptau_1, \uptau_2 ) 
= 
\mathbb{E} \big[ \big( \uptau_1 - \mathbf{1} \left\{ Y < X^{\prime} \boldsymbol{\beta}(\uptau_1) \right\}  \big) \big( \uptau_2 - \mathbf{1} \left\{ Y < X^{\prime} \boldsymbol{\beta}(\uptau_2)  \right\}  \big) X X^{\prime} \big]
\end{align} 
Thus, if the model is correctly specified, that is, $Q_{\uptau} (Y|X) = X^{\prime} \boldsymbol{\beta}(\uptau)$ almost surely, then $\Sigma( \uptau_1, \uptau_2 )$ simplifies to the following expression 
\begin{align}
\Sigma( \uptau_1, \uptau_2 ) 
= 
\big[ \mathsf{min}( \uptau_1, \uptau_2 ) - \uptau_1\uptau_2   \big] \mathbb{E} [ X X^{\prime} ].
\end{align}

\newpage

Notice that the proof of this theorem in \cite{angrist2006quantile}, proceeds by establishing the uniform consistency of the QR process $\uptau \mapsto \hat{\boldsymbol{\beta}}(\uptau)$. Furthermore, note that the class of functions 
\begin{align}
( \uptau, \beta ) \mapsto \big( \uptau - \mathbf{1} \left\{ Y < X^{\prime} \beta \right\}  X \big)
\end{align}
is Donsker, the estimating equation for the QR process, is given by the following expression 
\begin{align}
n^{-1/2} \sum_{ t = 1}^n \big[ \uptau - \mathbf{1} \left\{ Y_t \leq X_t \hat{\boldsymbol{\beta}}(\uptau) \right\} \big] X_t = o_p(1)
\end{align}
is expanded in $\hat{\boldsymbol{\beta}}(\uptau)$ around the true parameter value $\boldsymbol{\beta}(\uptau)$ since 
\begin{align}
J(\uptau) \sqrt{n} \left( \hat{\boldsymbol{\beta}}(\uptau) - \boldsymbol{\beta}(\uptau) \right)  
=
n^{- 1 / 2} \sum_{ t = 1}^n \big[ \mathbf{1} \left\{ Y_t \leq X_t \hat{\boldsymbol{\beta}}(\uptau) \right\} - \uptau \big] X_t = o_p(1).
\end{align}
uniformly in $\uptau \in \mathcal{T}_{\iota}$. Then, the conclusion of the theorem follows by the central limit theorem for empirical processes indexed by Donsker classes of functions. In summary, the above theorem establishes joint asymptotic normality for the entire QR process. Moreover, the theorem allows for misspecification and imposes a little structure on the underlying conditional quantile function, such as smoothness of $Q_{\uptau} (Y|X)$ in $X$. Therefore, the result in the theorem states that the limiting distribution of the QR process (and of any single QR coefficient) will, in general, be affected by misspecification. In particular, the covariance function that describes the limiting distribution is generally different from the covariance function that arises under correct specification.

\subsection{Nonparametric estimation on conditional quantile processes}

Following the framework of \cite{qu2015nonparametric} (see also \cite{xu2013nonparametric}), we start by reviewing the idea underlying the local linear regression. More precisely, for a given $\uptau \in (0,1)$, the method assumes that $\mathcal{Q} ( \uptau | x )$ is a smooth function of $x$ and exploits the following first-order Taylor approximation: 
\begin{align}
\mathcal{Q} ( \uptau | x_i ) \approx \mathcal{Q} ( \uptau | x ) + \left( x_i - x \right)^{\prime} \frac{ \partial \mathcal{Q} ( \uptau | x ) }{ \partial x }.  
\end{align} 

The local linear estimator of $\mathcal{Q} ( \uptau | x )$, denote by $\hat{\alpha} ( \uptau )$, is determined by 
\begin{align}
\left( \hat{\alpha} ( \uptau) , \hat{\beta} ( \uptau ) \right) = \underset{  \hat{\alpha} ( \uptau), \hat{\beta} ( \uptau)  }{ \mathsf{arg \ min} } \sum_{i=1}^n \big( y_i - \alpha( \uptau) - \left( x_i - x \right)^{\prime} \beta ( \uptau ) \big) \times K \left( \frac{ x_i - x }{ h_{n, \uptau} } \right), 
\end{align} 
where $\rho_{\uptau} (u) = u \big( \uptau - I \left( u < 0 \right) \big)$ is the check function, $K$ is a kernel function and $h_{n, \uptau}$ is a bandwidth parameter that depends on $\uptau$.

\newpage

In particular, the local linear regression has several advantages over the local constant fit such that: (i) the bias of $\hat{\alpha} ( \uptau)$ is not affected by the value of $f_X^{\prime} (x)$ and $\partial \mathcal{Q} ( \uptau | x ) / \partial x$, (ii) it is of the same order irrespective of whether $x$ is a boundary point, and (iii) plug-in data-driven bandwidth selection does not require estimating the derivatives of the marginal density, therefore is relatively simple to implement. 
  
Define with $u_i( \uptau ) = y_i - \alpha( \uptau ) - ( x_i - x )^{\prime} \beta ( \uptau )$, where $\alpha(\uptau) \in \mathbb{R}$ and $\beta(\uptau) \in \mathbb{R}^d$ are some candidate parameter values. Let 
\begin{align}
e_i ( \uptau ) 
&= 
\left[ \mathcal{Q} ( \uptau | x ) + \left( x_i - x \right)^{\prime} \frac{ \partial \mathcal{Q} ( \uptau | x ) }{ \partial x } \right] - \mathcal{Q} ( \uptau | x_i )
\\
\phi( \uptau ) 
&=
\sqrt{ n h_{n, \uptau}^d } 
\begin{pmatrix}
\alpha ( \uptau) - \mathcal{Q} ( \uptau | x ) 
\\
\\
h_{n, \uptau} \left( \beta( \uptau) -  \frac{ \partial \mathcal{Q} ( \uptau | x ) }{ \partial x } \right)
\end{pmatrix}.
\end{align}
Then, $u_i( \uptau )$ can be decomposed as 
\begin{align}
u_i( \uptau ) = u_i^{0} (\uptau ) - e_i( \uptau ) - \left( n h^d_{n, \uptau}  \right)^{- 1 / 2} z_{i, \uptau}^{\prime} \phi ( \uptau)
\end{align}
where $z_{i, \uptau}^{\prime} = \big( 1, \left( x_i - x \right)^{\prime} / h_{n, \uptau} \big)$. 

This decomposition is useful because it breaks $u_i ( \uptau )$ into three components: the true residual, the error due to the Taylor approximation and the error caused by replacing the unknown parameter values in the approximation with some estimates.

\subsection{A general Bahadur representation of M-estimators}

A Bahadur representation of quantile-dependent model coefficients based on M-estimators is essential for establishing asymptotic theory results. Various studies obtain such representations under different modelling conditions. \cite{he1996general}, consider a sequence of variables $\left\{ x_i \right\}_{ i = 1}^n$, that are independent but not necessarily identically distributed, while \cite{wu2005bahadur} considers a Bahadur representation for quantile processes of dependent sequences and \cite{ren2020local} the case of near-epoch dependence. 

Thus, following the framework proposed by \cite{he1996general}, suppose that there exists $\theta_0$ such that 
\begin{align}
\sum_i \mathbb{E} \psi ( \xi, \theta_0 ) = 0,
\end{align}
for some score function $\psi$. Therefore, we consider any $M-$estimator $\hat{\theta}_n$ of $\theta_0$ which satisfies
\begin{align}
\sum_i  \psi \left( x_i, \widehat{\theta}_n \right) = o \big( \delta_n\big ), \ \ \ \text{for some sequence} \ \delta_n.
\end{align}

\newpage

\paragraph{M-estimators with smooth score functions}

Consider the simplest case, where $\widehat{\theta}_n$ is defined through the following moment condition
\begin{align}
\sum_{i=1}^n \phi \left( y_i - z_i^{\prime} \theta \right) z_i = 0, 
\end{align}
where $\phi$ is a Lipschitz continuous function.  Denote with $Q_n = \sum_{i=1}^n z_i z_i^{\prime}$ and denote with $x_i = \left( y_i, z_i \right)$ and $\psi \left( x_i, \theta \right) = \phi \left( y_i - z_i^{\prime} \theta \right) z_i$. Notice that part of $x_i$ has a degenerate point-mass distribution.    

\begin{theorem}
If the following conditions are satisfied, then 
\begin{align}
\widehat{\theta}_n = - \left( \gamma Q_n \right)^{-1} \sum_{i=1}^n \phi ( e_i ) z_i + \mathcal{O} \left( \frac{ \mathsf{log} \mathsf{log} n }{n} \right) \ \ \text{almost surely}, 
\end{align}
\begin{enumerate}
\item[(C1)] both $\phi$ and $f^{\prime}$ are Lipschitz, 

\item[(C2)] $\mathbb{E} \phi (e) = 0$, $\gamma = \int_{- \infty}^{\infty} \phi(x) f^{\prime} (x) dx \neq 0$ and $\mathbb{E} \phi^{2 + \epsilon} (e) < \infty$ for some $\epsilon > 0$, 

\item[(C3)] $n^{-1} Q_n \to Q$ for some positive definite matrix $Q$ and $\sum_{i=1}^n \left| z_i \right|^{ 4 + \epsilon } = \mathcal{O} (n)$ for some $\epsilon > 0$. 

\end{enumerate}
\end{theorem}

\begin{theorem}
(\textit{Uniform Bahadur Representation}) Let Assumptions 1 to 5 hold. Then, the following results hold uniformly in $\uptau \in \mathcal{T}$. 

\begin{itemize}

\item[\textbf{(i)}] If $x$ is an interior point, then we have that 
\begin{align*}
&\sqrt{n h_{n, \uptau}^d } \bigg( \hat{\alpha}( \uptau ) - Q ( \uptau | x ) - d_{\uptau} h^2_{n, \uptau} \bigg)
\\
&=
\frac{ \displaystyle \left(  n d_{n, \uptau}^d \right)^{- 1 / 2}  \sum_{ t=1 }^n \big( \uptau - \mathbf{1} \left( u_t^{0} ( \uptau ) \leq 0 \right) \big) K_{t, \uptau} }{ \displaystyle f_X (x) f_{Y | X} \big( Q( \uptau | x ) \big| x \big) } + o_p(1),
\end{align*}
where $d_{\uptau} = \frac{1}{2} \mathsf{trace} \left( \frac{ \partial^2 Q( \uptau | x ) }{ \partial x \partial x^{\prime} } \mu_2 (K) \right)$, $K_{t, \uptau } = K \left( \frac{x_t - x }{ h_{n, \uptau } } \right)$ and $u_t^0 ( \uptau ) = y_t - Q( \uptau | x_t )$. 

\item[\textbf{(ii)}] If $x$ is a boundary point and Assumptions 6 holds, then
\begin{align*}
&\sqrt{n h_{n, \uptau}^d } \big( \hat{\alpha} ( \uptau ) - Q( \uptau | x ) - d_{b, \uptau } h_{n, \uptau}^2 \big)
\\
&= \frac{ \displaystyle \mathbf{1}_1^{\prime} N_x ( \uptau )^{-1}     \left( n h_{n, \uptau}^d \right)^{- 1 / 2} \sum_{t=1}^n \big( \uptau - \mathbf{1} \left( u_t^0 ( \uptau ) \leq 0 \right) \big) z_{t, \uptau} K_{t, \uptau} }{ \displaystyle f_X (x) f_{Y | X} \big( Q( \uptau | x ) \big| x \big) } + o_p(1),
\end{align*}
Define with $u_t ( \uptau ) = y_t - \alpha ( \uptau ) - ( x_t - x  )^{\prime} \beta ( \uptau )$, where $\alpha ( \uptau ) \in \mathbb{R}$ and $\beta ( \uptau ) \in \mathbb{R}^d$ are some candidate parameter values.  

\end{itemize}
\end{theorem}

Consider the joint density for $\hat{\beta} ( \uptau_1 )$ and $\hat{\beta} ( \uptau_2 )$. One apparent complication is that one needs $a_n \equiv \left| \uptau_1 - \uptau_2 \right| \to 0$, making the covariance matrix tend to singularity. Therefore, the particular approach requires that $D_n$ converges at a rate depending on $a_n$.  

\newpage 

\section{Quantile Time Series Regression Models: Stationary Case}

\subsection{Asymptotic Normality of QR Estimator under i.i.d}

Relevant studies include among others \cite{galvao2020unbiased}. 

\medskip

\begin{assumption}
The QR regression model is given  by
\begin{align}
y_t = \boldsymbol{x}_t^{\prime} \boldsymbol{ \beta } (\uptau) + u_{t} (\uptau), \ \ \ \ \text{with} \ \ \ \uptau \in (0,1).
\end{align} 
\end{assumption}
\begin{assumption}
$\left\{  u_t \right\}$ is an $\textit{i.i.d}$ sequence with distribution function $F$ and assume that $f \left( F^{-1} (\tau) \right) > 0$  in a neigborhood of $\tau$, where $f^{-1}$ is the density function and $F^{-1}$ is the inverse function of $F$.  
\end{assumption}
\begin{assumption}
$\mathbf{X}$ is the design matrix, $\boldsymbol{x}_t$ is the $t-$th column of   $\boldsymbol{x}_t^{\prime}$, $x_{t1} \equiv 1$ and $\mathbb{E}\left( \mathbf{X}^{\prime} \boldsymbol{X} \right) = Q$, is a positive definite matrix. 
\end{assumption}
\begin{assumption}
$\underset{ T \to \infty }{ \text{lim} } $ $\text{max}_{t,s} | x_{ts} | / \sqrt{T} = 0$.
\end{assumption}
Under the above assumptions, the quantile regression estimator has Bahadur representation given by the following expression
\begin{align}
\sqrt{T} \left( \hat{ \boldsymbol{ \beta }} (\uptau)   - \boldsymbol{ \beta } (\uptau) \right) = \frac{  Q^{-1}  }{  f \left(   F^{-1} ( \uptau) \right)  } \frac{1}{\sqrt{T} } \sum_{t=1}^T \boldsymbol{x}_t \psi_{\tau} \big( u_{t} (\uptau) \big) + o_p(1).
\end{align}
where $\psi_{\tau}( . ) = \tau - \mathbf{1} \left\{ . < 0 \right\}$ and $\psi_{  \boldsymbol{ \beta }  }( . ) \equiv \boldsymbol{x}_t \psi_{\tau}( . )$. Thus, the quantile regression estimator has asymptotic normal distribution given by
\begin{align}
D^{-1 / 2} \sqrt{T} \left( \hat{ \boldsymbol{ \beta }} (\uptau)   - \boldsymbol{ \beta } (\uptau) \right) \overset{  d }{ \to } \mathcal{N} \big( 0, \mathbf{I}_k \big).
\end{align}
where $D = \frac{ \displaystyle \uptau (1 - \uptau)  }{ \displaystyle f^2 \left( F^{-1} (\uptau) \right) } Q^{-1}$.

\medskip

Following \cite{xiao2012time}, a simple high-level assumption that we make on the QAR process is monotonicity of the functional form of the model. The monotonicity of the conditional quantile functions allows specific forms for the $\theta$ functions. Moreover, statistical inference can be conducted but the limiting distribution needs to be modified to accommodate the possible misspecification. Additionally, in various occasions it is necessary to employ a bootstrap resampling scheme to approximate the large sample properties of estimators and test statistics (see, \cite{hahn1995bootstrapping}, \cite{hagemann2017cluster}, \cite{galvao2021hac},  \cite{galvao2023bootstrap}).  

\newpage

\begin{example}
Suppose that random variables $Y_1, Y_2...$ are generated by a linear regression model given by 
\begin{align}
Y_t = X_t^{\prime} \beta_0 + \epsilon_t
\end{align}
Then, the quantile regression estimator $b_n$ of $\beta_0$ is obtained by minimizing the check loss function, a standard practise in the quantile regression literature. Both in probability weak convergence and almost sure weak convergence imply that the confidence intervals based on the order statistics of the bootstrap distribution have asymptotically correct coverage probability. 

\medskip

Therefore, it is reasonable to conjecture that the distribution of $\left( \hat{b}_n - b_n \right)$ will be a good approximation of that of $\left( b_n - \beta_0 \right)$. In particular, it has been established that the bootstrap distribution of an $m-$estimator converges weakly in probability to the asymptotic distribution of the $m-$estimator itself under quite general conditions. Thus, since the quantile regression estimator is an $m-$estimator, their limit results implies that $\sqrt{n} \left( \hat{b}_n - b_n \right)$ converges weakly to the asymptotic distribution of $\sqrt{n} \left( b_n - \beta_0 \right)$ in probability. Then, the asymptotic coverage probability of the confidence interval constructed by the bootstrap percentile method is equal to the nominal coverage probability.  

Define with 
\begin{align}
\boldsymbol{J}_0 
&= 
\mathbb{E} \big[ f(0| X_t) X_t X_t^{\prime} \big]
\\
\boldsymbol{M}_0 
&=
\uptau ( 1 - \uptau) \mathbb{E} \big[ X_t X_t^{\prime}  \big]
\end{align}
Then, it holds that 
\begin{align}
\sqrt{n} \big( b_n - \beta_0 \big) \Rightarrow \mathcal{N} \big( 0, \boldsymbol{J}_0^{-1} \boldsymbol{M}_0 \boldsymbol{J}_0^{-1} \big).
\end{align}
\end{example}

\begin{remark}
The $\uptau-$specific parameter vector $\beta(\uptau)$ can be estimated by minimizing the loss function 
\begin{align}
\underset{ \beta(\uptau)  }{ \mathsf{min} } \ \sum_{t=1}^n \rho_{\uptau} \big( y_t - x_t^{\prime} \beta(\uptau) \big)
\end{align}
where $\rho_{\uptau}( \mathsf{u} ) = \mathsf{u} \uptau$ if $u \geq 0$ and $\rho_{\uptau}( \mathsf{u} ) = \mathsf{u} ( \uptau - 1)$ if $u < 0$. The above estimation method consists part of the literature on quantile regression from the frequentist view. Furthermore, notice that a likelihood function approach can be employed based on the asymmetric Laplace distribution. More specifically, assuming that the error term follows an independent asymmetric Laplace distribution, such that
\begin{align}
f_{\uptau}( \mathsf{u} ) = \uptau ( 1 - \uptau) e^{ - \rho_{\uptau} ( \uptau )}, \ \ \ \mathsf{u} \in \mathbb{R}
\end{align}
where $\rho_{\uptau} ( \uptau )$ is the loss function of quantile regression. Therefore, it can be proved that the mode of $f_{\uptau}( \mathsf{u} )$ is the solution of $\underset{ \beta(\uptau)  }{ \mathsf{min} } \ \sum_{t=1}^n \rho_{\uptau} \big( y_t - x_t^{\prime} \beta(\uptau) \big)$. A second important aspect within this setting is the modelling of heteroscedasticity. Notice that a regression model with random coefficients and constant variance can be transformed into a regression with constant coefficients and heteroscedastic error terms.  
\end{remark}

\newpage 

\subsection{Quantile regressions with dependent errors}

Quantile regression can be applied to regression models with dependent errors. In particular, consider the following linear model
\begin{align}
Y_t = \alpha + \boldsymbol{\beta}^{\prime} \boldsymbol{X}_t + u_t = \boldsymbol{ \theta }^{\prime} \boldsymbol{Z}_t + u_t, 
\end{align}
where $X_t$ and $u_t$ are $k$ and $1-$dimensional weakly dependent stationary random variables, $\left\{ X_t \right\}$ and $\left\{ u_t \right\}$ are independent with each other, $\mathbb{E} \left( u_t \right) = 0$. Denote with $F_u(.)$ the distribution function of $u_t$, then conditional on $X_t$, the $\uptau-$th quantile of $Y_t$ is given by the following expression 
\begin{align}
\mathcal{Q}_{Y_t} \left( \uptau | \boldsymbol{X}_t \right) = \alpha + \boldsymbol{\beta}^{\prime} \boldsymbol{X}_t + F_u^{-1} ( \tau ) = \boldsymbol{\theta}( \uptau )^{\prime} \boldsymbol{Z}_t,  
\end{align} 
where $\boldsymbol{\theta} (\uptau ) = \left( \alpha + F_u^{-1} ( \uptau), \boldsymbol{\beta} \right)^{\prime}$. Then, the vector of parameters, $\boldsymbol{\theta} (\uptau )$, can be estimated by solving the optimization problem below
\begin{align}
\widehat{ \boldsymbol{\theta}} (\uptau ) = \underset{ \theta \in \mathbb{R}^p }{ \mathsf{arg \ min} } \sum_{t=1}^n \rho_{\tau} \left( Y_t - \boldsymbol{Z}_t \boldsymbol{\theta} \right). 
\end{align}
Define with $u_{t \uptau} = Y_t - \theta( \uptau )^{\prime} \boldsymbol{Z}_t$, we have that $\mathbb{E} \left[ \psi_{\uptau} ( u_{t \uptau} ) | X_t \right] = 0$. Moreover, under assumptions on moments and weak dependence on $( X_t, u_t )$, 
\begin{align}
n^{- 1 / 2} \sum_{t=1}^n \boldsymbol{Z}_t \psi_{ \uptau } ( u_{t \uptau} ) 
=
\begin{bmatrix}
\displaystyle n^{- 1 / 2} \sum_{t=1}^n \psi_{ \uptau } ( u_{t \uptau} ) 
\\
\displaystyle n^{- 1 / 2} \sum_{t=1}^n \boldsymbol{X}_t \psi_{ \uptau } ( u_{t \uptau} )
\end{bmatrix} 
\Rightarrow
\mathcal{N} \big( \boldsymbol{0}, \boldsymbol{\Sigma} ( \uptau ) \big),
\end{align}
where $\boldsymbol{\Sigma} ( \uptau )$ is the long-run covariance matrix of $\boldsymbol{Z}_t \psi_{ \uptau } ( u_{t \uptau} )$ defined by 
\begin{align}
\boldsymbol{\Sigma} ( \uptau ) = \underset{ n \to \infty }{ \mathsf{lim} } \left( n^{- 1 / 2} \sum_{t=1}^n \boldsymbol{Z}_t \psi_{ \uptau } ( u_{t \uptau} ) \right) \left( n^{- 1 / 2} \sum_{t=1}^n \boldsymbol{Z}_t \psi_{ \uptau } ( u_{t \uptau} ) \right)
= 
\begin{bmatrix}
\omega_{\psi}^2 ( \uptau ) & 0 
\\
0 & \boldsymbol{\Omega} ( \uptau )
\end{bmatrix}. 
\end{align}
Then, the quantile regression estimator has the following asymptotic representation: 
\begin{align}
\sqrt{n} \big( \widehat{\boldsymbol{\theta}} ( \uptau ) - \boldsymbol{\theta} ( \uptau ) \big) 
&= 
\frac{1}{ 2 f \big( F^{-1} ( \uptau ) \big) } \boldsymbol{\Sigma}_z^{-1} \frac{1}{ n^{1/2} } \sum_{t=1}^n \boldsymbol{Z}_t \psi_{ \uptau } ( u_{t \uptau} ), 
\\
\boldsymbol{\Sigma}_z &= \underset{ n \to \infty }{ \mathsf{lim} } \ \frac{1}{n} \sum_{t=1}^n \boldsymbol{Z}_t \boldsymbol{Z}_t^{\prime}.  
\end{align}
As a result, 
\begin{align}
\sqrt{n} \big( \widehat{\boldsymbol{\theta}} ( \uptau ) - \boldsymbol{\theta} ( \uptau ) \big)
\Rightarrow
\mathcal{N} \left( \boldsymbol{0} , \frac{1}{ 4 f \big( F^{-1} ( \uptau ) \big)^2 } \boldsymbol{\Sigma}_z^{-1} \boldsymbol{\Sigma} ( \uptau )  \boldsymbol{\Sigma}_z^{-1} \right).
\end{align}

\newpage

Thus, the above results may be extended to the case where other elements in $\boldsymbol{\theta} ( \uptau )$ are also $\uptau-$dependent. Statistical inference based on $\widehat{\boldsymbol{\theta}} ( \uptau )$ requires estimation of the covariance matrices $\boldsymbol{\Sigma}_z$ and  $\boldsymbol{\Sigma} ( \tau )$. The matrix $\boldsymbol{\Sigma}_z$ can be easily estimated by its sample analogue
\begin{align}
\widehat{ \boldsymbol{\Sigma} }_z = \frac{1}{n} \sum_{t=1}^n \boldsymbol{Z}_t \boldsymbol{Z}_t^{\prime}, 
\end{align}
while $\boldsymbol{\Sigma} ( \uptau )$ may be estimated following the HAC estimation literature. Define with $\widehat{u}_{t \tau} = Y_t - \widehat{ \boldsymbol{\theta} } ( \uptau )^{\prime} \boldsymbol{Z}_t$, we may estimate $\boldsymbol{\Sigma} ( \uptau )$ by
\begin{align}
\widehat{ \boldsymbol{\Sigma} } ( \uptau ) = \sum_{h = -M}^M k \left( \frac{h}{M} \right) \left[ \frac{1}{n} \sum_{ 1 \leq t, t + h \leq n }       \boldsymbol{Z}_t  \psi_{ \uptau } ( \hat{u}_{t \uptau} ) \boldsymbol{Z}_{t+h}^{\prime} \psi_{ \uptau } ( \hat{u}_{ (t+h) \uptau} ) \right], 
\end{align}
where $k(.)$ is the lag window defined on [-1,1] with $k(0) = 1$ and $M$ is the bandwidth parameter satisfying the property that $M \to \infty$ and $M / n \to 0$ as the sample size $n \to \infty$. For instance, the asymptotic properties for regression quantiles with $m-$dependent errors is examined by \cite{portnoy1991asymptotic}. Notice that the delta method is also employed by \cite{girard2021extreme}. 

Overall, when using Wald-type statistics for testing linear restrictions in quantile regression models one follows the following steps. Consider $\uptau \in \mathcal{T}_{\iota}$, where $\mathcal{T}_{\iota } = [ \iota , 1 - \iota  ]$ such that $0 < \iota  < 1/ 2$. Define with 
\begin{align}
\boldsymbol{\mathcal{E}}_{n}^o ( \uptau ) 
&:= 
\sqrt{n} \left( \hat{ \boldsymbol{\beta} }_n  - \boldsymbol{\beta} ( \uptau )  \right)
\\ 
\hat{ \boldsymbol{D} }_{n,0}  
&:=
\frac{1}{n} \sum_{t=1}^n \boldsymbol{x}_{t-1} \boldsymbol{x}_{t-1}^{\prime}
\\
\hat{ \boldsymbol{D} }_{n,1,m} ( \uptau ) 
&:=
\frac{1}{n} \sum_{t=1}^n \boldsymbol{x}_{t-1} \boldsymbol{x}_{t-1}^{\prime} \hat{f}_{i,m,n} \left( \xi_t ( \uptau ) \right)  
\end{align}
where $\hat{f}_{i,m,n} \left( \xi_t ( \uptau ) \right)$ is the conditional density estimator embedding a regression quantile spacing. Then, the following form of the Wald statistic follows
\begin{align}
\mathcal{W}_{n,m} ( \uptau ) \equiv \boldsymbol{\mathcal{E}}_{n}^o ( \uptau )^{\prime} \bigg[ \hat{ \boldsymbol{D} }_{n,1,m} ( \uptau ) \hat{ \boldsymbol{D} }_{n,0} \hat{ \boldsymbol{D} }_{n,1,m} ( \uptau ) \bigg]^{-1} \boldsymbol{\mathcal{E}}_{n}^o ( \uptau )  
\end{align}
The limit distribution of the test statistic $\mathcal{W}_{n,m} ( \uptau )$ is given by the Corollary below. Denote with $\mathcal{W}_{n,m} ( \uptau )$ and recall the following definition $
\boldsymbol{\Delta}_m ( \uptau ) \equiv \underset{ \mathsf{u} }{ \mathsf{arg \ min}  } \ \mathcal{Z} \left( \uptau, \mathsf{u} \right)$, where $\mathcal{Z} \left( \uptau, \mathsf{u} \right)$ is a criterion function. 

\begin{corollary}
Let $\boldsymbol{\mathcal{E}}^o ( \uptau )$ denote a random variable with distribution as the limit distribution of $\boldsymbol{\mathcal{E}}_n^o ( \uptau )$ 
\begin{align}
\boldsymbol{\mathcal{E}}^o ( \uptau ) \sim \mathcal{N} \big( \boldsymbol{0}, \boldsymbol{D}_1^{-1} ( \uptau ) \boldsymbol{D}_0 \boldsymbol{D}_1^{-1} ( \uptau ) \big)
\end{align}
where $\boldsymbol{D}_1 ( \uptau )$ and $\boldsymbol{D}_0$ are described above. Moreover, define with 
\begin{align}
\boldsymbol{D}_{1,m} ( \uptau ) \equiv \underset{ n \to \infty }{ \mathsf{plim} } \ \hat{ \boldsymbol{D} }_{n,1,m} ( \uptau )
\end{align}

\newpage

Then, under the assumptions of Theorem 1 we have 
\begin{align}
\mathcal{W}_{n,m} ( \uptau ) \overset{ d }{ \to } \boldsymbol{\mathcal{E}}^o ( \uptau )^{\prime} \bigg[ \hat{ \boldsymbol{D} }_{1,m} ( \uptau ) \hat{ \boldsymbol{D} }_{n,0} \hat{ \boldsymbol{D} }_{1,m} ( \uptau ) \bigg]^{-1} \boldsymbol{\mathcal{E}}^o ( \uptau ). 
\end{align}
\end{corollary}

\begin{remark}
Notice that the above Corollary implies that these results can be extended to accommodate inferential questions regarding regression-quantile process $
\bigg\{ \sqrt{n} \bigg( \hat{\boldsymbol{\beta}}_n ( \uptau )  - \boldsymbol{\beta} ( \uptau ) \bigg) : 0 < \uptau < 1 \bigg\}$. 
\end{remark}
Therefore, the particular expression is regarded as a stochastic process in the space $\mathcal{D} \left(  [0,1] \right)^d$ of $\mathbb{R}^d-$valued right-continuous functions with left-hand limits on $[0,1]$. Then, based on existing results on the weak convergence of regression quantile processes in $\mathcal{D} \left(  [0,1] \right)^d$, we can establish the limit of the sup-Wald test statistic as given by the following Corollary.  
\begin{corollary}
Suppose the assumptions of Theorem 1 hold for all $\uptau \in (0,1)$. Then, for any $\epsilon \in (0, 1 / 2)$, 
\begin{align}
\underset{ \uptau \in [ \epsilon, 1 - \epsilon ] }{ \mathsf{sup} } \ \mathcal{W}_n ( \uptau ) \overset{ d }{ \to } \underset{ \uptau \in [ \epsilon, 1 - \epsilon ] }{ \mathsf{sup} } \bigg\{ \boldsymbol{\mathcal{E}}^o ( \uptau )^{\prime} \bigg[ \hat{ \boldsymbol{D} }_{1,m} ( \uptau ) \hat{ \boldsymbol{D} }_{n,0} \hat{ \boldsymbol{D} }_{1,m} ( \uptau ) \bigg]^{-1} \boldsymbol{\mathcal{E}}^o ( \uptau )   \bigg\}
\end{align}
where the quantities given in the above Wald formulation refer to the same quantities above.
\end{corollary}

\subsection{Uniform Inference in Quantile Threshold Regression}

Following the framework of \cite{galvao2014testing}, denote with $\left( y_t, q_t, \boldsymbol{x}_t^{\prime} \right)^{\prime}$ be a triple of scalar dependent variable $y_t$, a scalar threshold variable $q_t$ and a vector $d$ of explanatory variables $\boldsymbol{x}_t$. Moreover, denote with $\mathcal{Q}_{ y_t } \left( \uptau | \mathcal{F}_{t-1}   \right)$ the conditional $\tau-$quantile of $y_t$ given the $\sigma-$algebra $\mathcal{F}_{t-1}$, where $\uptau \in (0,1)$, that is a random quantile within the compact set $(0,1)$. Then, consider testing the null hypothesis 
\begin{align}
\mathcal{H}_0: \mathcal{Q}_{ y_t } \left( \uptau | \mathcal{F}_{t-1} \right) = \boldsymbol{x}_t^{\prime} \boldsymbol{\theta}_1 ( \uptau ), \ \ \text{for all} \ \uptau \in \mathcal{T}  
\end{align}
against the alternative hypothesis
\begin{align}
\mathcal{H}_1: \mathcal{Q}_{ y_t } \left( \uptau | \mathcal{F}_{t-1} \right) =  I \left( q_t > \gamma_0 \right) \boldsymbol{x}_t^{\prime} \boldsymbol{\theta}_1 ( \uptau ) + I \left( q_t < \gamma_0 \right)  \boldsymbol{x}_t^{\prime} \boldsymbol{\theta}_2 ( \uptau ), \ \ \text{for some} \ \uptau_0 \in \mathcal{T}  
\end{align}
where $\mathcal{T}:= [ \uptau_L, \uptau_U ]$ is a bounded closed interval in $(0,1)$ and $\gamma_0$ is the threshold parameter. More precisely, the null hypothesis assumes that the conditional quantile function is linear in $\boldsymbol{x}_t$ uniformly over a given range of quantiles, while the alternative hypothesis assumes that the conditional quantile regression function follows a threshold model at some quantile. 

Therefore, to differentiate the alternative from the null hypothesis, we assume that $\boldsymbol{\theta}_1 ( \uptau_0 ) \neq \boldsymbol{\theta}_2 ( \uptau_0 )$. However, for notation convenience we formulate the null hypothesis with different notation. Denote with $\boldsymbol{\beta}_{(1)} ( \uptau_0 ) = \boldsymbol{ \theta }_1 ( \uptau_0 )$ and $\boldsymbol{\beta}_{(2)} ( \uptau_0 ) = \boldsymbol{ \theta }_2 ( \uptau_0 ) - \boldsymbol{ \theta }_1 ( \uptau_0 )$.

\newpage

Thus, the alternative hypothesis can be expressed as below
\begin{align}
\mathcal{H}_1: \mathcal{Q}_{ y_t } \left( \uptau | \mathcal{F}_{t-1} \right) = \boldsymbol{z}_t ( \gamma_0 )^{\prime} \boldsymbol{\beta} ( \uptau_0 ), \ \ \text{with} \ \ \boldsymbol{\beta}_{(2)} ( \uptau_0 ) \neq 0, \ \ \text{for some} \ \ \uptau_0 \in \mathcal{T}, 
\end{align} 
where $\boldsymbol{z}_t ( \gamma ) = \big( \boldsymbol{x}_t^{\prime}, I \left( q_t \leq \gamma \right) \boldsymbol{x}_t^{\prime} \big)^{\prime}$ and $\boldsymbol{\beta} ( \uptau_0 ) := \big( \boldsymbol{\beta}_{(1)}^{\prime} ( \uptau_0 ), \boldsymbol{\beta}_{(2)}^{\prime} ( \uptau_0 )  \big)^{\prime}$. Therefore, using this notation we can write the null hypothesis as below
\begin{align}
\mathcal{H}_0: \mathcal{Q}_{ y_t } \left( \uptau | \mathcal{F}_{t-1} \right) = \boldsymbol{z}_t ( \gamma )^{\prime} \boldsymbol{\beta} ( \uptau ), \ \ \text{with} \ \ \boldsymbol{\beta}_{(2)} ( \uptau ) = 0 \ \ \ \text{for all} \ \ \uptau \in \mathcal{T}, 
\end{align} 
regardless of the value of $\gamma \in \Gamma$. 

Thus, given $( \uptau, \gamma ) \in \mathcal{I} \times \Gamma$, let $\hat{ \boldsymbol{\beta} } ( \uptau, \gamma )$ be the estimator defined by 
\begin{align}
\hat{ \boldsymbol{\beta} } ( \uptau, \gamma ) := \underset{ \boldsymbol{b} \in \mathbb{R}^{2d} }{ \mathsf{arg \ min} } \ \frac{1}{n} \sum_{t=1}^n \rho_{ \uptau } \big( y_t - \boldsymbol{z}_t ( \gamma )^{\prime} \boldsymbol{b} \big),
\end{align} 

In other words, $\hat{ \boldsymbol{\beta} } ( \uptau, \gamma )$ represents the quantile regression estimator when we treat $\boldsymbol{z}_t ( \gamma )$ as explanatory variables. Furthermore, when $\mathcal{H}_0$ is true, under suitable regulatory conditions, $\hat{ \boldsymbol{\beta} }_2 ( \uptau, \gamma )$ converges in probability to $\mathbf{0}$ for each $( \uptau, \gamma ) \in \mathcal{T} \times \Gamma$. On the other hand, when $\mathcal{H}_1$ is true, then $\hat{\boldsymbol{\beta}}_2 ( \uptau_0, \gamma_0 )$ converges in probability to $\hat{\boldsymbol{\beta}}_{(2)} ( \uptau_0 )\neq 0$. However, we know a priori neither the quantile $\uptau_0$ where the linearity breaks down nor the true value of the threshold parameter $\gamma_0$ at that quantile. Therefore, it is reasonable to reject $\mathcal{H}_0$ if the magnitude of $\hat{ \boldsymbol{\beta} }_2 ( \uptau_0, \gamma_0 )$ is suitably large for some $( \uptau, \gamma ) \in \mathcal{T} \times \Gamma$.

A natural choice is to test $\mathcal{H}_0$ against $\mathcal{H}_1$ by the supremum of the Wald process as below
\begin{align}
SW_n := \underset{ ( \uptau, \gamma ) \in \mathcal{T} \times \Gamma }{ \mathsf{sup}  } \ n \hat{ \boldsymbol{\beta} }_{(2)} ( \uptau, \gamma )^{\prime} \bigg[ \boldsymbol{V}_{22} ( \uptau, \gamma ) \bigg]^{-1} \hat{ \boldsymbol{\beta} }_{(2)} ( \uptau, \gamma )
\end{align} 
where $\boldsymbol{V}_{22} ( \uptau, \gamma )$ is the asymptotic covariance matrix of $\sqrt{n} \hat{\boldsymbol{\beta}}_{(2)} ( \uptau, \gamma )$ under $\mathcal{H}_0$. 

\begin{assumption}[\cite{galvao2014testing}]
\label{assumption1}
There exists an open set $\mathcal{T}^{*} \subset (0,1)$ such that for each $\uptau \in \mathcal{T}^{*}$, there exists a unique vector $\boldsymbol{\beta}_{(1)}^{*} \in \mathbb{R}^p$ that solves the equation below
\begin{align}
\mathbb{E} \bigg[ \bigg( \uptau - I \left( y_t \leq \boldsymbol{x}_t^{\prime} \hat{\boldsymbol{\beta}}^{*}_{(1)} ( \uptau ) \right) \bigg) \boldsymbol{x}_t \bigg] = \mathbf{0}. 
\end{align} 
\end{assumption}
Therefore, when $\mathcal{H}_0$ is true, $\boldsymbol{\beta}_{(1)} ( \uptau ) = \boldsymbol{\beta}_{(1)}^{*} ( \uptau )$ and $ \mathcal{Q}_{ y_t } \left( \uptau | \mathcal{F}_{t-1} \right) = \boldsymbol{x}_t^{\prime} \boldsymbol{\beta}_{(1)}^{*} ( \uptau )$ for all $\uptau \in \mathcal{T}$. When $\mathcal{H}_0$ is not true,   
$\boldsymbol{\beta}_{(1)}^{*} ( \uptau )$ it can be interpreted as the coefficient vector of the best linear predictor of the conditional quantile function against a certain weighted mean-squared loss function. In other words, the condition given by Assumption \ref{assumption1} above demonstrates that the limiting null distribution of $SW_n$ depends on the probability limit of the quantile regression estimator $\hat{\boldsymbol{\beta}}_1 ( \uptau, \gamma )$ under the null hypothesis.

\newpage

\begin{assumption}[\cite{galvao2014testing}]
\label{assumption2}
Define the following matrices
\begin{align}
\boldsymbol{\Omega}_0 \left( \gamma_1, \gamma_2 \right) := \mathbb{E} \bigg[ \boldsymbol{z}_t ( \gamma_1 ) \boldsymbol{z}_t ( \gamma_2 )^{\prime} \bigg] \ \ \ \ \boldsymbol{\Omega}_1 \left( \uptau, \gamma \right) := \mathbb{E} \bigg[ f \left( \boldsymbol{x}_t \boldsymbol{\beta}_{(1)}^{*} ( \uptau ) \big| \boldsymbol{z}_t \right)\boldsymbol{z}_t ( \gamma ) \boldsymbol{z}_t ( \gamma )^{\prime} \bigg].
\end{align}
Then, we assume that $\boldsymbol{\Omega}_0 \left( \gamma_1, \gamma_2 \right)$ is a positive-definite matrix for each $\gamma \in \Gamma$, and $\boldsymbol{\Omega}_1 \left( \uptau, \gamma \right)$ is positive-definite for each $( \uptau, \gamma ) \in \mathcal{T} \times \Gamma$. 
\end{assumption}
Thus, under these conditions, the map $\uptau \in \mathcal{T}^{*} \mapsto \boldsymbol{\beta}_{(1)}^{*} ( \uptau )$ is continuously differentiable by the implicit function theorem. Furthermore, Assumption \ref{assumption2} guarantees that the matrices $\boldsymbol{\Omega}_0 \left( \gamma_1, \gamma_2 \right)$ and $\boldsymbol{\Omega}_1 \left( \uptau, \gamma \right)$ do not degenerate for each fixed $\gamma \in \Gamma$ and $( \uptau, \gamma ) \in \mathcal{T} \times \Gamma$, respectively. Due to the computational property of the quantile regression estimate we can select $\hat{ \boldsymbol{\beta} } ( \uptau, \gamma )$ in such a way that the path $( \uptau, \gamma ) \mapsto \hat{ \boldsymbol{\beta} } ( \uptau, \gamma )$ is bounded. Therefore, we can also assume that the path 
\begin{align}
( \uptau, \gamma ) \mapsto \sqrt{n}\bigg( \hat{ \boldsymbol{\beta} } ( \uptau, \gamma ) - \left(  \boldsymbol{\beta}_{(1)}^{*} ( \uptau )^{\prime} , \boldsymbol{0}  ^{\prime} \right)^{\prime} \bigg)
\end{align}
is bounded over $( \uptau, \gamma ) \in \mathcal{T} \times \Gamma$.   

Denote with $\ell^{\infty} \left( \mathcal{T} \times \Gamma   \right)$ the space of all bounded functions on $\mathcal{T} \times \Gamma$ equipped with the uniform topology, and $\left( \ell^{\infty} \left( \mathcal{T} \times \Gamma \right) \right)^{2d}$ to denote the $(2d)-$product space of $\ell^{\infty} \left( \mathcal{T} \times \Gamma \right)$ equipped with the product topology. Here, we denote with $\boldsymbol{\beta}^{*} ( \uptau ):= \left( \boldsymbol{\beta}_{(1)}^{*} ( \uptau )^{\prime}, \boldsymbol{0}    ^{\prime}  \right)^{\prime} \in \mathbb{R}^{2d}$.

\begin{theorem}[\cite{galvao2014testing}]
Assume conditions (C1)-(C6) hold. 
\

Then, under $\mathcal{H}_0: \mathcal{Q}_{ y_t } \left( \uptau | \mathcal{F}_{t-1} \right) = \boldsymbol{x}_t^{\prime} \boldsymbol{\beta}_{(1)}^{*} ( \uptau )$ for all $\uptau \in \mathcal{T}$, the random quantity $\sqrt{n} \bigg( \hat{ \boldsymbol{\beta} } ( \uptau, \gamma ) - \boldsymbol{\beta}^{*} ( \uptau, \gamma ) \bigg)$ admits the Bahadur representation given by the expression below
\begin{align}
\sqrt{n} \bigg( \hat{ \boldsymbol{\beta} } ( \uptau, \gamma ) - \boldsymbol{\beta}^{*} ( \uptau ) \bigg) 
= 
\boldsymbol{\Omega}_1 \left( \uptau, \gamma \right)^{-1} \frac{1}{ \sqrt{n} } \sum_{t=1}^n \bigg[ \uptau - I \bigg( y_t \leq \boldsymbol{x}_t^{\prime} \boldsymbol{\beta}_{(1)}^{*} ( \uptau ) \bigg) \bigg] \boldsymbol{z}_t \left( \gamma \right) + r_n ( \uptau, \gamma ),  
\end{align}
where 
\begin{align*}
\mathsf{sup}_{ ( \uptau, \gamma ) \in \mathcal{T} \times \Gamma } \norm{ r_n ( \uptau, \gamma ) } = o_p(1)
\end{align*}
Therefore, under $\mathcal{H}_0$, 
\begin{align}
\sqrt{n} \bigg( \hat{ \boldsymbol{\beta} } ( \uptau, \gamma ) - \boldsymbol{\beta}^{*} ( \uptau ) \bigg) 
\Rightarrow 
\boldsymbol{\Omega}_1 \left( \uptau, \gamma \right)^{-1} \boldsymbol{W} \left( \uptau, \gamma \right) \ \ \text{in} \  \left( \ell^{\infty} \left( \mathcal{T} \times \Gamma \right) \right)^{2d},
\end{align}
where $\boldsymbol{W} \left( \uptau, \gamma \right)$ is a zero-mean, continuous Gaussian process on  $\mathcal{T} \times \Gamma$ with covariance kernel 
\begin{align}
\mathbb{E} \bigg[ \boldsymbol{W} \left( \uptau_1, \gamma_1 \right)     \boldsymbol{W} \left( \uptau_1, \gamma_1 \right)^{\prime} \bigg] = \left( \uptau_1 \wedge \uptau_2 - \uptau_1 \uptau_2 \right) \boldsymbol{\Omega}_0 \left( \gamma_1 , \gamma_2 \right).   
\end{align}
\end{theorem}

\newpage

\begin{proof}
For any fixed vector $\boldsymbol{v} \in \mathbb{R}^d$, define the stochastic processes given below
\begin{align*}
\mathcal{U}_n \left( \uptau, \gamma \right) 
&:= 
\frac{1}{ \sqrt{n} } \sum_{t=1}^n \boldsymbol{v}^{\prime} \boldsymbol{x}_t I \left( q_t \leq \gamma \right) \bigg[ \uptau - I \left( y_t \leq  \boldsymbol{x}_t^{\prime} \boldsymbol{\beta}^{*}_{(1)} ( \uptau ) \right) \bigg],
\\
\mathcal{V}_n \left( \uptau, \gamma \right) 
&:= 
\frac{1}{ \sqrt{n} } \sum_{t=1}^n \bigg[ \boldsymbol{v}^{\prime} \boldsymbol{x}_t I \left( q_t \leq \gamma \right) I \bigg( y_t \leq       \boldsymbol{x}_t^{\prime} \left( \boldsymbol{\beta}^{*}_{(1)} ( \uptau ) + s n^{- 1/ 2} \boldsymbol{v} \right) \bigg) \bigg] 
\\
&- \mathbb{E} \bigg[ \boldsymbol{v}^{\prime} \boldsymbol{x}_t I \left( q_t \leq \gamma \right) F \left( \boldsymbol{x}_t^{\prime} \left( \boldsymbol{\beta}^{*}_{(1)} ( \uptau ) + s n^{- 1/ 2} \boldsymbol{v} \right) \bigg| \boldsymbol{z}_t  \right) \bigg]
\end{align*}
where $\left( \uptau, \gamma, s \right) \in \mathcal{T} \times \Gamma \times [0,1]$. 

We begin our analysis by investigating the asymptotic behaviour of  $\mathcal{U}_n$ and $\mathcal{V}_n$. Denote with $\mathcal{K} \left( \uptau, \lambda  \right)$ to be the Kiefer process on $[0,1] \times [0, + \infty)$ such that $\mathcal{K} \left( \uptau, \lambda  \right)$ is a zero-mean continuous Gaussian process on $[0,1] \times [0, + \infty)$ with covariance kernel given by 
\begin{align}
\mathbb{E} \bigg[ \mathcal{K} \left( \uptau_1, \lambda_1 \right) \mathcal{K} \left( \uptau_2, \lambda_2 \right) \bigg] = \lambda_1 \wedge \lambda_2 \left( \uptau_1 \wedge \uptau_2  - \uptau_1 \uptau_2  \right).
\end{align}
Furthermore, define with $\mathcal{H} ( \gamma ) := \mathbb{E} \big[ \left( \boldsymbol{v}^{\prime} \boldsymbol{x}_t \right)^2 I \left( q_t \leq \gamma \right) \big]$. Under conditions (C2) and (C3), the map $\gamma \mapsto \mathcal{H} ( \gamma )$. Under conditions (C2) and (C3), the map $\gamma \mapsto \mathcal{H} ( \gamma )$ is continuous and non-decreasing.

\textbf{Part (i)} Observe that under the null hypothesis, $\bigg\{ \uptau - I \left( y_t \leq  \boldsymbol{x}_t^{\prime} \boldsymbol{\beta}^{*}_{(1)} ( \uptau ) \right),  t \in \mathbb{Z} \bigg\}$ is a martingale difference sequence with respect to $\big\{ \mathcal{F}_{t-1}, t \in \mathbb{Z} \big\}$, that is, $\mathbb{E} \left[ I \left( y_t \leq \boldsymbol{x}_t^{\prime} \hat{\boldsymbol{\beta}}^{*}_{(1)} ( \uptau ) \right) \bigg|  \mathcal{F}_{t-1} \right] = \uptau$. The finite dimensional convergence follows from the martingale central limit theorem. Therefore, it remains to show the stochastic equicontinuity of the process $\mathcal{U}_n \left( \uptau, \gamma \right)$.

More precisely, we decompose the stochastic process $\mathcal{U}_n \left( \uptau, \gamma \right)$ as below
\begin{align*}
\mathcal{U}_n \left( \uptau, \gamma \right) 
&:= 
\frac{ \uptau }{ \sqrt{n} } \sum_{t=1}^n \boldsymbol{v}^{\prime} \boldsymbol{x}_t I \left( q_t \leq \gamma \right)  -  \frac{ 1 }{ \sqrt{n} } \sum_{t=1}^n \boldsymbol{v}^{\prime} \boldsymbol{x}_t I \left( q_t \leq \gamma \right) I \left( y_t \leq  \boldsymbol{x}_t^{\prime} \boldsymbol{\beta}^{*}_{(1)} ( \uptau ) \right)
\\
&=
\frac{ \uptau }{ \sqrt{n} } \sum_{t=1}^n \bigg\{ \boldsymbol{v}^{\prime} \boldsymbol{x}_t I \left( q_t \leq \gamma \right) - \mathbb{E} \bigg[ \boldsymbol{v}^{\prime} \boldsymbol{x}_t I \left( q_t \leq \gamma \right) \bigg] \bigg\}  -  V_n \left( \uptau, \gamma, 0 \right). 
\end{align*}

\medskip

Define the stochastic process as below
\begin{align}
\tilde{V}_n ( \gamma ) := n^{- 1 / 2} \sum_{t=1}^n \bigg[ \boldsymbol{v} ^{\prime}  \boldsymbol{x}_t I \left( q_t \leq \gamma \right) \bigg] - \mathbb{E} \bigg[ \boldsymbol{v} ^{\prime}  \boldsymbol{x}_t I \left( q_t \leq \gamma \right) \bigg]
\end{align}
Then, we have that 
\begin{align}
U_n ( \uptau, \gamma ) = \uptau \tilde{V}_n ( \gamma ) - V_n ( \uptau, \gamma, 0 )
\end{align}

\newpage

Thus, by the previous calculation, we see that the following convergence result holds
\begin{align}
\rho \left( \left( \gamma_1, \boldsymbol{\beta}_{(1)}^{*} \left( \uptau_1 \right)  \right), \left( \gamma_2, \boldsymbol{\beta}_{(1)}^{*} \left( \uptau_2 \right) \right) \right) \to 0, \ \ \text{as} \ \ \norm{ \left( \uptau_1 - \uptau_2 , \gamma_1 - \gamma_2 \right) } \to 0,
\end{align}
which implies that the process $( \uptau, \gamma ) \mapsto V_n ( \uptau, \gamma, 0 )$ is stochastically equi-continuous over $\mathcal{T} \times \Gamma$ with respect to the Euclidean metric. Similarly, it can be shown that the process $\upgamma \mapsto \tilde{V}_n ( \gamma )$ is stochastically equi-continuous over $\Gamma$ with respect to the Euclidean metric (see,  \cite{galvao2014testing}).  
\end{proof}
We introduce the local objective function as below
\begin{align}
\mathcal{Z}_n \left( \boldsymbol{u}, \uptau, \gamma \right) :=  \sum_{t=1}^n \left\{ \rho_{\uptau} \left( y_t - \boldsymbol{x}_t^{\prime} \boldsymbol{\beta}^{*}_{(1)} ( \uptau ) - n^{- 1 / 2} \boldsymbol{u}^{\prime} z_t ( \gamma ) \right) - \rho_{\uptau} \left( y_t - \boldsymbol{x}_t^{\prime} \boldsymbol{\beta}^{*}_{(1)} ( \uptau ) \right) \right\}
\end{align}
where $\boldsymbol{u} \in \mathbb{R}^{2d}$ and $( \uptau, \gamma ) \in \mathcal{T} \times \Gamma$. Furthermore, we observe that the normalized random quantity $\sqrt{n} \left( \hat{ \boldsymbol{\beta} } ( \uptau, \gamma ) - \boldsymbol{\beta}^{*} ( \uptau ) \right)$ minimizes $\mathcal{Z}_n \left( \boldsymbol{u}, \uptau, \gamma \right)$ with respect to $\boldsymbol{u}$ for each fixed $( \uptau, \gamma ) \in \mathcal{T} \times \Gamma$. Thus, to prove Theorem 1, we utilize Theorem 2 of \cite{kato2009asymptotics}. More specifically, since $\mathcal{Z}_n \left( \boldsymbol{u}, \uptau, \gamma \right)$ is now convex in $\boldsymbol{u}$, by Theorem 2 of \cite{kato2009asymptotics}\footnote{Notice that Theorem 2 of \cite{kato2009asymptotics} guarantees that there is no need to establish the uniform $n^{- 1/ 2}$ rate of $\hat{ \boldsymbol{\beta} } \left( \uptau, \gamma \right)$ in the course of establishing the weak convergence of the process.}, it is sufficient to prove the following proposition.      

\medskip

\begin{proposition}
Assume conditions (C1)-(C6). Then, under $\mathcal{H}_0$, 
\begin{align*}
\mathcal{Z}_n \left( \boldsymbol{u}, \uptau, \gamma \right) 
&= 
- \frac{1}{\sqrt{n} } \sum_{t=1}^n \bigg[ \uptau - I \big\{ y_t \leq \boldsymbol{x}_t^{\prime} \boldsymbol{\beta}^{*}_{(1)} (\uptau) \big\} \bigg] \boldsymbol{u}^{\prime} \boldsymbol{z}_t ( \gamma )
\\
&+ \frac{1}{2} \boldsymbol{u}^{\prime} \boldsymbol{\Omega}_1 ( \uptau, \gamma ) \boldsymbol{u} + \Delta_n \left(  \boldsymbol{u}, \uptau, \gamma    \right), 
\end{align*}
where $\mathsf{sup}_{ ( \uptau, \gamma ) \in \mathcal{T} \times \Gamma } \left| \Delta_n \left(  \boldsymbol{u}, \uptau, \gamma    \right) \right| = o_p(1)$ for each fixed $\boldsymbol{u} \in \mathbb{R}^{2d}$, and 
\begin{align}
\frac{1}{ \sqrt{n} } \sum_{t=1}^n \bigg[ \uptau - I \big\{ y_t \leq \boldsymbol{x}_t^{\prime} \boldsymbol{\beta}^{*}_{(1)} (\uptau) \big\} \bigg] \boldsymbol{z}_t ( \gamma ) \Rightarrow \boldsymbol{W} \left( \uptau, \gamma \right) \ \ \text{in} \ \ \left( \ell^{\infty} \left( \mathcal{T} \times \Gamma \right) \right)^{2d}.
\end{align}
\end{proposition}
Using \cite{knight1998limiting}  identity we obtain 
\begin{align}
\rho_{\uptau} (x - y) - \rho_{ \uptau } (x) = - y \bigg\{ \uptau - I (x \leq 0 ) \bigg\} + y \int_0^1 \bigg[ I \big( x \leq ys \big) - I \big( x \leq 0 \big) \bigg] ds,
\end{align}
Then, we decompose the stochastic process $\mathcal{Z}_n \left( \boldsymbol{u}, \uptau, \gamma \right)$ into three parts as below: 
\begin{align*}
\mathcal{Z}_n \left( \boldsymbol{u} , \uptau, \gamma \right) 
= 
&- \frac{1}{\sqrt{n} } \sum_{t=1}^n \boldsymbol{u}^{\prime} \boldsymbol{z}_t ( \gamma ) \bigg[ \uptau - I \big\{ y_t \leq \boldsymbol{x}_t^{\prime} \boldsymbol{\beta}^{*}_{(1)} (\uptau) \big\} \bigg] 
\\
&+ \frac{1}{n} \sum_{t=1}^n \boldsymbol{u}^{\prime} \boldsymbol{z}_t (\gamma) \int_0^1 \sqrt{n} \bigg[ I \left\{ y_t - \boldsymbol{x}_t^{\prime} \boldsymbol{\beta}^{*}_{(1)} (\uptau) \leq n^{- 1 / 2} \boldsymbol{u}^{\prime} \boldsymbol{z}_t (\gamma) \right\} - I \left\{ y_t - \boldsymbol{x}_t^{\prime} \boldsymbol{\beta}^{*}_{(1)} (\uptau) \leq 0 \right\} \bigg] ds  
\end{align*}

\newpage

We denote with 
\begin{align}
Z_n^{(21)} \left( \boldsymbol{u} , \uptau, \gamma \right) = \mathbb{E} \left[ \boldsymbol{u}^{\prime} \boldsymbol{z}_t (\gamma) \int_0^1 \sqrt{n} \bigg\{ F \left( \boldsymbol{x}_t^{\prime} \boldsymbol{\beta}^{*}_{(1)} (\uptau) + s n^{- 1 / 2} \boldsymbol{u}^{\prime} \boldsymbol{z}_t (\gamma) \bigg| \boldsymbol{z}_t \right) - \uptau \bigg\}  ds  \right]
\end{align} 
Since we have that
\begin{align}
F \left( \boldsymbol{x}_t^{\prime} \boldsymbol{\beta}^{*}_{(1)} (\uptau) \bigg| \boldsymbol{z}_t \right) 
= 
\mathbb{E} \bigg[ I \left\{ y_t - \boldsymbol{x}_t^{\prime} \boldsymbol{\beta}^{*}_{(1)} (\uptau) \right\} \bigg| \boldsymbol{z}_t \bigg] 
= 
\mathbb{E} \bigg[ \mathbb{E} \bigg[ I \left\{ y_t - \boldsymbol{x}_t^{\prime} \boldsymbol{\beta}^{*}_{(1)} (\uptau) \right\}   \bigg| \mathcal{F}_{t-1} \bigg] \bigg| \boldsymbol{z}_t \bigg]
= 
\uptau
\end{align} 
Therefore, by the dominated convergence theorem, we have that  
\begin{align}
\underset{ ( \uptau, \gamma ) \in \mathcal{T} \times \Gamma }{ \mathsf{sup} } \ \bigg| Z_n^{(21)} \left( \boldsymbol{u} , \uptau, \gamma \right) - \frac{1}{2} \boldsymbol{u}^{\prime} \Omega_1 ( \uptau, \gamma ) \boldsymbol{u} \bigg| \to 0, \ \ \forall \ \boldsymbol{u} \in \mathbb{R}^{2d}.
\end{align} 
Furthermore, we define with 
\begin{align*}
R_n ( \uptau, \gamma, s ) 
&:= 
\frac{1}{ \sqrt{n} } \sum_{t=1}^n \boldsymbol{u}^{\prime} \boldsymbol{z}_t (\gamma) I \left\{ y_t - \boldsymbol{x}_t^{\prime} \boldsymbol{\beta}^{*}_{(1)} (\uptau) \leq n^{- 1 / 2} \boldsymbol{u}^{\prime} \boldsymbol{z}_t (\gamma) \right\} 
\\
&- 
\frac{1}{ \sqrt{n} } \sum_{t=1}^n \mathbb{E} \bigg[ \boldsymbol{u}^{\prime} \boldsymbol{z}_t (\gamma) F \left( \boldsymbol{x}_t^{\prime} \boldsymbol{\beta}^{*}_{(1)} (\uptau) + s n^{- 1 / 2} \boldsymbol{u}^{\prime} \boldsymbol{z}_t (\gamma) \bigg| \boldsymbol{z}_t \right) \bigg] 
\end{align*} 
 where $( \uptau, \upgamma ) \in \mathcal{T} \times \Gamma$ and $s \in \mathbb{R}$. Therefore, we shall show that 
\begin{align}
\underset{ ( \uptau, \gamma ) \in \mathcal{T} \times \Gamma }{ \mathsf{sup} } \ \underset{ s \in [0,1] }{ \mathsf{sup} } \bigg| R_n ( \uptau, \gamma, s ) - R_n ( \uptau, \gamma, 0 )  \bigg| \overset{ p }{ \to } 0, 
\end{align}  
which leads to 
\begin{align}
\underset{ ( \uptau, \gamma ) \in \mathcal{T} \times \Gamma }{ \mathsf{sup} } \ \bigg| Z_n^{(22)} \left( \boldsymbol{u} , \uptau, \gamma \right) \bigg| \overset{ p }{ \to } 0, 
\end{align}  
 
\begin{remark}
Notice although the above framework proposed by \cite{galvao2014testing} corresponds to estimation and inference for a quantile threshold model under the assumption of a stationary pair $(Y_t, X_t)$, one can consider an extension to the case of regressors with possible unit roots and nonstationarity. In that case, the framework proposed by \cite{caner2001threshold} can be useful.  Consider for instance a two-parameter generalization of the usual functional limit theorem. 

Define stochastic integration with respect to the two-parameter process $W(s,u)$. Let
\begin{align*}
J(u) &= \int_0^1 X(s) dW(s,u)
\equiv \underset{ N \to \infty }{ \mathsf{plim} } \sum_{j=1}^N X \left( \frac{j-1}{N} \right) \left\{ W \left( \frac{j}{N}, u \right) - W \left( \frac{j-1}{N}, u \right)  \right\}
\end{align*} 
Notice that the integration is over the first argument of $W(s,u)$, holding the second argument constant. We will call the process $J(u)$ a stochastic integral process. 

\newpage

Then, it holds that
\begin{align}
\frac{1}{\sigma \sqrt{T} } \sum_{t=1}^n X_{nt-1} \mathbf{1}_{t-1} (u) e_t 
= 
\int_0^1 X_n(s) dW_n (s,u)
\Rightarrow 
J(u) = \int_0^1 X(s) dW (s,u)
\end{align}
on $u \in [0,1]$ as $n \to \infty$ and $J(u)$ is almost surely continuous. This result is a natural extension of the theory of weak convergence to stochastic integrals. 
\end{remark}

\subsection{Local Linear Quantile Regression}

The main idea of the local linear fit is to approximate the $q_{\tau}(z)$ in a neighbourhood of $x$ by a linear function such that 
\begin{align}
q_{\tau}(z) \approx   q_{\tau}(x) + \left( \dot{q}_{\tau}(x) \right)^{\prime} \left( z - x \right) \equiv \alpha_0 + \alpha_1^{\prime} ( z - x)   
\end{align}
where 
\begin{align}
q_{\tau}(x) = \left( \frac{ \partial q_{\tau} (x) }{ \partial x_1},..., \frac{ \partial q_{\tau} (x) }{ \partial x_p} \right)^{\prime}     
\end{align}
Therefore, we define an estimator by setting $q_{\tau} (x)$  at $x = \left( x_1,..., x_p \right)^{\prime} \in \mathbb{R}^p$. Thus, locally estimating $\big( q_{\tau} (x), \dot{q}_{\tau}(x) \big)$ is equivalent to estimating $\widehat{q}_{\tau} (x) \equiv \widehat{\alpha}_0$ and  $\widehat{ \dot{q} }_{\tau} (x) \equiv \widehat{\alpha}_1$. Then, 
\begin{align}
\begin{pmatrix}
\widehat{a}_0
\\
\widehat{a}_1
\end{pmatrix}
=
\underset{ ( \alpha_0, \alpha_1 ) }{ \mathsf{arg min} } \sum_{i=1}^n \rho_{\tau} \big( Y_i - \alpha_0 - \alpha_1^{\prime} ( X_i - x ) \big) . K_h \left( X_i - x \right),
\end{align}
where $K_h (x) = h_n^{ -p } K ( x / h_n )$, $K$ is a kernel function on $\mathbb{R}^p$, and $h_n > 0$ is the bandwidth. 

\begin{remark}
Notice that for statistical estimation and inference purposes it is required to employ a kernel estimator and kernel smoothing methods (see e.g., \cite{yu1998local} and \cite{linton2005estimating}). Further aspects of local linear quantile regression are investigated in the studies of \cite{hallin2009local} and \cite{ren2020local} among others. 
\end{remark}

\subsubsection{Asymptotic Theory: Bahadur representation}

Assume that $\left\{ \left( Y_t, X_t \right) \right\}$ is a stationary multivariate time series on a probability space $\left( \Omega, \mathcal{F}, P \right)$, where $X_t$ and $Y_t$ are random variables. Notice that $X_t$ may consist of both the lags of endogenous and/or exogenous variables. In particular, we are interested in the $\tau-$conditional quantile function of $Y_t$ given $X_t = x$,    
\begin{align}
\mathcal{Q}_{\tau} (x) = \underset{   }{ \mathsf{arg min} } \ \mathbb{E} \big[ \rho_{\tau} \left( Y_t | X_t = x \right) \big]    
\end{align}

\newpage

\begin{remark}
Notice that the conditional quantile regression was initially developed under \textit{i.i.d} samples for linear regression models in the econometrics literature.  We consider the weak conditions to ensure the Bahadur representation of the local linear estimators of $q_{\tau}(x)$ is first order differentiable, then its derivatives can be estimated reasonably well by the local linear fitting.
\end{remark}

\begin{theorem}
Assume that Assumptions A1, A2, A3 are satisfied for some $\alpha \geq 1$, and that the quantile function $q_{\tau}(x)$ is twice continuously differentiable at $x$. Then, it holds that     
\begin{align}
\sqrt{ n h_n^p }
\begin{pmatrix}
\big( \widehat{q}_{\tau}(x) - q_{\tau}(x) \big)
\\
h_n \big( \widehat{ \dot{q} }_{\tau}(x) - q_{\tau}(x) \big)
\end{pmatrix}
= 
\Phi_{\tau} (x) \frac{1}{ \sqrt{ n h_n^p} } \sum_{i=1}^n \psi_{\tau} ( Y_i^{*})
\begin{pmatrix}
1
\\
\frac{ ( X_i - x ) }{ h_n } 
\end{pmatrix}
+ o_p(1),
\end{align}
as $n \to \infty$, where $\psi_{\tau}(x) = \tau - \mathbf{1} \left\{ y < 0 \right\}$ and 
\begin{align*}
Y_i^{*} &:= Y_i^{*} ( \tau ) \equiv Y_i - q_{\tau} (x) - \big( \dot{q}_{\tau} (x) \big)^{\prime} ( X_i - x )
\\
\phi_{\tau} (x) &:= \left( f_{Y|X} ( q_{\tau}(x) | x ) f_X(x) \right)^{-1}
\end{align*}
\end{theorem}

\begin{remark}
Notice that if $q_{\tau}(x)$ has the first order derivatives which are Lipschitz continuous, then $q_{\tau}(x)$ and its derivatives can be estimated with optimal convergence rates of Stone (1980) as in the \textit{i.i.d} setting.    
\end{remark}

\subsubsection{Asymptotic Normality: Quantile Dependent Estimators}

Based on the powerful took of the weak Bahadur representation, we can establish the asymptotic distribution of the local linear quantile regression estimates under near-epoch dependence. The following lemmas are helpful for proving the asymptotic normality result. Suppose that 
\begin{align}
W_n := 
\begin{pmatrix}
w_{n0}
\\
w_{n1}
\end{pmatrix},
(W_n)_j := \left(  n h_n^p \right)^{-1} \sum_{i=1}^n \psi_{\tau} \left( Y_i^{*} \right) \left( \frac{X_i - x}{ h_n} \right)_j K \left(  \frac{ \left( X_i - x \right)}{h_n} \right), j = 0,..., p,
\end{align}
Furthermore, the usual Cramer-Wold device will be adopted. For all $c := \left( c_0, c_1^{\prime}   \right)^{\prime} \in \mathbb{R}^{ p + 1}$. Define, 
\begin{align}
A_n := \left( n h_n^p \right)^{1/2} c^{\prime} W_n =  \frac{1}{ \sqrt{  n h_n^p } }  \psi_{\tau} \left(  Y_i^{*} \right) K_c \left( \frac{ X_i - x}{ h_n } \right),  
\end{align}
Then, the expectation and asymptotic variance of the above expression is provided below.
\begin{lemma}
\begin{align}
\mathbb{E} \big[ \phi_{\tau}(x) A_n \big] 
&=  
\sqrt{ n h_n^p } \left[ \left( 1 + o(1) \right) 
 \begin{pmatrix}
B_0(x)
\\
B_1(x)
\end{pmatrix}
\right],   
\\
B_0(x) &= \frac{1}{2} f_X^{-1} (x) \mathsf{trace} \big[  \dot{q}_{\tau} (x) \int u u^{\prime} K(u) du    \big] h_n^2,
\\
B_1(x) &= \left(  B_{11}(x),..., B_{1p}(x) \right)^{\prime}, \ \  B_{1j}(x) = \frac{1}{2} f_X^{-1}(x) \mathsf{trace} \left[  \dot{q}_{\tau} (x) \int u u^{\prime} u_j K(u) du  \right] h_n^2,  
\end{align}
\end{lemma}

\newpage

Then, the asymptotic variance is given by 
\begin{align}
\underset{ n \to \infty }{ \mathsf{lim} }  \mathsf{Var} \big[ \phi_{\tau} (x) A_n \big] = c^{\prime} \Sigma c,    
\end{align}
\begin{align}
\Sigma := \phi^2_{\tau}(x) \tau(1 - \tau) f_X(x)
\begin{pmatrix}
\displaystyle \int K^2 (u) du & \displaystyle \int u^{\prime} K^2 (u) du   
\\
\displaystyle \int u K^2 (u) du  & \displaystyle \int u u^{\prime} K^2(u) du
\end{pmatrix}
\end{align}

\begin{lemma}
Suppose that Assumptions in Lemma 3.1 hold. Denote by $\sigma^2$ the asymptotic variance of $A_n$. Then, it holds that 
\begin{align}
\left( n h_n^p \right)^{1/2} \frac{ c^{\prime} \big[ W_n(x) - \mathbb{E} \left[ W_n (x) \big] \right] }{\sigma}  
\end{align}
is asymptotically standard normal as $n \to \infty$.
\end{lemma}

\subsection{Stein-Type Estimator in Quantile Regression Model}

\paragraph{Full and sub-model estimations}

We consider the following partitioned form as defined by \cite{yuzbacsi2017pretest}:
\begin{align}
y_i = \boldsymbol{\beta}_1^{\prime} \boldsymbol{x}_{1t} + \boldsymbol{\beta}_2^{\prime} \boldsymbol{x}_{2t} + \epsilon_{i}, \ \ i = 1,...,n,
\end{align}
where $p = p_1 + p_2$ and $\boldsymbol{\beta}_1$, $\boldsymbol{\beta}_2$ parameters are of order $p_1$ and $p_2$ respectively and $\boldsymbol{x}_i = \left( \boldsymbol{x}_{1i}^{\prime},  \boldsymbol{x}_{2i}^{\prime} \right)$ and $\epsilon_{i}$'s are the error terms. The conditional quantile function for the response variable $y_i$ is written in the form: 
\begin{align}
\mathcal{Q}_{ \uptau } ( y_i | \boldsymbol{x}_i ) = \boldsymbol{\beta}_1^{\prime} \boldsymbol{x}_{1t} + \boldsymbol{\beta}_2^{\prime} \boldsymbol{x}_{2t}, \ \ 0 < \uptau < 1,
\end{align}
The null hypothesis of interest can be formulated as below
\begin{align}
\mathbb{H}_0: \boldsymbol{\beta}_2 ( \uptau ) = \boldsymbol{0}_{p_2}
\end{align}

Considering the full model (FM) versus the sub-model (SM) the testing hypothesis is evaluated based on the following Wald statistic
\begin{align}
\label{Wald}
\mathcal{W}_n ( \uptau ) = n \omega^{-2} \left( \hat{ \boldsymbol{\beta}}^{FM}_2 ( \uptau ) - \boldsymbol{\beta}_2 ( \uptau ) \right)^{\prime}  \big[ \boldsymbol{D}^{22} \big]^{-1} \left( \hat{ \boldsymbol{\beta}}^{FM}_2 ( \uptau ) - \boldsymbol{\beta}_2 ( \uptau ) \right)
\end{align}
where $\boldsymbol{D}_{ij}$ for $i,j \in \left\{ 1,2 \right\}$ is the $(i,j)-$th partition of the $\boldsymbol{D}$ matrix and $\boldsymbol{D}^{ij}$ is the $(i,j)-$th partition of the $\boldsymbol{D}^{-1}$ matrix such that 
\begin{align}
\boldsymbol{D}^{22} = \left( \boldsymbol{D}_{22} - \boldsymbol{D}_{21} \boldsymbol{D}_{11}^{-1} \boldsymbol{D}_{12}    \right)^{-1}
\end{align}

\newpage

Denote with $\omega := \sqrt{ \uptau( 1 - \uptau)} \big/ f \big( F^{-1} ( \uptau ) \big)$, where the term $f \big( F^{-1} ( \uptau ) \big)$ is called the sparsity parameter or quantile density parameter and the sensitivity of the test statistic naturally depends on this parameter.  
\begin{proposition}
Under Assumption the the Wald statistic given by expression \eqref{Wald}, that is, $\mathbb{H}_0: \boldsymbol{\beta}_2 ( \uptau ) = \boldsymbol{0}_{p_2}$, has a $\chi^2$ limiting distribution with $p_2$ degrees of freedom. 
\end{proposition}

The full model (FM) quantile regression estimator is obtained by
\begin{align}
\widehat{ \boldsymbol{\beta}}^{FM} ( \uptau ) = \underset{ \beta \in \mathbb{R}^p }{ \mathsf{arg \ min}} \ \sum_{t=1}^n \rho_{\uptau} \left( y_t - \boldsymbol{\beta}^{\prime} \boldsymbol{x}_t \right)
\end{align}
The sub-model (SM) quantile regression estimator is given by
\begin{align}
\widehat{ \boldsymbol{\beta}}^{SM} ( \uptau ) = \left( \widehat{ \boldsymbol{\beta}}_1^{SM} ( \uptau ), \boldsymbol{0}_{p_2} \right)
\end{align}

The sub-model (SM) quantile regression estimator is obtained by
\begin{align}
\widehat{ \boldsymbol{\beta}}_1^{SM} ( \uptau ) = \underset{ \beta \in \mathbb{R}^{p_1} }{ \mathsf{arg \ min}} \ \sum_{t=1}^n \rho_{\uptau} \left( y_t - \boldsymbol{\beta}_1^{\prime} \boldsymbol{x}_{1t} \right)
\end{align}

\begin{theorem}
Under Assumption (A1) and (A2), the distribution of the quantile regression full model with $\textit{i.i.d}$ variables is given by
\begin{align}
\sqrt{n} \big(  \widehat{ \boldsymbol{\beta} }_1^{FM} ( \uptau ) -  \boldsymbol{\beta}  ( \uptau ) \big) \overset{ d }{ \to } \mathcal{N} \big( 0, \omega^2 \boldsymbol{D}^{-1} \big)
\end{align}
as $n \to \infty$. 
\end{theorem}

Let $\left\{ K_n \right\}$ be a sequence of local alternatives given by 
\begin{align}
K_n : \boldsymbol{\beta}_2 ( \uptau ) = \frac{ \boldsymbol{\kappa} }{ \sqrt{n} }
\end{align}
where $\boldsymbol{\kappa} = \left( \kappa_1,..., \kappa_{p_2} \right)^{\prime} \in \mathbb{R}^{p_2}$ is a fixed vector. When $\boldsymbol{\kappa} = \boldsymbol{0}_{ p_2 }$, then the null hypothesis is true. Furthermore, we consider the following proposition to establish the asymptotic properties of the estimators. 

\bigskip

\begin{proposition}
Define the following normalized statistical distances
\begin{align}
\boldsymbol{\vartheta}_1 
&= 
\sqrt{n} \big( \widehat{ \boldsymbol{\beta} }_1^{FM} ( \uptau ) - \boldsymbol{\beta}_1 ( \uptau ) \big)
\\
\boldsymbol{\vartheta}_2 
&= 
\sqrt{n} \big( \widehat{ \boldsymbol{\beta} }_1^{SM} ( \uptau ) - \boldsymbol{\beta}_1 ( \uptau ) \big)
\\
\boldsymbol{\vartheta}_3 
&= 
\sqrt{n} \big( \widehat{ \boldsymbol{\beta} }_1^{FM} ( \uptau ) - \widehat{ \boldsymbol{\beta} }_1^{SM} ( \uptau )  \big)
\end{align}
Under the regularity conditions, Theorem 3 and the local alternatives $\left\{ K_n \right\}$, as $n \to \infty$ we have the following joint distributions:
\end{proposition}

\newpage 

\begin{lemma}
Let $\boldsymbol{X}$ be a $q-$dimensional normal vector distributed as $\mathcal{N}\left( \boldsymbol{\mu}_x, \boldsymbol{\Sigma}_q \right)$, then for a measurable function of $\varphi$, we have that
\begin{align*}
\mathbb{E} \big[ \boldsymbol{X} \varphi \left( \boldsymbol{X}^{\top} \boldsymbol{X} \right) \big] &= \boldsymbol{\mu}_x \mathbb{E} \big[ \varphi \chi^2_{q+2} \left( \Delta \right) \big]
\\
\mathbb{E} \big[ \boldsymbol{X} \boldsymbol{X}^{\top} \varphi \left( \boldsymbol{X}^{\top} \boldsymbol{X} \right) \big] &= \boldsymbol{\Sigma}_q \mathbb{E} \big[ \varphi \chi^2_{q+2} \left( \Delta \right) \big] + \boldsymbol{\mu}_x \boldsymbol{\mu}_x^{\top} \mathbb{E} \big[ \varphi \chi^2_{q+4} \left( \Delta \right) \big]
\end{align*}
\end{lemma}
where $\chi^2_{\nu} \left( \Delta \right)$ is a non-central chi-square distribution with $\nu$ degrees of freedom and non-centrality parameter  $\Delta$. 
\begin{align}
\begin{pmatrix}
\boldsymbol{\vartheta}_1
\\
\boldsymbol{\vartheta}_3
\end{pmatrix}
\sim 
\mathcal{N}
\left[ 
\begin{pmatrix}
\boldsymbol{0}_{ p_1 }
\\
- \boldsymbol{\delta}
\end{pmatrix},
\begin{pmatrix}
\omega^2 \boldsymbol{D}_{11.2}^{-1} & \boldsymbol{\Sigma}_{12}
\\
\boldsymbol{\Sigma}_{21} & \boldsymbol{\Phi}
\end{pmatrix}
\right],
\\
\nonumber
\\
\begin{pmatrix}
\boldsymbol{\vartheta}_3
\\
\boldsymbol{\vartheta}_2
\end{pmatrix}
\sim 
\mathcal{N}
\left[ 
\begin{pmatrix}
- \boldsymbol{\delta}
\\
\boldsymbol{\delta}
\end{pmatrix},
\begin{pmatrix}
\boldsymbol{\Phi} & \boldsymbol{\Sigma}^{\star}
\\
\boldsymbol{\Sigma}^{\star} & \omega^2 \boldsymbol{D}_{11}^{-1}
\end{pmatrix}
\right]
\end{align}

where $\boldsymbol{\delta} = \boldsymbol{D}_{11}^{-1} \boldsymbol{D}_{12} \boldsymbol{\kappa}$, $\boldsymbol{\Phi} = \omega^2 \boldsymbol{D}_{11}^{-1} \boldsymbol{D}_{12}  \boldsymbol{D}_{22.1}^{-1} \boldsymbol{D}_{21} \boldsymbol{D}_{11}^{-1}$, $\boldsymbol{\Sigma}_{12} = - \omega^2 \boldsymbol{D}_{12} \boldsymbol{D}_{21} \boldsymbol{D}_{11}^{-1}$ and $\boldsymbol{\Sigma}^{\star} = \boldsymbol{\Sigma}_{21} + \omega^2 \boldsymbol{D}_{11.2}^{-1}$.

\begin{proof}
Consider the asymptotic bias as below
\begin{align}
\mathsf{Bias} \left(  \widehat{ \boldsymbol{\beta} }_1^{FM} ( \uptau )   \right) = \mathbb{E} \left\{ \underset{ n \to \infty }{ \mathsf{lim} } \sqrt{n} \left( \widehat{ \boldsymbol{\beta} }_1^{FM} ( \uptau ) -    \boldsymbol{\beta}_1 \right) \right\} = \boldsymbol{0}_{p_1}
\end{align}
Similarly, the asymptotic covariance is defined as
\begin{align*}
\mathsf{avar} \left(  \widehat{ \boldsymbol{\beta} }_1^{FM} ( \uptau )   \right) 
&= 
\mathbb{E} \left\{ \underset{ n \to \infty }{ \mathsf{lim} } n \left( \widehat{ \boldsymbol{\beta} }_1^{FM} ( \uptau ) -    \boldsymbol{\beta}_1 \right) \left( \widehat{ \boldsymbol{\beta} }_1^{FM} ( \uptau ) - \boldsymbol{\beta}_1 \right)^{\prime} \right\}
\\
&=  
\mathsf{Cov} \left( \boldsymbol{\vartheta}_1, \boldsymbol{\vartheta}_1^{\prime} \right) + \mathbb{E} \left( \boldsymbol{\vartheta}_1  \right) \mathbb{E} \left( \boldsymbol{\vartheta}_1^{\prime}  \right)
\\
&= 
 \mathsf{Cov} \left( \boldsymbol{\vartheta}_1, \boldsymbol{\vartheta}_1^{\prime} \right)
\\
&=  \omega^2 \boldsymbol{D}_{11.2}^{-1}
\end{align*}
where $\boldsymbol{D}_{11.2}^{-1} = \left( \boldsymbol{D}_{11} - \boldsymbol{D}_{12} \boldsymbol{D}_{22}^{-1} \boldsymbol{D}_{21}   \right)^{-1}$. Therefore, $\boldsymbol{\vartheta}_1 \sim \mathcal{N} \left( \boldsymbol{0}_{p_1}, \omega^2 \boldsymbol{D}_{11.2}^{-1}  \right)$. Notice that in this setting, $\Delta = \boldsymbol{\kappa}^{\prime} \left( \omega^2 \boldsymbol{D}_{22.1}^{-1} \right)^{-1} \boldsymbol{\kappa}$, $d = p_2 - 2$ and $\mathbb{H}_{\nu} ( x, \Delta )$ is the cumulative distribution function of the non-central chi-squared distribution with non-centrality parameter $\Delta$ and $\nu$ degree of freedom, and 
\begin{align}
\mathbb{E} \left( \chi^{-2j}_{\nu} \left( \Delta \right) \right) = \int_0^{\infty} d \mathbb{H}_{\nu} \left( x, \Delta \right). 
\end{align}

\end{proof}
Further details on the use of Stein-type estimators in statistical applications can be found in the studies of \cite{nkurunziza2012shrinkage} and  \cite{chen2016class}. An application of restricted estimators for structural break testing is proposed by \cite{nkurunziza2021inference}.

\newpage

\section{Quantile Time Series Regression Models:  Nonstationary Case}

The econometric framework for estimation and inference for quantile autoregression is presented by \cite{koenker2004unit}, \cite{koenker2006quantile} and \cite{galvao2009unit} (see, also \cite{koenker2017handbook}). 

Consider the following quantile autoregressive model 
\begin{align}
\widehat{\theta} (\tau) = \underset{ \theta  }{ \mathsf{arg max} } \sum_{t=1}^T \rho_{\tau} \big( y_t - \theta (\tau)^{\prime} z_t \big),
\end{align}
where $\rho_{\tau} ( \mathsf{u} ) = \mathsf{u} \big( \tau - \mathbf{1} \left\{ \mathsf{u} < 0   \right\} \big)$ is the check function. 

The following unit root t-ratio test statistic holds
\begin{align}
t_f ( \tau ) \Rightarrow t( \tau ) = \frac{1}{ \sqrt{\tau (1 - \tau)} } \left( \int_0^1 \underline{B}_{\mu}^2 \right) \left( \int_0^1 \underline{B}_{\mu} d B_{\psi}^{\tau} \right).
\end{align}
Notice that the asymptotic distribution of the above t-ratio test statistic is nonstandard since $B_{\psi}$ and $B_{\mu}$ are correlated. However, it can be decomposed into a linear combination of two independent parts. Specifically the following decomposition holds
\begin{align}
\int_0^1 \underline{B}_{\mu} d B_{\psi}^{\tau} = \int_0^1 \underline{B}_{\mu} d B_{\psi . \mu }^{\tau} + \lambda_{\mu \psi} (\tau) \int_0^1 \underline{ \boldsymbol{B} }_{\mu} d B_{\psi}^{\tau} 
\end{align}

\subsection{Quantile regression for nonstationary time series}

Consider the ADF regression model. Denote the $\sigma-$field generated by $\left\{ u_s, s \leq t \right\}$ by $\mathcal{F}_t$, then conditional on $\mathcal{F}_{t-1}$, the $\uptau-$th conditional quantile of $Y_t$ is given by 
\begin{align}
\mathcal{Q}_{ Y_t } \left( \uptau | \mathcal{F}_{t-1} \right) = \mathcal{Q}_{u}( \uptau ) + \alpha_1 Y_{t-1} + \sum_{j=1}^p \alpha_{j+1} \Delta Y_{t-j}. 
\end{align} 
Moreover, denote with $\alpha_0 ( \uptau ) = \mathcal{Q}_{u}( \uptau )$, $\alpha_j ( \uptau )$, $j = 1,...,p, p = q+1$, and define with 
\begin{align}
\alpha ( \uptau ) = \big( \alpha_0 ( \tau ), \alpha_1,..., \alpha_{q+1}     \big), \ \ X_t = \big( 1, Y_{t-1}, \Delta Y_{t-1}, ..., \Delta Y_{t-q}    \big)^{\prime},  
\end{align}
we have that $\mathcal{Q}_{ Y_t } \left( \uptau | \mathcal{F}_{t-1} \right) = X_t^{\prime} \alpha ( \uptau )$. Then, the unit root quantile autoregressive model can be estimated
\begin{align}
\underset{ \alpha }{ \mathsf{min} } \sum_{t=1}^n \rho_{\uptau} \big( Y_t - X_t^{\prime} \alpha \big). 
\end{align}
Denote with $w_t = \Delta Y_t$, $u_{t \uptau} = Y_t - X_t^{\prime} \alpha ( \uptau )$, under the unit root hypothesis and other regularity assumptions 
\begin{align}
n^{ - 1 /2} \sum_{ t = 1 }^{ \floor{nr} } \big( w_t, \psi_{\uptau} \left(  u_{t \uptau } \right) \big)^{\prime} \Rightarrow \big( B_w(r), B_{\psi}^{\uptau}(r) \big)^{\prime} = BM \big( 0, \boldsymbol{\Sigma}^{\star} ( \uptau ) \big), 
\end{align}

\newpage

where 
\begin{align}
\boldsymbol{\Sigma}^{\star} ( \uptau )  =
\begin{bmatrix}
\sigma_w^2 & \sigma_{w \psi } (\uptau) 
\\
\sigma_{w \psi } (\uptau) & \sigma_{\psi}^2 (\uptau)  
\end{bmatrix},
\end{align}
is the long-run covariance matrix of the bivariate Brownian motion and can be written as $\boldsymbol{\Sigma}_0 ( \uptau ) + \boldsymbol{\Sigma}_1 ( \uptau ) + \boldsymbol{\Sigma}_1 ( \uptau )^{\prime}$, where
\begin{align}
\boldsymbol{\Sigma}_0 ( \uptau )  &= \mathbb{E} \left[ \big( w_t, \psi_{\uptau} ( u_{t \uptau} ) \big)^{\prime} \big( w_t, \psi_{\uptau} ( u_{t \uptau} ) \big) \right]
\\
\boldsymbol{\Sigma}_1 ( \uptau ) &= \sum_{s=2}^{\infty} \mathbb{E} \big[ \big( w_t, \psi_{\uptau} ( u_{t \uptau} ) \big)^{\prime} \big( w_s, \psi_{\uptau} ( u_{s \uptau} ) \big) \big].
\end{align}
In addition we have that, 
\begin{align}
\frac{1}{n} \sum_{t=1}^n Y_{t-1} \psi_{\uptau} \left( u_{t \uptau} \right) \Rightarrow \int_0^1 B_w d B_{\psi}^{\uptau}. 
\end{align}
Notice that the random function $\displaystyle n^{- 1 / 2} \sum_{t=1}^{ \floor{nr} } \psi_{\uptau} \left( u_{t \uptau} \right)$ converges to a two-parameter process such that $B_{ \psi }^{\uptau} (r) = B_{\psi} ( \uptau, r )$, which is partially a Brownian bridge in the sense that for fixed $r$, $B_{ \psi }^{\uptau} (r) = B_{ \psi } ( \uptau, r)$ is a rescaled Brownian bridge, while for each $\uptau$, $\displaystyle n^{- 1 / 2} \sum_{t=1}^{ \floor{nr} } \psi_{\uptau} \left( u_{t \uptau} \right)$ converges weakly to a Brownian motion with variance equal to $\tau (1 - \uptau)$.  Moreover, for each fixed pair $( \uptau_0, r )$ we have that
\begin{align}
B_{ \psi }^{ \uptau_0 }(r) = B_{\psi} ( \uptau_0 , r ) \sim \mathcal{N} \big( 0,   \uptau_0 ( 1 - \uptau_0 )r \big)
\end{align}
Let $\widehat{\alpha}( \uptau ) = \left( \widehat{\alpha}_0( \uptau ), \widehat{\alpha}_1( \uptau )    ,..., \widehat{\alpha}_p( \uptau ) \right)$ and $D_n = \mathsf{diag} \left( \sqrt{n}, n , \sqrt{n}, ... , \sqrt{n} \right)$.

\medskip

\begin{theorem}
Let $y_t$ be determined by the ADF regression model, under the unit root assumption $\alpha_1 = 1$ and other regularity conditions, 
\begin{align}
D_n \big( \widehat{\alpha}( \uptau ) - \alpha( \uptau) \big) \Rightarrow \frac{1}{ f \big( F^{-1} ( \uptau) \big)}
\begin{bmatrix}
\displaystyle  \int_0^1 \tilde{\boldsymbol{B}}_w (r) \tilde{\boldsymbol{B}}_w^{\top} (r)  & \boldsymbol{0}_{2 \times q}  
\\
\boldsymbol{0}_{2 \times q}  &  \boldsymbol{\Omega}_{\Phi}
\end{bmatrix}
\end{align} 
where $\tilde{\boldsymbol{B}}_w (r) = [ 1, \boldsymbol{B}_w (r) ]^{\top}$, $\boldsymbol{\Phi} = \left[ \boldsymbol{\Phi}_1, ...,   \boldsymbol{\Phi}_q \right]^{\top}$ is a $q-$dimensional normal variate with covariance matrix $\uptau ( 1 - \uptau) \boldsymbol{\Omega}_{\Phi}$, where 
\begin{align}
\boldsymbol{\Omega}_{\Phi} = 
\end{align}
and $\boldsymbol{\Phi}$ is independent with $\int_0^1 \tilde{\boldsymbol{B}}_w (r) d \tilde{\boldsymbol{B}}_{\psi} (\uptau)$. 
\end{theorem}

\newpage

Now define with $B^{\mu}_w (r) = B_w (r) - \int_0^1 B_w$ is a demeaned Brownian motion. Therefore, we can derive for any fixed $\uptau$, the test statistic $t_n( \uptau )$, that is, the quantile regression counterpart of the well-known ADF t-ratio for a unit root. Thus, it can be shown that the limiting distribution of $t_n( \uptau )$ is nonstandard and depends on nuisance parameters $\left( \sigma_w^2, \sigma_{w \psi } ( \uptau) \right)$ as $B_w$ and $B_{ \psi }^{\uptau}$ are correlated Brownian motions.  

Notice that the limiting distribution of the t-ratio $t_n( \uptau )$ can be decomposed as a linear combination of two (independent) distributions, with weights determined by a long-run (zero frequency) correlation coefficient that can be consistently estimated. Following \cite{phillips1990statistical} we have that
\begin{align}
\int_0^1 \boldsymbol{B}^{\mu}_w d \boldsymbol{B}_{\psi}^{ \uptau } 
=
\int _0^1 \boldsymbol{B}^{\mu}_w d \boldsymbol{B}_{\psi.w}^{ \uptau } + \lambda_{ \omega \psi} ( \uptau ) \int_0^1 \boldsymbol{B}^{\mu}_w d \boldsymbol{B}_{w}, 
\end{align} 
where $\lambda_{ \omega \psi} ( \uptau ) = \sigma_{ w \psi } ( \uptau )/ \sigma_w^2$ and $\boldsymbol{B}_{\psi.w}^{ \uptau }$ is a Brownian motion with variance $\sigma^2_{ \psi . w} (\uptau) = \sigma^2_{\psi} ( \uptau ) - \sigma^2_{ w \psi } ( \uptau ) / \sigma_w^2$ and is independent of $\boldsymbol{B}^{\mu}_w$.

Therefore, the limiting distribution of the t-ratio $t_n ( \uptau )$ can be decomposed as below
\begin{align}
\frac{1}{ \sqrt{\uptau ( 1- \uptau ) } } \frac{ \displaystyle \int _0^1 \boldsymbol{B}^{\mu}_w d \boldsymbol{B}_{\psi.w}^{ \uptau } }{ \displaystyle \left( \int_0^1 \boldsymbol{B}^{\mu 2}_w \right)^{1 / 2} } + \frac{ \lambda_{ \omega \psi} ( \uptau ) }{ \sqrt{\uptau ( 1- \uptau ) } } \frac{ \displaystyle  \int_0^1 \boldsymbol{B}^{\mu}_w d \boldsymbol{B}_{w} }{ \displaystyle  \left( \int_0^1 \boldsymbol{B}^{\mu 2}_w \right)^{1 / 2} }.
\end{align}

For notation convenience we can also rewrite the Brownian motions $\boldsymbol{B}_w(r)$ and $\boldsymbol{B}^{\uptau}_{ \psi . w }(r)$:
\begin{align*}
\boldsymbol{B}_w(r) &=  \sigma_w \boldsymbol{W}_1 (r)
\\
\boldsymbol{B}^{\uptau}_{ \psi . w }(r) &= \sigma_{ \psi . w }( \uptau ) \boldsymbol{W}_2 (r),
\end{align*} 
and for the corresponding demeaned Brownian motions as below
\begin{align*}
\boldsymbol{B}^{\mu}_w(r) &=  \sigma_w \boldsymbol{W}_1 (r)
\\
\boldsymbol{W}^{\mu}_{1}(r) &= \boldsymbol{W}_{1}(r) - \int_0^1 \boldsymbol{W}_{1}(s)  ds,
\end{align*}
where $\boldsymbol{W}_{1}(s)$ and $\boldsymbol{W}_{2}(s)$ are standard Brownian motions and are independent stochastic processes. Note also that $\sigma^2_{\psi} ( \uptau ) = \uptau ( 1 - \uptau)$, and the limiting distribution of $t_n( \uptau )$ can be written as below
\begin{align}
\delta \left( \int_0^1 \boldsymbol{W}_1^2 \right)^{- 1 / 2} \int_0^1 \boldsymbol{W}_1 d \boldsymbol{W}_1 + \sqrt{1 - \delta^2 } \mathcal{N} (0,1), 
\end{align}
where 
\begin{align}
\delta = \delta( \uptau ) = \frac{ \sigma_{ w \psi }( \uptau )  }{   \sigma_{w} \sigma_{ \psi }( \uptau ) } = \frac{ \sigma_{ w \psi }( \uptau ) }{ \sigma_{w} \sqrt{\uptau (1 - \uptau)} }. 
\end{align}

\newpage 

\begin{example}
Consider the partitioned linear model studied by \cite{hasan1997robust} given below
\begin{align}
y = \boldsymbol{X} \boldsymbol{\beta} + \boldsymbol{Z} \boldsymbol{\gamma} + \boldsymbol{u} 
\end{align}
Suppose we are interested in testing the hypothesis $H_0: \gamma = 0$, with $\beta$ unspecified versus the (Pitman) local alternatives, $H_n: \gamma = \gamma_0 / \sqrt{n}$. 

Consider the Augmented Diceky-Fuller model 
\begin{align}
\Delta y_t  = \beta_0 + \gamma y_{t-1} + \sum_{j=1}^p \beta_j \Delta y_{t-j} + e_t, 
\end{align}
where our aim is to test the hypothesis, $\mathbb{H}_0: \gamma = 0$, versus the local alternatives, $\mathbb{H}_T: \gamma_T = \gamma_0 / T$. Furthermore, partitioning the design matrix we write $\boldsymbol{x}_t = \left( 1, \Delta y_{t-1}, ...,  \Delta y_{t-p} \right)^{\prime}$, $z_t = y_{t-1}$.   

Consider the bivariate process $\left\{ \left( u_t, v_t \right)  \right\}$. Then, we denote with
\begin{align*}
\Sigma 
= 
\underset{ T \to \infty }{ \text{lim} } \text{var} \left( T^{- 1 / 2} \sum_{t=1}^T \left( u_t, v_t \right)^{\prime} \right)
&= 
\underset{ T \to \infty }{ \text{lim} } T^{-1} \mathbb{E} \left[ \left( \sum_{t=1}^T \left( u_t, v_t \right) \right) \left( \sum_{t=1}^T \left( u_t, v_t \right)^{\prime} \right) \right]
\\
&=
\Sigma_0 + \Sigma_1  + \Sigma_1^{\prime}  
\end{align*} 
where 
\begin{align}
\Sigma_0 
&= 
\underset{ T \to \infty }{ \text{lim} } T^{-1}  \sum_{t=1}^T \mathbb{E} \left[ \left( u_t, v_t \right)^{\prime} \left( u_t, v_t \right) \right]
\\
\Sigma_1 
&= 
\underset{ T \to \infty }{ \text{lim} } T^{-1} \sum_{t=1}^T \sum_{s=1}^{t-1} \mathbb{E} \left[ \left( u_s, v_s \right)^{\prime} \left( u_t, v_t \right) \right]
\end{align}
and denote the lower triangular form of the Cholesky decomposition of $\Sigma$ as below 
\begin{align}
\Sigma^{1 / 2} = \Sigma_{11}^{ - 1 / 2}
\begin{bmatrix}
\Sigma_{11} & 0 
\\
\Sigma_{12} & \Delta^{1 / 2} 
\end{bmatrix}
\end{align}
where $\Delta = \left| \Sigma \right|$. Thus for the development of the asymptotic theory consider the following standardized bivariate process
\begin{align}
W_T(s) = \Sigma^{- 1 / 2} T^{- 1 / 2} \sum_{t=1}^{ \floor{Ts} } \left( u_t , v_t \right)^{\prime}
\end{align} 
Then, the authors prove that the modified ADF test-statistic (rank-type statistic) under the null hypothesis and a sequence of local alternatives is asymptotically normal. Moreover, the authors comment on the well-known fact by now that the OLS theory under the alternative hypothesis is non-standard, that is, non-Gaussian since it involves a limit distribution which depends on Ornstein-Uhlenbeck processes.

\end{example}

\newpage

\subsection{Nonstationary Nonlinear Quantile Regression}

Consider the scalar-valued random variable $y_t$ that follows the nonlinear time series model (see, \cite{uematsu2019nonstationary}) 
\begin{align}
y_t = \alpha_0 + g( x_t, \beta_0 ) + u_t, \ \ \ t \in \left\{ 1,..., n \right\},
\end{align}
where $g : \mathbb{R} \times \mathbb{R}^{\ell} \to \mathbb{R}$ is a known regression function and the error term $u_t$ is a zero-mean stationary process. We denote with $g ( x_t , \beta ) \equiv g_t (\beta)$. Furthermore, the regressor $x_t$ follows an $I(1)$ time series process as defined below
\begin{align}
x_t = x_{t-1} + v_t, 
\end{align}  
where $x_0 = 0$ and the innovation sequence $v_t$ is assumed to be stationary with mean zero. 

Then, the $\uptau-$th quantile NQR estimator $\hat{\theta}_n (\uptau) = \left( \hat{\alpha}_n ( \uptau ), \hat{\beta}_n(\uptau) \right)$ is obtained by the following minimization problem:
\begin{align}
\hat{\theta}_n (\uptau) = \underset{ \theta \in \Theta }{ \mathsf{arg \ min} } \ \sum_{t=1}^n \rho_{\uptau} \bigg( y_t - \alpha - g_t ( \beta ) \bigg),
\end{align}
where $\rho_{\uptau} ( \mathsf{u} ) = \mathsf{u} \big( \uptau - \mathbf{1} \left( \mathsf{u} < 0 \right) \big)$ is the check function and $\psi_{\uptau} ( \mathsf{u} ) = \uptau - \mathbf{1} \left( \mathsf{u} < 0 \right)$. Next, we impose some parametric assumptions regarding the distribution of the error term $u_t$. 

Let $\alpha_{0 \uptau} := \alpha_0 + F^{-1} ( \uptau )$ and define the new parameter vector $\theta_{0 \uptau} = \big( \alpha_{0 \uptau}, \beta_{0 \uptau}^{\prime} \big)^{\prime}$. We then, rewrite the error term as below
\begin{align}
u_{t \uptau } = y_t - \alpha_{0 \uptau} - g_t ( \beta_0 ) := u_t - F^{-1} ( \uptau ). 
\end{align}  

Notice that $\mathbb{E} \psi_{\uptau} ( u_{t \uptau} ) = 0$ and $\mathcal{Q}_{ u_{t \uptau} } ( \uptau ) = 0$, where $\mathcal{Q}_{ u_{t \uptau} } ( \uptau ) = 0$ is the $\uptau-$th quantile of $u_{t \uptau}$. Therefore, in view of $\theta_{0 \uptau}$ and $u_{t \uptau}$, the argument of the check function can be written as below
\begin{align}
y_t - \alpha - g_t ( \beta ) = u_{t \uptau} - ( \alpha - \alpha_{0 \uptau} ) - \bigg[ g( x_t, \beta ) -  g( x_t, \beta_0 )  \bigg]
\end{align}
Then, for the error terms $u_{t \uptau}$ and $v_t$, we construct two partial sum processes as below
\begin{align}
U_n^{\psi} ( \uptau, r ) := n^{- 1 / 2} \sum_{t=1}^{\floor{nr} }
\psi_{\uptau} \left( u_{t \uptau} \right) \ \ \ \text{and} \ \ \ V_n(r) = n^{- 1 / 2} \sum_{t=1}^{\floor{nr} } v_{t+1}. 
\end{align}

\newpage

\begin{assumption}
We assume the following conditions: Let $\left\{ x_t \right\}$ is adapted to the filtration $\mathcal{F}_{t-1}$ and for all $r \in [0,1]$, the vector $\big( U_n^{\psi} ( \uptau, r ), V_n(r) \big)$ converges weakly to a two-dimensional vector Brownian motion $\big( U_n^{\psi} ( \uptau, r ), V_n(r) \big)$  with a covariance matrix $
r \Omega ( \uptau )$
\end{assumption}

\medskip

Following the framework of \cite{uematsu2019nonstationary}, we employ the proof of Lemma 7.5 as below. In particular, we denote with 
\begin{align}
W_n (\lambda ) = \sum_{t=1}^n w_t (\lambda), 
\end{align}
where 
\begin{align}
w_t (\lambda ) = \big( \gamma_t (\lambda) - u_t(\uptau) \big)
\times \mathbf{1} \big\{ \gamma_t (\lambda) > u_t (\uptau) > 0 \big\}. 
\end{align}
Notice that to avoid technical problems in taking conditional expectations, we consider truncation of $\upgamma_t(\lambda)$ at some finite number $m > 0$ and denote with $W_{nm} (\lambda) = \sum_{t=1}^n w_{tm} ( \lambda )$, where 
\begin{align}
w_{tm} ( \lambda ) 
&= 
\big( \gamma_t( \lambda ) - u_t( \lambda ) \big) \mathbf{1} \big\{  \gamma_t( \lambda ) > u_t(\uptau) > 0 \big\} I_{tm}(\lambda),
\\
I_{tm}(\lambda) 
&=
\mathbf{1} \big( 0 < \upgamma_t (\lambda) \leq m \big). 
\end{align}

Further define with 
\begin{align}
\bar{w}_{tm} (\lambda) = \mathbb{E}_{t-1} \big[ w_{tm} (\lambda)  \big]
\end{align}
and
\begin{align}
\bar{W}_{nm}(\lambda) = \sum_{t=1}^n \bar{w}_{tm} (\lambda).
\end{align}
Therefore, using the notation above, we will derive the probability limit of $W_n(\lambda)$ through the following steps. First, we compute the limiting variable of 
\begin{align*}
\bar{W}_{nm}(\lambda)= \sum_{t=1}^n \bar{w}_{tm} (\lambda)
\end{align*}
by considering the simultaneous limits $n \to \infty$ and $m \to \infty$. Next, we check the asymptotic equivalence of $W_{nm} (\lambda)$ and $\bar{W}_{nm} (\lambda)$. Finally, we verify that the asymptotic equivalence of $W_n(\lambda)$ and $W_{nm}(\lambda)$, that is, the effect of truncation by $I_{tm}(\lambda)$ is asymptotically negligible.

\begin{remark}
Notice that the objective function for the setting of the nonstationary quantile predictive regression model, becomes globally convex in the parameter. Hence, the method based on the convexity lemma by \cite{pollard1991asymptotics} is applicable. Therefore, our proofs for the asymptotic theory of test of parameter restrictions is based on the asymptotic theory framework proposed by \cite{xiao2009quantile}.
\end{remark}

\newpage

\subsection{Quantile Predictive Regression Model}

The predictive regression model has been widely used for investigating the predictability of asset returns. Applications include estimation and inference methodologies for predictive regression model based on a parametric (see, \cite{kostakis2015Robust}) or nonparametric (see, \cite{kasparis2015nonparametric}) conditional mean functional form. The current literature mainly covers linear conditional mean autoregressive and predictive regression models as well as dynamic conditional mean predictive regression models. When one is interested for testing the presence of quantile predictability, related studies have considered the implementation of the model based on the conditional quantile functional form. Related studies include \cite{galvao2009unit}, \cite{maynard2023inference}, \cite{lee2016predictive}, \cite{fan2019predictive}, \cite{katsouris2022asymptotic},  \cite{cai2022new} and \cite{liu2023unified}.  The inclusion of model intercepts in both the predictive regression model as well as the autoregressive model either based on a conditional mean or on a conditional quantile functional forms imposes additional challenges when developing inference techniques for the $\beta$ parameter based on the usual likelihood ratio test or Wald-type statistics;  especially when the autoregressive coefficient is nearly-integrated (see, \cite{liu2019asymptotic},  \cite{wang2022testing}, \cite{yang2023unified}).

Various studies in the literature employ the likelihood-based approach to investigate the properties and asymptotic theory of corresponding estimators and test statistics for nonstationary autoregressive processes and predictive regression models. These likelihood approaches include the weighted least squares, the empirical likelihood methodology as well as the profile likelihood approach.  In particular, \cite{zhu2014predictive},  \cite{liu2019empirical, liu2019unified} and \cite{yang2021unified} consider an application of the empirical likelihood approach which induces uniform inference across the different degrees of persistence and confidence regions for the coefficients of the linear predictive regression. Moreover, \cite{chen2013uniform} propose a uniform inference approach (see, also \cite{chen2009bias}\footnote{Related studies include \cite{chen2009restricted}, \cite{chen2012restricted} as well as \cite{zhao2000restricted}.}) which is valid regardless of the persistence of regressors using the weighted least squares approximate restricted likelihood (WLSRL). 

Furthermore, \cite{liu2019unified} propose predictability tests that correspond to  testing the null hypothesis of no statistical significance for the model coefficients of a predictive regression model with  additional covariates the difference of the predicting variable. Specifically, these predictability tests are based on the unified empirical likelihood methodology that rejects the null hypothesis $H_0: \beta_1 = \beta_2 = 0$ or testing the hypothesis $H_0: \beta_1 = 0$ and $H_0: \beta_2 = 0$. Moreover, under the presence of known nonzero model intercept in the predictive regression model using weighted score equations and the sample-splitting methodology it can be proved that nuisance parameter-free inference holds regardless of the persistence properties of regressors. However, implementing the sample-split approach into a multivariate nonstationary vector of regressors remains a computational challenging task as this would require to split the data into blocks equal to the dimension of the vector-valued regressor and then use each block to construct score equations with respect to each variable in the vector. A similar difficulty exists when one considers the IVX filtation to a multivariate or even high-dimensional setting (see, \cite{XuGuo2022}). 

\newpage

Both methodologies fundamentally are trying to solve a similar problem as the presence of model intercept in the predictive regression model causes a nonstandard limiting distribution results especially when $X_t$ is nonstationary. To overcome this problem these two different methods consider a different approach. On the one hand the empirical likelihood method takes the first difference of the paired data such that $\left\{ Y_{t+1} - Y_t \right\}_{t=1}^n$ and $\left\{ X_{t+1} - X_t \right\}_{t=1}^n$ but a new issue arises as the corresponding error sequence $\left\{ U_{t+1} - U_t \right\}_{t=1}^n$ violates the independence assumption. Therefore, to fix this issue a split sample methodology is applied which means that the data is splited into two parts, the first difference transformation is applied based on a larger lag (i.e., the value of $m = \floor{n/2}$) which ensures $m-$dependence holds (i.e., approximate or asymptotic independence as $m \to \infty$), and lastly the empirical likelihood estimation is applied. The asymptotic theory for the empirical likelihood confidence interval for $\beta_0$ is then proved to follow standard distribution results (i.e., converges to a $\chi^2-$distribution).

On the other hand, the IVX filtration proposed by \cite{magdalinos2009econometric} operates differently but still with the aim to solve a similar issue, to avoid having nonstandard asymptotic results under the presence of nonstationary regressors. More specifically, the IVX filtration is an endogenous instrumentation procedure which is used as an alternative to the OLS estimator for the unknown model coefficient $\beta_0$. In particular, the IVX instruments are constructed with the main consideration of obtaining less persistence regressors than the original series $\left\{ X_t \right\}$ without having to consider a first-difference transformation based on a large value of the lag as well as the application of sample splitting to ensure that the distribution of the innovations of the predictive regression have independent realizations. This methodology implies that asymptotic theory for the unknown model coefficient follows a mixed gaussian distribution and the corresponding Wald test follow a $\chi^2-$distribution.  

In summary, both methodologies have proved to produce uniform inference for unit root moderate deviations in the case of a local-to-unity autoregressive specification. The empirical likelihood plus weighted least squares induces uniform inference regadless of persistence properties and the presence of an intercept in the autoregressive specification of regressors. Similarly the IVX instrumentation also induces uniform inference regardless of nonstationary properties (as in \cite{kostakis2015Robust}) while more recently is has been proved by \cite{Magdalinos2022uniform} that is indeed also robust to the presence of an intercept as well as more general nonstationarity properties. Similar aspects hold  for conditional quantile functional forms. In particular, \cite{lee2016predictive} develops a framework for the implementation of the IVX instrumentation in quantile predictive regression models which is found to have standard asymptotic theory while more recently \cite{liu2023unified}  proved that their approach is indeed robust to both the presence of intercept in the autoregressive specification and nonstationarity of the $\left\{ X_t \right\}$ process. Currently, \cite{katsouris2022asymptotic}  investigates the extension of the framework proposed by \cite{Magdalinos2022uniform}, to quantile autoregressions and quantile predictive regression models. Lastly, another aspect is the development of a framework for testing the presence of structural breaks in quantile predictive regression models, which is proposed by \cite{katsouris2023structural} (see, also \cite{qu2008testing}).

\newpage

\subsubsection{Local Unit Root Framework}

We first focus on the quantile predictive regression model which can be extended within a dynamic setting. Therefore, within the framework proposed by  \cite{katsouris2021optimal, katsouris2023statistical}, we consider the estimation procedure of VaR of a financial institution which is given by expression  under the additional assumption that the predictors are generated via a local-to-unity parametrization.   Before doing that we study carefully the framework proposed by \cite{lee2016predictive}. 

To begin with the standard conditional mean predictive regression model is given by
\begin{align}
y_t = \beta_0 + \beta_1^{\prime} x_{t-1} + u_{0t}
\end{align} 
with $\mathbb{E} \left( u_{0t} | \mathcal{F}_{t-1} \right) = 0$. Furthermore, the $p-$dimensional vector of predictors $x_{t-1}$ is generated via the following LUR specification
\begin{align}
x_t = R_n x_{t-1} + u_{xt} , \ \ R_n = \left( I_p + \frac{C}{n^{\alpha}}  \right), \ \text{for some} \ \alpha > 0.
\end{align}
where $n$ is the sample size, and $C = diag \left( c_1,...,c_p  \right)$  is the matrix for the degree of persistence. For the purpose of this paper, we consider that predictors employed for the estimation of the risk measures, belong only to one of the two persistence classes below
\begin{enumerate}
\item[(LUR)] Local-to-Unity: $\alpha = 1$ and $c_i \in ( - \infty, 0  ),  \ \forall \ i \in \left\{ 1,...,p  \right\}$.
\item[(MI)] Mildly Integrated: $\alpha \in (0,1)$ and $c_i \in ( - \infty, 0), \ \forall \ i \in \left\{ 1,...,p  \right\}$.
\item[(ME)] Mildly Explosive: $\alpha \in (0,1)$ and $c_i \in (0, +\infty),   \ \forall \ i \in \left\{ 1,...,p  \right\}$.
\end{enumerate} 

The innovation structure allows for linear process dependence for $u_{xt}$ and imposes a conditionally homoscedastic martingale difference sequence condition for $u_{0t}$ as is a standard practise in the related predictive regression literature. Under the related regulatory conditions, the usual \textit{FCLT} holds (as per \cite{phillips1992asymptotics})
\begin{align}
\frac{1}{\sqrt{n}} \sum_{j=1}^{ \floor{ ns } }
u_j :=
\frac{1}{\sqrt{n}}
\begin{bmatrix}
u_{0j} \\
u_{xj}
\end{bmatrix}
\equiv
\begin{bmatrix}
B_{0n} (s) \\
B_{xn} (s)
\end{bmatrix}
\Rightarrow
\begin{bmatrix}
B_{0} (s) \\
B_{x} (s)
\end{bmatrix}
:= 
\text{BM}
\begin{bmatrix}
\Sigma_{00}  &  \Sigma_{0x}  \\
\Sigma_{x0}  &  \Sigma_{xx} 
\end{bmatrix}
\end{align}
Next, we consider the corresponding quantile predictive regression model, which requires that
\begin{align}
Q_{y_t} \left( \tau | \mathcal{F}_{t-1} \right) = \beta_{0 \tau} + \beta_{1, \tau}^{\prime} x_{t-1}
\end{align}
Furthermore, in order to define the innovation structure of the quantile predictive regression system, we define the piecewise derivative of the loss function such that $\psi_{\tau} ( u ) = \tau - \mathbf{1} \left\{ u < 0 \right\}$. 

Thus, we have that $u_{0 t\tau} = u_{0t} - F^{-1}_{ u_0} (\tau)$ where $F^{-1}_{ u_0} (\tau)$  denotes the unconditional $\tau-$quantile of $u_{0t}$. Then, the corresponding \textit{FCLT} for the quantile predictive regression system is 
\begin{align}
\frac{1}{\sqrt{n}}  \sum_{t=1}^{ \floor{ nr} } 
\begin{bmatrix}
\psi_{\tau} (u_{0t \tau} ) \\
u_{xt}
\end{bmatrix}
\Rightarrow
\begin{bmatrix}
B_{ \psi_{\tau} } ( r ) \\
B_x ( r )
\end{bmatrix}
=
BM 
\begin{bmatrix}
\tau (1 - \tau) & \Sigma_{\psi_{\tau} x} 
\\
\Sigma_{x \psi_{\tau}} & \Omega_{xx}  
\end{bmatrix}
\end{align}   
Therefore, the following assumption holds. 
\begin{assumption}
\label{quantile.function.assumption}
The sequence of stationary conditional pdf $\big\{ f_{ u_{0t \tau}, t-1}(.)   \big\}$ evaluated at zero satisfies a $\textit{FCLT}$ with mean given by $\mathbb{E} \left[  f_{ u_{0t \tau}, t-1}(0) \right] > 0$ such that
\begin{align}
\frac{1}{ \sqrt{n} } \sum_{t=1}^{ \floor{nr} } \big( f_{ u_{0t \tau}, t-1}(0) - f_{ u_{0t \tau}}(0) \big) \Rightarrow B_{ f_{ u_{0t \tau} } } (r).
\end{align}
\end{assumption}

\subsubsection{Limit theory for quantile regression}

Note in the case that we allow for persistence regressors in the QR specification, then nonstandard limit theory applies, thus is worth examining the particular case. We follow the derivations in \cite{lee2016predictive}. 
\begin{align}
\hat{\beta}^{\text{QR}} \left( \tau  \right) = \underset{ \beta \in \mathbb{R}^{p+1} }{  \text{arg min} } \sum_{t=1}^n \phi_{\tau} \left( y_{t} - \beta^{\prime} X_{t-1}  \right)
\end{align}
where $\phi_{\tau}( u ) = u \left( \tau - \mathbf{1}_{ \left\{ u < 0 \right\} } \right)$, for $\tau \in (0,1)$, represents the asymmetric QR function and $X_{t-1} = \left( 1, x_{t-1}^{\prime} \right)^{\prime}$ includes the intercept and the set of regressors $x_{t-1}$. Following \cite{lee2016predictive}, we use the following normalization matrices according to the persistence properties of the regressors, that is, 
\begin{align}
D_n := 
\left\{
\begin{array}{ll}
      diag \big( \sqrt{n}, n \text{I}_p  \big)    & \text{for (LUR)}  \\
      diag \left( \sqrt{n}, n^{\frac{\alpha + 1 }{2}} \text{I}_p  \right)    & \text{for (MI)}  \\
      diag \big( \sqrt{n}, n^{\alpha} R_n^n   \big)    & \text{for (ME)}  \\
\end{array} 
\right. 
\end{align}
\begin{theorem}\citep{lee2016predictive}
Under Assumption  \ref{quantile.function.assumption} and \textit{FCLT} it follows that,
\begin{align*}
D_n \left( \hat{\beta} \left( \tau  \right) - \beta \left( \tau  \right)  \right) 
\Rightarrow
\left\{
\begin{array}{ll}
    f_{ u_{0t \tau}}(0)^{-1} 
    \begin{bmatrix}
    1 & \displaystyle \int_0^1 J_c(r)^{\prime} dr \\
    \displaystyle \int_0^1 J_c(r) dr & \displaystyle \int_0^1 J_c(r)  J_c(r)^{\prime} dr 
    \end{bmatrix} 
    \begin{bmatrix}
    B_{\psi_{\tau}}(1) \\
    \displaystyle \int_0^1 J_c(r)^{\prime} dr
    \end{bmatrix}                 & \text{(LUR)}  
\\
\\
     \mathcal{N} \displaystyle  \left( 0, \frac{ \tau (1 - \tau) }{ f_{u_{0\tau}}(0)^2  }  \right) 
     \begin{bmatrix}
        1  &  0 \\
        0  &  V_{xx}^{-1}
     \end{bmatrix}          
         & \text{(MI)}  
\\
\\
      \mathcal{MN} \displaystyle  \left( 0, \frac{ \tau (1 - \tau) }{ f_{u_{0\tau}}(0)^2  }  \right) 
     \begin{bmatrix}
        1  &  0 \\
        0  &  \widetilde{V}_{xx}^{-1}
     \end{bmatrix}    & \text{(ME)}  \\
\end{array} 
\right. 
\end{align*}
where the stochastic matrices $V_{xx}$ and  $V_{xx}$ are defined as below
\begin{align*}
V_{xx} = \int_0^{\infty} e^{rC} \Omega_{xx}  e^{rC} dr , \ \ \widetilde{V}_{xx} = \int_0^{\infty} e^{rC} Y_c Y_c^{\prime}  e^{rC} dr, \ \text{and} \ Y_c \equiv \mathcal{N} \left( 0, \int_0^{\infty} e^{-rC} \Omega_{xx}  e^{-rC} dr \right).
\end{align*}
\end{theorem}

\subsubsection{Limit theory for IVX-QR regression}

We consider the instrumental variable $Z_{tn}$ which is based on the IVX methodology proposed by \cite{phillipsmagdal2009econometric}. The IVX instrument is constructed as 
\begin{align}
Z_{tn} = \sum_{j=0}^{t-1} \left( 1 - \frac{c_z}{n^{\delta}}  \right)^j \left( X_{t-j} - X_{t-j-1} \right)
\end{align}
After demeaning the QR specification we obtain the transformed model without model intercept
\begin{align}
y_{t\tau} = \beta_{\tau}^{\prime} x_{t-1} + u_{0t \tau} 
\end{align} 
Then, to consider the corresponding limiting distribution of the IVX-QR estimators, we use the following embedded normalizations as below
\begin{align}
\tilde{Z}_{t-1,n} := \tilde{D}_n^{-1} \tilde{z}_{t-1} \ \ \text{and} \ \ \tilde{X}_{t-1,n} := \tilde{D}_n^{-1} \tilde{x}_{t-1} 
\end{align} 
and using the notation $\alpha  \wedge \delta = \text{min} \left( \alpha, \delta \right)$,
\begin{align}
\tilde{D}_n = 
\left\{
\begin{array}{ll}
     n^{\frac{\alpha \wedge \delta }{2}} \text{I}_p & \text{for (LUR) and (MI)}  \\
     n^{\frac{\alpha + 1 }{2}} R_n^n                & \text{for (ME)}  
\end{array} 
\right. 
\end{align}

The next theorem gives the limit theory of the IVX-QR estimator under various degrees of persistence (see, \cite{lee2016predictive} and \cite{katsouris2023structural}).

\medskip

\begin{theorem}(IVX-QR Limit Theory, see \cite{lee2016predictive}) Under Assumption \ref{quantile.function.assumption} it follows  
\begin{align*}
\tilde{D}_n \left( \hat{\beta}_{1}^{IVX-QR} \left( \tau  \right) - \beta_1 \left( \tau  \right)  \right) 
\Rightarrow
\left\{
\begin{array}{ll}
    \mathcal{MN} \displaystyle  \left( 0, \frac{ \tau (1 - \tau) }{ f_{u_{0\tau}}(0)^2  } \psi_{cxz}^{-1} V_{cxz} \left( \psi_{cxz}^{-1} \right)^{\prime} \right) 
     & \text{for (LUR) and (MI)}  
\\
\\
      \mathcal{MN} \displaystyle  \left( 0, \frac{ \tau (1 - \tau) }{ f_{u_{0\tau}}(0)^2  }  \left( \widetilde{V}_{xx} \right)^{-1}  \right) 
     & \text{for (ME)}  \\
\end{array} 
\right. 
\end{align*}
\end{theorem} 
The above Theorem shows that we obtain a mixed normal limiting distribution regardless of the persistence properties of the predictors (see, Theorem 3.1 in \cite{lee2016predictive}).

\bigskip

\begin{proof}
We aim to show that 
\begin{align}
\left( \hat{\beta}_1^{\text{IVX-QR}} \left( \tau  \right) - \beta_1 \left( \tau  \right)  \right) = \mathcal{O}_p \left(  n^{ - (1 + \delta ) / n } \right)
\end{align}

Consider the estimator distance measure such as 
\begin{align}
\hat{ \mathcal{E} } ( \tau ) = \left( \hat{\beta}_1 (\tau) - \beta_1 (\tau) \right)
\end{align}
and note that 
$\mathbb{E} \left[ \psi_{\tau} \left( u_{0t\tau} - \hat{ \mathcal{E} }( \tau )^{\prime} x_{t-1} \right) \right]$ which can be expanded around $\mathcal{E}( \tau )  = 0$ such that $\beta = \beta \left( \tau \right)$.

Then we obtain that
\begin{align*}
&n^{ - \frac{(1 + \delta ) }{2} } \sum_{t=1}^n \tilde{z}_{t-1} \left\{ \psi_{\tau} \left( u_{0t\tau} - \left( \hat{\beta}_1 (\tau) - \beta_1 (\tau) \right)^{\prime} x_{t-1} \right) \right\} 
\\
&= n^{ - \frac{(1 + \delta ) }{2} } \sum_{t=1}^n \tilde{z}_{t-1} \left\{ \psi_{\tau} \left( u_{0t\tau} - \hat{ \mathcal{E} } ( \tau )^{\prime} x_{t-1} \right)  - \mathbb{E}_{t-1} \left[ \psi_{\tau} \left( u_{0t\tau} - \hat{ \mathcal{E} } ( \tau )^{\prime} x_{t-1} \right)  \right] - \psi_{\tau} \left( u_{0t\tau} \right)  \right\}
\\
&+  n^{ - \frac{(1 + \delta ) }{2} } \sum_{t=1}^n \tilde{z}_{t-1} \mathbb{E}_{t-1} \big[ \psi_{\tau} \left( u_{0t\tau} \right) \big]
\end{align*}

In other words, we can prove that 
\begin{align}
G_{\tau, n} + \sum_{t=1}^n f_{ u_{0t \tau}, t-1 } (0) \tilde{Z}_{t-1} X_{t-1,n}^{\prime} n^{ (1 + \delta) / 2 } \big( \hat{\beta}_{1} ( \tau )  -  \beta_{1}( \tau ) \big)
\end{align}

Therefore, we have that 
\begin{align}
n^{ \frac{1 + \delta}{2} } \big( \hat{\beta}_{1} ( \tau )  -  \beta_{1}( \tau ) \big) = \left( M_{ \beta_n(\tau) } \right)^{-1} G_n ( \tau ) + o_p(1), 
\end{align}

\end{proof}

\newpage 

\subsection{Derivations on Quantile Cointegrating regression}

\subsubsection{Quantile regression on cointegrated time series}

Consider the following partial sum process
\begin{align}
Y_n(r) = \frac{1}{ \omega^{*}_{\psi} } n^{- 1/ 2} \sum_{j=1}^{ \floor{nr} } \psi_{\uptau} \left( \epsilon_{j \uptau} \right), 
\end{align}
where $\omega^{*}_{\psi}$ is the long-run variance of $\psi_{\uptau} \left( \epsilon_{j \uptau} \right)$ and the partial sum process follows an invariance principle and converges weakly to a standard Brownian motion $W(r)$. Choosing a continuous functional $h(.)$ that measures the fluctuation of $Y_n(r)$, then a robust test for cointegration can be constructed based on $h \big( Y_n(r) \big)$. By the continuous mapping theorem, under regularity conditions and the null of cointegration, 
\begin{align}
h \big( Y_n(r) \big) \Rightarrow h \big( W(r) \big).
\end{align} 
In principle, any metric that measures the fluctuations in $Y_n(r)$ is a natural candidate for the functional $h$. The classical KS typr or CvM type measures are of particular interest. A robust test for cointegration can then be constructed based on 
\begin{align}
\widehat{Y}_n(r) = \frac{1}{ \widehat{ \omega }^{*}_{\psi} } n^{- 1/ 2} \sum_{j=1}^{ \floor{nr} } \psi_{\uptau} \left( \widehat{ \epsilon }_{j \uptau} \right), 
\end{align}
where $\widehat{ \omega }^{*}_{\psi}$ is a consistent estimator of $\omega^{*}_{\psi}$.

Consider the following cointegration model
\begin{align}
y_t = \alpha + \beta^{\prime} \boldsymbol{x}_t + u_t \equiv \boldsymbol{\theta}^{\prime} \boldsymbol{z}_t + u_t 
\end{align}
where $\boldsymbol{x}_t$ is a $k-$dimensional vector of integrated regressors, $\boldsymbol{z}_t = \left( 1, \boldsymbol{x}_t^{\prime} \right)^{\prime}$ and $u_t$ is mean zero stationary innovation sequence. The quantile regression estimator of the cointegrating vector can be obtained by solving the following (where $p = k +1$)
\begin{align}
\widehat{ \boldsymbol{\theta} } ( \uptau ) = \underset{ \boldsymbol{\theta}  \in \mathbb{R}^{p+1} }{  \text{arg min} } \sum_{t=1}^n \rho_{ \uptau } \left( y_t - \boldsymbol{z}_t^{\prime} \boldsymbol{\theta} \right),
\end{align}
where $\rho_{ \uptau } (u) = u \big( \uptau - \mathbf{1} \left\{  \mathsf{u} < 0  \right\} \big)$. 

In order to derive the limit distribution of the quantile regression estimator of the cointegrating vector we follow the approach as presented in \cite{xiao2009quantile}. Let $f(.)$ and $F_{y|x}(.)$  be the pdf and the CDF of $u_t$ and denote the inverse of the check function with $\psi_{ \uptau } (u) = \big[ \uptau - \mathbf{1} \left\{  u < 0  \right\} \big]$. The following assumptions set the background theory in order to support the econometric analysis of this section. 

\newpage

\subsubsection{Main Assumptions}

\begin{assumption}
\label{assumptionA}
Let $\boldsymbol{v}_t = \Delta \boldsymbol{x}_t$. Then, $\boldsymbol{w}_t = ( u_t, \boldsymbol{v}_t^{\prime} )^{\prime}$ is a zero-mean, stationary sequence of $(k+1)-$dimensional random vectors. The partial sums of the vector process $\big( \psi_{\uptau} \big( u_t ( \uptau ) \big), \boldsymbol{v}_t \big)$ follow a multivariate invariance principle as expressed below
\begin{align}
n^{-1/2} \sum_{t=1}^{ \floor{ nr} } 
\begin{bmatrix}
\psi_{\uptau} \big( u_t ( \uptau ) \big) \\
\boldsymbol{v}_t
\end{bmatrix}
\Rightarrow
\begin{bmatrix}
B_{ \psi } ( r ) \\
\boldsymbol{B}_v ( r )
\end{bmatrix}
:=
\mathcal{BM} \big( \boldsymbol{0}, \boldsymbol{\Omega}_{ww} \big)
\end{align}  
where $\boldsymbol{\Omega}_{ww}$ is the covariance matrix of the Brownian motion $\left(  B_{ \psi_{\uptau} } ( r ),  \boldsymbol{B}_v ( r )^{\prime} \right)^{\prime}$.
\end{assumption}  

\begin{assumption}
\label{assumptionB}
The distribution function of $u_t$, $F_{y|x}( \mathsf{u} )$, has a continuous density $f( \mathsf{u} )$, with $f(\mathsf{u}) > 0$ on $\left\{ \mathsf{u} : 0 < F( \mathsf{u} ) < 1 \right\}$.
\end{assumption}

\begin{assumption}
\label{assumptionC}
The conditional distribution function $F_{t-1}( \mathsf{u} ) = \mathbb{P} \big( u_t < \mathsf{u} | u_{t-j} , j \geq 1  \big) $ has a derivative $f_{t-1} ( \mathsf{u} )$, almost surely, and $f_{t-1}(s_n)$ is uniformly integrable for any sequence $s_n \mapsto F^{-1} ( \uptau )$, and $\mathbb{E} \left( f_{t-1}^{\delta} \left(  F_{y|x}^{-1} ( \uptau ) \right)  \right) < \infty$ for some $\delta > 1$. 
\end{assumption}
Furthermore, we partition $\boldsymbol{\Omega}_{ww}$ as below
\begin{align}
\boldsymbol{\Omega}_{ww}
= 
\begin{bmatrix}
\omega_{\psi}^2  &  \boldsymbol{\Omega}_{\psi v} 
\\
\boldsymbol{\Omega}_{ v\psi } & \boldsymbol{\Omega}_{vv} 
\end{bmatrix}
\end{align}

Note that the asymptotic behaviour of $n^{-1} \sum_{t=1}^n \boldsymbol{x}_t \psi_{ \uptau } \big( u_t ( \uptau ) \big)$. Under Assumption \eqref{assumptionA} the following asymptotic result holds 
\begin{align}
n^{-1} \sum_{t=1}^n \boldsymbol{x}_t \psi_{ \uptau } \big( u_t ( \uptau ) \big) \Rightarrow \int_0^1 \boldsymbol{B}_v d\boldsymbol{B}_{ \psi_{ \uptau } } + \lambda_{ v \psi_{ \uptau } } 
\end{align}
where $\lambda_{ v \psi_{ \uptau } } $ is the one-sided long-run covariance between $\boldsymbol{v}_t$ and  $\psi_{ \uptau } \big( u_t ( \uptau ) \big)$.  

Due to the nonstationarity in $\boldsymbol{x}_t$, the two model parameters in the parameter vector $\widehat{ \boldsymbol{\theta} } ( \uptau ) = \left( \hat{ \boldsymbol{\alpha} }(\uptau), \hat{ \boldsymbol{\beta} }(\uptau)^{\prime}  \right)^{\prime}$ have different rates of convergence. In particular, the estimate of the cointegrating vector $\hat{ \boldsymbol{\beta} }(\uptau)^{\prime}$ converges at rate $n$, while the intercept $\hat{ \boldsymbol{\alpha} }(\uptau)$ converges at rate $\sqrt{n}$. Thus, in order to account for the different convergence rates among the model intercept  and the estimate of the cointegrating vector we use the normalization matrix $\mathcal{ \boldsymbol{D} }_n = \mathsf{diag} \left( \sqrt{n} , n \boldsymbol{I}_k \right)$, where $\boldsymbol{I}_k$ is the $k \times k$ identity matrix. The limit distribution of the quantile regression estimator for the cointegration model is given by the Theorem \ref{theorem1}.

\newpage

\begin{theorem} \label{theorem1}
Under Assumptions \eqref{assumptionA} - \eqref{assumptionC}
\begin{align}
\mathcal{ \boldsymbol{D} }_n \left( \widehat{ \boldsymbol{\theta} } ( \uptau ) - \boldsymbol{\theta} ( \uptau )  \right)
\Rightarrow \frac{1}{ f \big( F_{y|x}^{-1} ( \uptau) \big) } \left[ \int_0^1 \widetilde{ \boldsymbol{B} }_v \widetilde{ \boldsymbol{B} }_v^{\prime} \right]^{-1} \times \left[ \int_0^1 \widetilde{ \boldsymbol{B }}_v (r) d \boldsymbol{B}_{ \psi_{\uptau}} + \Delta_{v \psi_{\uptau} }  \right],
\end{align}
where $\widetilde{ \boldsymbol{B} } = \big( 1, \boldsymbol{B}_v(r)^{\prime} \big)^{\prime}$, and $\Delta_{v \psi} = \left( 0, \lambda_{ v \psi}^{\prime} \right)^{\prime}$. In particular, we obtain that
\begin{align} 
\label{limit.beta}
n \left( \widehat{ \boldsymbol{\beta} } ( \tau ) - \boldsymbol{\beta} \right) \Rightarrow \frac{1}{ f \left( F_{y|x}^{-1} ( \tau) \right)} \left[ \int_0^1 \boldsymbol{B}_v^{\mu} \boldsymbol{B}_v^{ \prime \mu }  \right]^{-1} \times \left[ \int_0^1  \boldsymbol{B}_v^{\mu}(r) d\boldsymbol{B}_{\psi}  +  \lambda_{ v \psi} \right]  
\end{align}
where $\boldsymbol{B}_v^{\mu}(r) = \boldsymbol{B}_v(r) - r \boldsymbol{B}_v(1)$ is a $k-$dimensional demeaned Brownian motion. 
\end{theorem}

\begin{remark}
Note that the quantile regression estimator is consistent at the usual $\mathcal{O}(n)$ rate, however similar to OLS, the quantile regression estimator suffers from second order bias $( \lambda_{ v \psi} )$ which is induced from the correlation between the set of regressors $x_t$ and the residual $u_t$. Furthermore, note that the Brownian motions $\boldsymbol{B}_v(r)$ and $\boldsymbol{B}_{ \psi_{\uptau} } ( r )$ are generally correlated as long as $\boldsymbol{\Omega}_{\psi_{\uptau} v} \neq 0$. Moreover, for both stationary and nonstationary time series regression, the limit distribution in \eqref{limit.beta} depends on the sparsity function $1 / f \left( F_{y|x}^{-1} ( \uptau ) \right)$. In particular, in the special case when $\Delta_{ v \psi } = 0$ and $\boldsymbol{\Omega}_{\psi_{\uptau} v} = 0$, that is, $x_t$ and $u_t$ are independent, then the limit distribution given by \eqref{limit.beta} is a mixed normal.
\end{remark}

\subsubsection{A fully-modified quantile regression estimator}
\label{SuppA1}

In this section, we are interested to propose a robust econometric framework for robust estimation and inference for the quantile regression in the case of a cointegrating model. Firstly, we note that the asymptotic behaviour of quantile regression-based  inference procedures depends on the limit distribution of $\widehat{ \boldsymbol{\theta} } ( \tau )$. However, as shown by Theorem \ref{theorem1}, the limiting processes $\boldsymbol{B}_v ( r )$  and $B_{ \psi_{\uptau} } ( r )$ are correlated Brownian motions whenever contemporaneous correlation between $v_t$ and $\psi_{ \uptau } \big( u_t( \uptau ) \big)$ exists. 

Moreover, super-consistency, $\widehat{ \boldsymbol{\theta} } ( \uptau )$ is second-order biased and the miscentering effect in the limit distribution is reflected in $\Delta_{v \psi}$. Consequently, the distribution of the test based on the quantile regression residual will be dependent on nuisance parameters. Therefore, in order to restore the asymptotic nuisance parameter free property of the inference procedure, we need to modify the original quantile regression estimator so that we obtain a mixed normal limit distribution. Usually in the literature, there are two approaches one can take to achieve the restoration of the the asymptotic nuisance parameter free property. For example, one approach is the implementation of a nonparametric fully-modification of the original quantile regression estimator, and a second approach is the parametrically augmented quantile regression implementation using leads and lags. In order to develop a fully-modified quantile cointegrating regression estimator one can follow the approach proposed by \cite{phillips1990statistical}.

\newpage

We first decompose the asymptotic distribution given by \eqref{limit.beta} into the following two components:
\begin{align}
\frac{1}{ f \left( F_{y|x}^{-1} ( \uptau ) \right) } \left[ \int_0^1 \boldsymbol{B}_v^{\mu} \boldsymbol{B}_v^{\prime \mu}  \right]^{-1} &\times \left[ \int_0^1  \boldsymbol{B}_v^{\mu} d\boldsymbol{B}_{\psi . v}  \right], \text{and}   
\\
\frac{1}{ f \left( F_{y|x}^{-1} ( \uptau ) \right) } \left[ \int_0^1 \boldsymbol{B}_v^{\mu} \boldsymbol{B}_v^{\prime \mu}  \right]^{-1} &\times \left[ \int_0^1  \boldsymbol{B}_v^{\mu} d\boldsymbol{B}_{v}^{\prime} \boldsymbol{\Omega}_{vv}^{-1} \boldsymbol{\Omega}_{v \psi} + \lambda_{v \psi }  \right]
\end{align}
where 
\begin{align}
B_{\psi . v} &= B_{\psi} ( r ) - \Omega_{\psi v } \Omega_{vv}^{-1} B_v (r)
\\
\omega^2_{ \psi . v} &= \omega^2_{ \psi} - \Omega_{\psi v} \Omega_{vv}^{-1} \Omega_{v \psi}
\end{align}
such that $B_{\psi . v}$ is a Brownian motion with variance $\omega^2_{ \psi . v}$ as defined above. Moreover, note that $\boldsymbol{B}_{\psi . v}$ is independent of $\boldsymbol{B}_v(r)$ and the first term in the above decomposition, $\left[ \int_0^1 \boldsymbol{B}_v^{\mu} \boldsymbol{B}_v^{\prime \mu}  \right]^{-1} \left[ \int_0^1 \boldsymbol{B}_v^{\mu} d\boldsymbol{B}_{\psi . v}  \right]$, is a mixed Gaussian variate. The basic idea of fully-modification on $\widehat{ \boldsymbol{\beta} } ( \uptau )$ is to construct a nonparametric correction to remove the second term in the above decomposition. 

Therefore, to facilitate the nonparametric correction, we consider the following kernel estimators of $\Omega_{vv}$, $\Omega_{v \psi}$, $\lambda_{v \psi}$,  $\lambda_{vv}$, such that
\begin{align}
\widehat{\lambda}_{v \psi} &= \sum_{ h = 0}^M k \left( \frac{h}{M} \right) \Sigma_{v \psi } (h ), \ \ \ \ \ \widehat{\lambda}_{vv} = \sum_{ h = 0}^M k \left( \frac{h}{M} \right) \Sigma_{vv} (h),
\\
\widehat{\Omega}_{v \psi} &= \sum_{ h = -M}^M k \left( \frac{h}{M} \right) \Sigma_{v \psi } (h ), \ \ \ \ \ \widehat{\Omega}_{vv} = \sum_{ h = -M}^M k \left( \frac{h}{M} \right) \Sigma_{vv} (h),
\end{align}
where $k(.)$ is the kernel function defined on the interval $[-1,1]$ with $k(0) = 1$, and $M$ the bandwidth parameter satisfying the property that $M \to \infty$ and $M / n \to 0$. The sample covariance matrices are
\begin{align}
\boldsymbol{\Sigma}_{ v\psi_{\uptau} } (h) =  \frac{1}{n} \sum_{t=1}^{t+h = n } v_t \psi_{\tau} \left( \hat{u}_{t+h, \tau} \right), \ \ \ \ \boldsymbol{\Sigma}_{vv} (h) =  \frac{1}{n} \sum_{t=1}^{t+h = n } \boldsymbol{v}_t \boldsymbol{v}_{t+h}^{\prime}
\end{align}
We define the following nonparametric fully modified quantile regression estimators as below \begin{align}
\widehat{ \boldsymbol{\theta} } ( \uptau )^{+} &= 
\begin{pmatrix}
\widehat{ \alpha } ( \uptau ) 
\\
\widehat{ \beta } ( \uptau )^{+}
\end{pmatrix}  
\\
\widehat{ \boldsymbol{\beta} } ( \uptau )^{+} &=  \widehat{ \boldsymbol{\beta} } ( \uptau ) - \frac{1}{ f_{y|x} \left( \widehat{ F_{y|x}^{-1} ( \uptau)  }\right) } \left[ \sum_{t=1}^n \boldsymbol{x}_t \boldsymbol{x}_t^{\prime} \right]^{-1} \times \left[ \sum_{t=1}^n \boldsymbol{x}_t \boldsymbol{v}_t ^{\prime} \widehat{\boldsymbol{\Omega} }_{vv}^{-1} \widehat{\Omega}_{v \psi} + \widehat{\lambda}_{v \psi}^{+}  \right]
\end{align}
and $\widehat{\lambda}_{v \psi}^{+} = \widehat{\lambda}_{v \psi} - \widehat{\lambda}_{v v} \widehat{\Omega}_{vv}^{-1} \widehat{\Omega}_{v\psi}   $. 

\newpage

Similar to the fully modified OLS estimators, the fully modified quantile regression estimator of the cointegrating vector has a mixed normal distribution in the limit.
\begin{theorem}
\label{theorem2}
Under Assumptions \eqref{assumptionA}-\eqref{assumptionC}
\begin{align}
\mathcal{ \boldsymbol{D} }_n \left( \widehat{ \boldsymbol{\theta} } ( \uptau )^{+} - \widehat{ \boldsymbol{\theta} } ( \uptau ) \right) \Rightarrow \frac{1}{ f_{y|x} \left( F_{y|x}^{-1} ( \uptau)\right) } \left[ \int_0^1 \widetilde{ \boldsymbol{B} }_v \widetilde{ \boldsymbol{B} }_v^{\prime} \right]^{-1} \int_0^1 \widetilde{ \boldsymbol{B} }_v d\boldsymbol{B}_{ \psi.v }
\\
\equiv \mathcal{MN} \left( 0, \frac{ \omega^2_{ \psi.v } }{ f_{y|x} \left(  F_{y|x}^{-1} ( \uptau) \right)^2} \left[ \int_0^1 \widetilde{\boldsymbol{B}}_v \widetilde{\boldsymbol{B}}_v^{\prime} \right]^{-1} \right).
\end{align}
In particular, the cointegrating vector has the following asymptotic distribution
\begin{align}
n \left( \widehat{ \boldsymbol{\beta} } ( \uptau )^{+} - \widehat{ \boldsymbol{\beta}  } ( \uptau ) \right) \Rightarrow \frac{1}{ f_{y|x} \left( F_{y|x}^{-1} ( \uptau ) \right) } \left[ \int_0^1 \boldsymbol{B}_v^{\mu} \boldsymbol{B}_v^{\prime \mu} \right]^{-1} \int_0^1 \boldsymbol{B}_v^{\mu} d\boldsymbol{B}_{ \psi.v }
\\
\equiv \mathcal{MN} \left( 0, \frac{ \omega^2_{ \psi.v } }{ f_{y|x} \left(  F_{y|x}^{-1} ( \uptau) \right)^2} \left[ \int_0^1 \boldsymbol{B}_v^{\mu} \boldsymbol{B}_v^{\prime \mu} \right]^{-1} \right).
\end{align}
\end{theorem}
Note that the fully modified quantile regression estimator and the resulting asymptotic mixed normal asymptotic distribution facilitates statistical inference (such as the use of the Wald test) based on quantile cointegrating regression. Therefore, the classical inference problem of linear restrictions on the cointegrating vector $\beta$ is such that $\mathbb{H}_0: \mathcal{R} \boldsymbol{\beta} = r$, where $\mathcal{R}$ denotes an $(q \times k)$ matrix of linear restrictions and $r$ denotes a $q-$ dimensional vector.

Thus, under the null hypothesis, and the assumptions of Theorem \ref{theorem2}, we have that 
\begin{align}
\frac{ f_{y|x} \left(  F_{y|x}^{-1} ( \uptau) \right)  }{ \omega_{\psi.v} } \left[  \mathcal{ \boldsymbol{R} } \left( \int_0^1 \boldsymbol{B}_v^{\mu} \boldsymbol{B}_v^{\prime \mu} \right)^{-1} \mathcal{\boldsymbol{R}}^{\prime} \right]^{-1 / 2} n \left( \mathcal{\boldsymbol{R}} \widehat{ \boldsymbol{\beta} } ( \uptau )^{+} - r \right) \Rightarrow \mathcal{N} \big( \boldsymbol{0}, \boldsymbol{I}_q \big) 
\end{align}
where $\mathcal{N} \big( \boldsymbol{0}, \boldsymbol{I}_q \big)$ represents a $q-$dimensional standard Normal distribution. 
    
Define with $\mathcal{ \boldsymbol{Q} }_x = \sum_{t=1}^n ( \boldsymbol{x}_t - \bar{\boldsymbol{x}} ) ( \boldsymbol{x}_t - \bar{\boldsymbol{x}} )^{\prime}$. Then, a the classical Wald test in the case of quantile cointegrating regression can be constructed via the following expression  
\begin{align}
\mathcal{W}_n \left( \tau \right) = \frac{ f_{y|x} \left(  \widehat{ F_{y|x}^{-1} ( \uptau)} \right)  }{ \widehat{ \omega}_{\psi.v} } \left( \mathcal{\boldsymbol{R}} \widehat{ \boldsymbol{\beta} } ( \uptau )^{+} - r \right)^{ \prime } \left[ \mathcal{ \boldsymbol{R} }  \mathcal{ \boldsymbol{Q} }_x^{-1} \mathcal{ \boldsymbol{R} } \right]^{-1} \left( \mathcal{\boldsymbol{R}} \widehat{ \boldsymbol{\beta} } ( \uptau )^{+} - r \right), 
\end{align}
where $f_{y|x} \left(  \widehat{ F_{y|x}^{-1} ( \uptau)} \right)$ and $\widehat{ \omega}_{\psi.v}$ are consistent estimators of $f_{y|x} \left( F_{y|x}^{-1} ( \uptau) \right)$ and $\omega_{\psi.v}$ respectively. Then, the asymptotic distribution of the Wald statistic is summarized via the following Theorem.

\newpage 

\begin{theorem}
Under the assumptions of Theorem \ref{theorem2} and the linear restriction under the null hypothesis $\mathbb{H}_0$, we have that
\begin{align*}
\mathcal{W}_n \Rightarrow \chi_q^2,
\end{align*}
\end{theorem}
where $\chi_q^2$ is a centred Chi-square random variable with $q-$degrees of freedom. 

\begin{proof}
We consider the derivations of the proof of Theorem 1 in \cite{xiao2009quantile}. The following results is useful in developing asymptotics for the regression quantile estimates: For $\mathsf{u} \neq 0$, we have that
\begin{align}
\label{trick}
\rho_{\uptau} ( \mathsf{u} - \mathsf{v} ) - \rho_{\uptau}(\mathsf{u}) &=
- \mathsf{v} \psi_{\uptau}(\mathsf{u}) + ( \mathsf{u} - \mathsf{v} ) \bigg[ \mathbf{1} \big\{ \mathsf{v}  < \mathsf{u}  < 0 \big\} - \mathbf{1} \big\{ 0 < \mathsf{u} < \mathsf{v} \big\}   \bigg]
\\
\psi_{\uptau}( \mathsf{u} ) &= \big[ \tau - \mathbf{1} \big\{ \mathsf{u} < 0  \big\} \big], \ \ u_{t} (\uptau) = u_t - F^{-1} ( \tau )  \ \ \text{and} \ \  Q_{ u_t (\uptau) } = 0.    
\end{align}
Moreover, we denote with $\boldsymbol{\theta}( \uptau ) := \big( \alpha + F^{-1}( \uptau ), \boldsymbol{\beta} \big)^{\prime}$, then $u_t (\uptau)  = y_t - \boldsymbol{\theta} ( \uptau )^{\prime} \boldsymbol{z}_t$. Furthermore, let $\hat{ \boldsymbol{v} } = \boldsymbol{D}_n  \left( \widehat{ \boldsymbol{\theta} }( \uptau) - \boldsymbol{\theta}( \uptau) \right)$, where $\boldsymbol{D}_n = \mathsf{diag} \left( \sqrt{n}, n,...n \right)$. Then, we obtain
\begin{align}
\rho_{\uptau} \left( y_t - \widehat{ \boldsymbol{\theta} }( \uptau)^{\prime} \boldsymbol{z}_t  \right) 
= 
\rho_{\uptau} \left( u_t (\uptau) - \left( \boldsymbol{D}_n^{-1} \hat{ \boldsymbol{v} } \right)^{\prime} \boldsymbol{z}_t \right).
\end{align} 
Denote with 
\begin{align}
\mathcal{Z}_n( \boldsymbol{v} ) = \sum_{t=1}^n \left[  \rho_{\uptau} \left(  u_t (\uptau) - \left( \boldsymbol{D}_n^{-1} \boldsymbol{v} \right)^{\prime} \boldsymbol{z}_t  \right)  - \rho_{\uptau} \big( u_t (\uptau) \big) \right]
\end{align}
which is a convex random function. 

Then, the minimization of $\mathcal{Z}_n(\boldsymbol{v})$ is equivalent to the minimization of the original optimization problem, which implies that if $\hat{\boldsymbol{v}}$ is the minimizer of $\mathcal{Z}_n(\boldsymbol{v})$, then we have that $\hat{\boldsymbol{v}} = \boldsymbol{D}_n  \left( \widehat{ \boldsymbol{\theta} }( \uptau) - \boldsymbol{\theta} ( \uptau) \right)$. Furthermore, the convexity of $\mathcal{Z}_n( .)$ implies that $\hat{\boldsymbol{v}}$ converges in distribution to the minimizer of $\mathcal{Z}(.)$. In general, $u_t$ and $\Delta \boldsymbol{x}_t$ are correlated and thus $B_{\psi}$ and $\boldsymbol{B}_v$ are correlated Brownian motions. Using the invariance principle given by Assumption \ref{assumptionA} and the property given by expression \eqref{trick}, then the objective function of the minimization problem can be written as below
\begin{align*}
\mathcal{Z}_n( \boldsymbol{v} ) 
&= 
\sum_{t=1}^n \bigg[  \rho_{\uptau} \left(  u_t (\uptau) - \left( \boldsymbol{D}_n^{-1} \boldsymbol{v} \right)^{\prime} \boldsymbol{z}_t  \right)  - \rho_{\uptau} \big( u_t (\uptau) \big) \bigg] 
\nonumber
\\
&\ = - \sum_{t=1}^n  \left( \boldsymbol{D}_n^{-1} \boldsymbol{v} \right)^{\prime} \boldsymbol{z}_t  \psi_{\uptau} \big( u_t (\uptau) \big)  
\nonumber
\\
&\ \ \ \ \ + \sum_{t=1}^n  \left( u_t (\uptau) - \left( \boldsymbol{D}_n^{-1} \boldsymbol{v} \right)^{\prime} \boldsymbol{z}_t \right)  \bigg[ \mathbf{1}\left\{ \left( \boldsymbol{D}_n^{-1} \boldsymbol{v} \right)^{\prime} \boldsymbol{z}_t < u_t (\uptau) < 0 \right\} - \mathbf{1}\left\{ 0 < u_t (\uptau) < \left( \boldsymbol{D}_n^{-1} \boldsymbol{v} \right)^{\prime} \boldsymbol{z}_t \right\}   \bigg]
\end{align*}  

\newpage

Thus, under Assumptions \ref{assumptionA} - \ref{assumptionB}, we have that
\begin{align}
\boldsymbol{D}_n^{-1} \sum_{t=1}^n \boldsymbol{z}_t \psi_{ \uptau } \big( u_t ( \uptau ) \big) 
= 
\begin{bmatrix}
\displaystyle n^{-1 / 2}  \sum_{t=1}^n \psi_{ \uptau } \big( u_t ( \uptau ) \big) 
\\
\displaystyle  n^{-1}  \sum_{t=1}^n \boldsymbol{x}_t \psi_{ \uptau } \big( u_t ( \uptau ) \big) 
\end{bmatrix}
\Rightarrow
\begin{bmatrix}
\displaystyle \int_0^1 d \boldsymbol{B}_{ \psi_{ \uptau } } 
\\
\displaystyle \int_0^1 \boldsymbol{B}_{x} d \boldsymbol{B}_{ \psi_{ \uptau } } + \lambda_{x \psi} 
\end{bmatrix}
\end{align}
Next, we examine the limits of the following expressions
\begin{align}
\sum_{t=1}^n  \left( u_t ( \uptau )  -  \boldsymbol{v}^{ \prime } \boldsymbol{D}_n^{-1} \boldsymbol{z}_t \right) &\mathbf{1}\left\{ 0 < u_t ( \uptau )  < \boldsymbol{v}^{ \prime } \boldsymbol{D}_n^{-1} \boldsymbol{z}_t \right\}  \
\\
\sum_{t=1}^n  \left( u_t ( \uptau )  - \boldsymbol{v}^{ \prime } \boldsymbol{D}_n^{-1} \boldsymbol{z}_t \right) &\mathbf{1}\left\{ \boldsymbol{v}^{ \prime } \boldsymbol{D}_n^{-1} \boldsymbol{z}_t < u_t ( \uptau )  < 0 \right\}  \
\end{align}
Note that we denote with $\boldsymbol{v} = \boldsymbol{D}_n \left( \boldsymbol{\theta} - \boldsymbol{\theta} ( \uptau ) \right)$, and partition $\boldsymbol{v}$ and $ \boldsymbol{\theta} ( \uptau )$ comfortable with $\boldsymbol{z}_t = (1, \boldsymbol{x}_t^{\prime})^{\prime}$, 
\begin{align*}
v = 
\begin{bmatrix}
v_1
\\
v_2
\end{bmatrix},
\ \  \
\boldsymbol{\theta} ( \uptau )
= 
\begin{bmatrix}
\alpha( \uptau )
\\
\boldsymbol{\beta}( \uptau )
\end{bmatrix}
\end{align*}
Then, for convenience of the asymptotic analysis, we denote with 
\begin{align}
W_n( v) = \sum_{t=1}^n \big( \boldsymbol{v}^{ \prime } \boldsymbol{D}_n^{-1} \boldsymbol{z}_t - u_t ( \uptau ) \big) \mathbf{1} \left\{ 0 < u_t ( \uptau ) < \boldsymbol{v}^{ \prime } \boldsymbol{D}_n^{-1} \boldsymbol{z}_t  \right\}.
\end{align}
Furthermore, to avoid technical problems in taking conditional expectations, we consider truncation of $\boldsymbol{v}^{ \prime } \boldsymbol{D}_n^{-1} \boldsymbol{z}_t$ at some finite number $m > 0$ and denote with 
\begin{align}
W_{nm} (v) &= \sum_{t=1}^n \xi_{tm} (v) 
\\
\xi_{tm} (v) &= \big( \boldsymbol{v}^{ \prime } \boldsymbol{D}_n^{-1} \boldsymbol{z}_t  - u_t ( \uptau ) \big) \mathbf{1} \left\{ 0 < u_t ( \uptau ) < \boldsymbol{v}^{ \prime } \boldsymbol{D}_n^{-1} \boldsymbol{z}_t \right\} \mathbf{1} \left\{ \boldsymbol{v}^{ \prime } \boldsymbol{D}_n^{-1} \boldsymbol{z}_t \leq m \right\}.
\end{align}
Moreover, we denote the information set up to time $t$ as $\mathcal{F}_{t-1} = \sigma \big\{ u_{t-j}, v_{t-j+1}, j \geq 1 \big\}$, then $z_t \in \mathcal{F}_{t-1}$. We further define with 
\begin{align*}
\bar{\xi}_{tm} (v) &= \mathbb{E} \big[ \left( \boldsymbol{v}^{ \prime } \boldsymbol{D}_n^{-1} \boldsymbol{z}_t  - u_t ( \uptau ) \right)  \mathbf{1} \big\{ 0 < u_t ( \uptau ) < \boldsymbol{v}^{ \prime } \boldsymbol{D}_n^{-1} \boldsymbol{z}_t \big\} \times \mathbf{1} \big\{  \boldsymbol{v}^{ \prime } \boldsymbol{D}_n^{-1} \boldsymbol{z}_t \leq m \big| \mathcal{F}_{t-1} \big\} \big]
\\
\bar{W}_{nm} (v) &= \sum_{t=1}^n \bar{\xi}_{tm} (v), 
\end{align*}
Then, we have that $\big\{ \xi_{tm} (v) - \bar{\xi}_{tm} (v) \big\}$ is a martingale difference sequence. Therefore, we have that 
\begin{align*}
\bar{W}_{nm} (v) 
&= 
\sum_{t=1}^n \mathbb{E} \big[ \left( \boldsymbol{v}^{ \prime } \boldsymbol{D}_n^{-1} \boldsymbol{z}_t  - u_t ( \uptau ) \right)  \mathbf{1} \big\{ 0 < u_t ( \uptau ) < \boldsymbol{v}^{ \prime } \boldsymbol{D}_n^{-1} \boldsymbol{z}_t \big\} \times \mathbf{1} \big\{  \boldsymbol{v}^{ \prime } \boldsymbol{D}_n^{-1} \boldsymbol{z}_t \leq m \big| \mathcal{F}_{t-1} \big\} \big]
\end{align*}

\newpage

By replacing with \textcolor{red}{$u_t ( \uptau ) = F^{-1} (u) - u_t$}. Moreover, denote with 
\begin{itemize}
\item $\textcolor{red}{ \boldsymbol{a} := \big[ \boldsymbol{v}^{ \prime } \boldsymbol{D}_n^{-1} \boldsymbol{z}_t + F^{-1} (\uptau) \big] \times \mathbf{1} \left\{ \boldsymbol{v}^{ \prime } \boldsymbol{D}_n^{-1} \boldsymbol{z}_t \leq m \right\} }$

\item $\textcolor{blue}{ \boldsymbol{b} := F^{-1} (\uptau) \leq s \leq \big[ \boldsymbol{v}^{ \prime } \boldsymbol{D}_n^{-1} \boldsymbol{z}_t + F^{-1} (\uptau) \big] \times \mathbf{1} \left\{ \boldsymbol{v}^{ \prime } \boldsymbol{D}_n^{-1} \boldsymbol{z}_t \leq m \right\} }$

\end{itemize}
Then, we obtain the following result
\begin{align*}
\bar{W}_{nm} (v) 
&=
\sum_{t=1}^n \mathbb{E} \big[ \left( \boldsymbol{v}^{ \prime } \boldsymbol{D}_n^{-1} \boldsymbol{z}_t  - \textcolor{red}{ F^{-1} (u) - u_t} \right)  \mathbf{1} \big\{ 0 < \textcolor{red}{u_t} < \boldsymbol{v}^{ \prime } \boldsymbol{D}_n^{-1} \boldsymbol{z}_t + \textcolor{red}{ F^{-1} (u) } \big\} \times \mathbf{1} \big\{  \boldsymbol{v}^{ \prime } \boldsymbol{D}_n^{-1} \boldsymbol{z}_t \leq m \big| \mathcal{F}_{t-1} \big\} \big]
\\
&= 
\sum_{t=1}^n \int_{ F^{-1} ( \uptau ) }^{ \textcolor{red}{ \boldsymbol{a} } } \times \left[ \int_r^{ \textcolor{red}{ \boldsymbol{a} } } ds \right] f_{t-1} (r) dr 
\\
&=
\sum_{t=1}^n \int_{ \textcolor{blue}{ \boldsymbol{b} } } \times \int_{ F^{-1} ( \uptau ) \leq r \leq s } f_{t-1} (r) dr ds
\\
&= 
\sum_{t=1}^n \int_{ F^{-1} ( \uptau ) }^{ \textcolor{red}{ \boldsymbol{a} } } \big[ s - F^{-1} ( \uptau ) \big] \times \left[ \frac{ F_{t-1}(s) - F_{t-1} \big( F^{-1} ( \uptau ) \big) }{ s - F^{-1} ( \uptau ) } \right] ds.
\end{align*}
Furthermore, notice that the probability density function $f_{t-1} ( s_n )$ is uniformly integrable for any sequence $s_n \to F^{-1} ( \uptau )$, then we have 
\begin{align*}
\bar{W}_{nm} (v) 
&=
\sum_{t=1}^n \int_{ F^{-1} ( \uptau ) }^{ \textcolor{red}{ \boldsymbol{a} } } \big[ s - F^{-1} ( \uptau ) \big] f_{t-1} \big( F^{-1} ( \uptau ) \big) ds + o_{ \mathbb{P} } (1)
\\
&= 
\sum_{t=1}^n f_{t-1} \big( F^{-1} ( \uptau ) \big) \times \left\{ \frac{ \big[ s - F^{-1} ( \uptau ) \big]^2 }{2} \bigg|_{ F^{-1} ( \uptau ) }^{ \textcolor{red}{ \boldsymbol{a} } } \right\} + o_{\mathbb{P}} (1)
\end{align*} 
Therefore, we have that 
\begin{align*}
\bar{W}_{nm} (v) 
&=
\frac{1}{2} \sum_{t=1}^n f_{t-1} \big( F^{-1} ( \uptau ) \big) \big[ \boldsymbol{v}^{ \prime } \boldsymbol{D}_n^{-1} \boldsymbol{z}_t \big]^2 \times \mathbf{1} \left\{ \boldsymbol{v}^{ \prime } \boldsymbol{D}_n^{-1} \boldsymbol{z}_t \leq m \right\} + o_{\mathbb{P}}(1)
\\
&=
\frac{1}{2n} \sum_{t=1}^n f_{t-1} \big( F^{-1} ( \uptau ) \big) \boldsymbol{v}^{ \prime } \big[ \sqrt{n} \boldsymbol{D}_n^{-1} \boldsymbol{z}_t \boldsymbol{z}_t^{\prime} \boldsymbol{D}_n^{-1} \sqrt{n} \big] \times \boldsymbol{v} \times \mathbf{1} \left\{ \boldsymbol{v}^{ \prime } \boldsymbol{D}_n^{-1} \boldsymbol{z}_t \leq m \right\} + o_{\mathbb{P}}(1).  
\end{align*}
Furthermore, notice that by stationarity of the sparsity function $f_{t-1} \big( F^{-1} ( \uptau ) \big)$ , we have 
\begin{align}
\underset{ 0 \leq r \leq 1 }{ \mathsf{sup} } \left| \frac{1}{ n^{1 - \epsilon} } \sum_{t=1}^{ \floor{nr} } \big[ f_{t-1} \left( F^{-1} ( \uptau ) \right) \big] - f \left( F^{-1} ( \uptau ) \right) \right| \overset{ p }{ \to } 0
\end{align} 
for some $\epsilon > 0$. Thus, it follows that 
\begin{align*}
\bar{W}_{nm} (v) \Rightarrow \frac{1}{2} f \big( F^{-1} (\uptau)   \big) \boldsymbol{v}^{\prime} 
\left\{ 
\int_0^1 
\begin{bmatrix}
1 &  \displaystyle \int_0^1 \boldsymbol{B}_x
\\
\displaystyle \int_0^1 \boldsymbol{B}_x & \displaystyle \int_0^1 \boldsymbol{B}_x \boldsymbol{B}_x^{\prime}
\end{bmatrix} \times 
\mathbf{1} \big\{ 0 < \mathsf{v}_1 + \mathsf{v}_2 \boldsymbol{B}_x (s) \leq m \big\}
\right\} \boldsymbol{v} := \eta_m
\end{align*}

\newpage

We now follow similar arguments as in \cite{pollard1991asymptotics} by noting that 
\begin{align}
\big( \boldsymbol{v}^{ \prime } \boldsymbol{D}_n^{-1} \boldsymbol{z}_t \big) \times \mathbf{1} \big\{ 0 \leq \boldsymbol{v}^{ \prime } \boldsymbol{D}_n^{-1} \boldsymbol{z}_t \leq m \big\} \overset{ p }{ \to } 0
\end{align}
which holds uniformly in $t$. Therefore, we obtain that
\begin{align*}
\sum_{t=1}^n \mathbb{E} \big[ \xi_{tm} (v)^2 | \mathcal{F}_{t-1} \big] \leq \mathsf{max} \bigg\{  \big( \boldsymbol{v}^{ \prime } \boldsymbol{D}_n^{-1} \boldsymbol{z}_t \big) \times \mathbf{1} \big\{ 0 \leq \boldsymbol{v}^{ \prime } \boldsymbol{D}_n^{-1} \boldsymbol{z}_t \leq m \big\} \bigg\} \times \sum_{t=1}^n \bar{\xi}_{tm} (v) \overset{ p }{ \to } 0. 
\end{align*}
Therefore, the following summation of martingale difference sequence 
\begin{align}
\sum_{t=1}^n \big[ \xi_{tm} (v) - \bar{\xi}_{tm} (v) \big]
\end{align}
converges to zero in probability. Notice that by the asymptotic equivalence lemma, the limit distribution of $\sum_{t=1}^n \xi_{tm} (v)$ is the same as that of $\sum_{t=1}^n \bar{\xi}_{tm} (v)$ such that, 
\begin{align}
W_{nm} ( v ) \Rightarrow \eta_{m}
\end{align}

Let $m \to \infty$, we have that 
\begin{align}
\eta_{m} \Rightarrow \frac{1}{2} f \big( F^{-1} \uptau ) \big) \boldsymbol{v}^{\prime} \big[ \int_0^1 \boldsymbol{B}_z \boldsymbol{B}_z^{\prime} \big] \boldsymbol{v} \mathbf{1} \big\{  \boldsymbol{v}^{\prime} \boldsymbol{B}_z (s) > 0 \big\} = \eta,
\end{align}
where $\boldsymbol{B}_z = \left( 1, \boldsymbol{B}_x (s)^{\prime}  \right)$.  
\end{proof}
Next, we show that
\begin{align}
\underset{ m \to \infty }{ \mathsf{lim} } \underset{ n \to \infty }{ \mathsf{lim \ sup} } \ \mathbb{P} \big( \big| W_n(v) - W_{nm} (v) \big| \geq \epsilon \big) = 0. 
\end{align}
In particular, the above result holds because of the following 
\begin{align*}
\mathbb{P} 
\big( \big| W_n(v) - W_{nm} (v) \big| \geq 0 \big) 
&=
\mathbb{P} \left[ \sum_{t=1}^n \big( \boldsymbol{v}^{ \prime } \boldsymbol{D}_n^{-1} \boldsymbol{z}_t - u_t (\uptau) \big) \times \mathbf{1} \big\{ 0 < u_{t}(\uptau) < \boldsymbol{v}^{\prime} \boldsymbol{D}_n^{-1} \boldsymbol{z}_t \big\} \times \mathbf{1} \big\{ \boldsymbol{v}^{\prime} \boldsymbol{D}_n^{-1} \boldsymbol{z}_t > m \big\} > 0 \right]
\\
&
\leq \mathbb{P} \left[ \bigcup_{t=1}^n \big\{ \boldsymbol{v}^{\prime} \boldsymbol{D}_n^{-1} \boldsymbol{z}_t  > m \big\} \right]
= 
\mathbb{P} \big[ \underset{ t }{ \mathsf{max} } \big\{ \boldsymbol{v}^{\prime} \boldsymbol{D}_n^{-1} \boldsymbol{z}_t \big\} > m \big],
\end{align*}
\begin{align}
\underset{ m \to \infty }{ \mathsf{lim} } \mathbb{P} \left[ \underset{ 1 \leq t \leq 1 }{ \mathsf{sup} } \boldsymbol{v}^{\prime} \boldsymbol{D}_n^{-1} \boldsymbol{z}_t > m  \right] = 0.
\end{align}
Then, notice that by \cite{billingsley1968convergence} we have that
\begin{align}
W_{nm} (v) \Rightarrow \eta, 
\end{align}
which implies that 
\begin{align*}
&\sum_{t=1}^n \left[ \big( \boldsymbol{D}_n^{-1} \boldsymbol{v}^{\prime} \big)^{\prime} \boldsymbol{z}_t - u_t(\uptau) \right] \mathbf{1} \left\{ 0 < u_t(\uptau) < \big( \boldsymbol{D}_n^{-1} \boldsymbol{v}^{\prime} \big)^{\prime} \boldsymbol{z}_t \right\}
\Rightarrow
\frac{1}{2} f \big( F^{-1} (v) \big) \boldsymbol{v}^{\prime} \int_0^1 \boldsymbol{B}_z \boldsymbol{B}_z^{\prime} \boldsymbol{v} 
\end{align*}
Similarly, we can show that 
\begin{align*}
\sum_{t=1}^n \big( u_t (\uptau) - \big( \boldsymbol{D}_n^{-1} \boldsymbol{v}^{\prime} \big)^{\prime} \boldsymbol{z}_t  \big) \mathbf{1} \big\{ \big( \boldsymbol{D}_n^{-1} \boldsymbol{v}^{\prime} \big)^{\prime} \boldsymbol{z}_t < u_t(\uptau) < 0 \big\}
\Rightarrow
\frac{1}{2} f \big( F^{-1} (v) \big) \boldsymbol{v}^{\prime} \left[ \int_0^1 \boldsymbol{B}_z \boldsymbol{B}_z^{\prime} \right] \boldsymbol{v} \mathbf{1} \big\{ \boldsymbol{v}^{\prime} \boldsymbol{B}_z (s) < 0 \big\}. 
\end{align*}
Therefore, we have that 
\begin{align*}
&\sum_{t=1}^n \big( u_t(\uptau) - \big( \boldsymbol{D}_n^{-1} \boldsymbol{v} \big)^{\prime} \boldsymbol{z}_t \big) \times \bigg[ \mathbf{1} \big\{ \big( \boldsymbol{D}_n^{-1} \boldsymbol{v} \big)^{\prime} \boldsymbol{z}_t  <  u_t(\uptau) < 0 \big\} - \mathbf{1} \big\{ 0  <  u_t(\uptau) < \big( \boldsymbol{D}_n^{-1} \boldsymbol{v} \big)^{\prime} \boldsymbol{z}_t \big\} \bigg]
\\
&\Rightarrow 
f \big( F^{-1} (\uptau) \big) \boldsymbol{v}^{\prime}
\begin{bmatrix}
1 & \displaystyle \int_0^1 \boldsymbol{B}_x
\\
\displaystyle \int_0^1 \boldsymbol{B}_x & \displaystyle \int_0^1 \boldsymbol{B}_x \boldsymbol{B}_x^{\prime}
\end{bmatrix}
\boldsymbol{v}.
\end{align*}
As a result we obtain that 
\begin{align*}
\mathcal{Z}_n (v) 
&= 
\sum_{t=1}^n \left[ \rho_{\uptau} \big( u_t(\uptau) - \big( \boldsymbol{D}_n^{-1} \boldsymbol{v} \big)^{\prime} \boldsymbol{z}_t \big) - \rho_{\uptau} \big( u_t(\uptau) \big) \right]
- 
\sum_{t=1}^n \big( \boldsymbol{D}_n^{-1} \boldsymbol{v} \big)^{\prime} \boldsymbol{z}_t \psi_{\uptau} \big( u_t (\uptau) \big) 
\\
&\ \ + 
\sum_{t=1}^n \big( u_t(\uptau) - \big( \boldsymbol{D}_n^{-1} \boldsymbol{v} \big)^{\prime} \boldsymbol{z}_t \big) \big) \times \bigg[ \mathbf{1} \bigg\{ \big( \boldsymbol{D}_n^{-1} \boldsymbol{v} \big)^{\prime} \boldsymbol{z}_t  <  u_t(\uptau) < 0 \bigg\} - \mathbf{1} \bigg\{ 0  <  u_t(\uptau) < \big( \boldsymbol{D}_n^{-1} \boldsymbol{v} \big)^{\prime} \boldsymbol{z}_t \bigg\} \bigg]
\\
&\Rightarrow
- \boldsymbol{v}^{\prime} 
\begin{bmatrix}
\displaystyle \int_0^1 d \boldsymbol{B}_{ \psi_{\uptau} }
\\
\displaystyle \int_0^1  \boldsymbol{B}_x d \boldsymbol{B}_{ \psi_{\uptau} } + f \big( F^{-1} ( \uptau) \big) \boldsymbol{v}^{\prime} 
\end{bmatrix} 
+ 
f \big( F^{-1} (\uptau) \big)
\boldsymbol{v}^{\prime} 
\begin{bmatrix}
1 & \displaystyle \int_0^1 \boldsymbol{B}_x
\\
\displaystyle \int_0^1 \boldsymbol{B}_x & \displaystyle \int_0^1 \boldsymbol{B}_x \boldsymbol{B}_x^{\prime}
\end{bmatrix}
\boldsymbol{v}
:= \mathcal{Z}(v).
\end{align*}
Therefore, by the convexity Lemma of \cite{pollard1991asymptotics} and arguments presented by \cite{knight1998limiting}, notice that the functional $\mathcal{Z}_n(v)$ is minimized at $\widehat{v} = \boldsymbol{D}_n \big( \widehat{\boldsymbol{\theta}}(\uptau) - \boldsymbol{\theta}(\uptau) \big)$ while $\mathcal{Z}(v)$ is minimized at 
\begin{align}
\frac{1}{2 f \big( F^{-1} (\uptau) \big)}
\begin{bmatrix}
1 & \displaystyle \int_0^1 \boldsymbol{B}_x
\\
\displaystyle \int_0^1 \boldsymbol{B}_x & \displaystyle \int_0^1 \boldsymbol{B}_x \boldsymbol{B}_x^{\prime}
\end{bmatrix}^{-1}
\times
\begin{bmatrix}
\displaystyle \int_0^1 d \boldsymbol{B}_{ \psi_{\uptau} } 
\\
\displaystyle \int_0^1 \boldsymbol{B}_x d \boldsymbol{B}_{ \psi_{\uptau} }  + \lambda_{x \psi_{\uptau} }
\end{bmatrix}
\end{align}
by Lemma A of \cite{knight1989limit} we have that 
\begin{align}
\boldsymbol{D}_n \big( \widehat{ \boldsymbol{\theta}} (\uptau) - \boldsymbol{\theta}(\uptau) \big) 
\Rightarrow 
\frac{1}{2 f \big( F^{-1} (\uptau) \big)}
\begin{bmatrix}
1 & \displaystyle \int_0^1 \boldsymbol{B}_x
\\
\displaystyle \int_0^1 \boldsymbol{B}_x & \displaystyle \int_0^1 \boldsymbol{B}_x \boldsymbol{B}_x^{\prime}
\end{bmatrix}^{-1}
\times
\begin{bmatrix}
\displaystyle \int_0^1 \boldsymbol{B}_{ \psi_{\uptau} } 
\\
\displaystyle \int_0^1 \boldsymbol{B}_x d \boldsymbol{B}_{ \psi_{\uptau} }  + \lambda_{x \psi_{\uptau} }
\end{bmatrix}.
\end{align}

\newpage

Therefore, by Theorem 1, the asymptotic distribution of the random measure $n \big( \widehat{\beta}(\uptau) - \beta(\uptau) \big)$ can be expressed as below
\begin{align*}
\frac{1}{ f \big( F^{-1} (\uptau) \big) } \left[ \int_0^1 \boldsymbol{B}_v^{\mu} \boldsymbol{B}_v^{\mu \prime} \right]^{-1} \left( \int_0^1 \boldsymbol{B}^{\mu}_v  d \boldsymbol{B}_{ \psi_{\uptau}. v } \right) 
+ 
\frac{1}{ f \big( F^{-1} (\uptau) \big) } \left[ \int_0^1 \boldsymbol{B}_v^{\mu} \boldsymbol{B}_v^{\mu \prime} \right]^{-1}  \bigg[ \boldsymbol{B}_v^{\mu} d \boldsymbol{B}_v^{\prime} \boldsymbol{\Omega}_{vv}^{-1} \boldsymbol{\Omega}_{v \psi_{\uptau} } + \lambda_{ v \psi_{\uptau} } \bigg].
\end{align*} 
Furthermore, we have that 
\begin{align*}
\big[ \widehat{\boldsymbol{\Omega}}_{vv}, \ \ \ \widehat{\boldsymbol{\Omega}}_{v \psi_{\uptau} }, \ \ \ \widehat{\lambda}_{v \psi_{\uptau} }, \ \ \ \widehat{\lambda}_{vv} \big]
\end{align*}
are consistent estimates of $\boldsymbol{\Omega}_{vv}$, $\boldsymbol{\Omega}_{v \psi_{\uptau} }$, $\lambda_{v \psi_{\uptau} }$ and $\lambda_{vv}$. Thus, we obtain the following weak convergence \begin{align*}
\boldsymbol{D}_n \left( \widehat{\boldsymbol{\theta}} (\uptau)^{+} - \boldsymbol{\theta}(\uptau) \right)
&= 
\begin{pmatrix}
\sqrt{n} \big[ \widehat{\alpha}(\uptau) - \alpha(\uptau) \big] 
\\
\\
n \big[ \widehat{\boldsymbol{\beta}}^{+}(\uptau) - \boldsymbol{\beta}(\uptau) \big] 
\end{pmatrix}
\\
&=
\begin{pmatrix}
\sqrt{n} \big[ \widehat{\alpha}(\uptau) - \alpha(\uptau) \big] 
\\
\\
n \big[ \widehat{\boldsymbol{\beta}} (\uptau) - \boldsymbol{\beta}(\uptau) \big] 
\end{pmatrix}
- 
\begin{bmatrix}
0 
\\
\\
- \frac{ \displaystyle 1}{ \displaystyle f \big( \widehat{F^{-1} ( \uptau )}  \big) } \left[ \displaystyle \frac{1}{n^2} \sum_{t=1}^n \boldsymbol{x}_t \boldsymbol{x}_t^{\prime} \right]^{-1} \left[ \displaystyle \frac{1}{n} \sum_{t=1}^n \boldsymbol{x}_t\boldsymbol{v}_t^{\prime} \widehat{\boldsymbol{\Omega}}_{vv}^{-1}  \widehat{\boldsymbol{\Omega}}_{v \psi_{\uptau} } + \widehat{\lambda}_{v \psi_{\uptau} }^{+}  \right]
\end{bmatrix}
\\
&\Rightarrow
\frac{ \displaystyle 1}{ \displaystyle f \big( \widehat{F^{-1} ( \uptau )}  \big) } \left[ \int_0^1 \widetilde{\boldsymbol{B}}_v \widetilde{\boldsymbol{B}}_v^{\prime} \right]^{-1} \left( \int_0^1 \widetilde{\boldsymbol{B}}_v  d \boldsymbol{B}_{ \psi_{\uptau}. v } \right)  
\end{align*}

\medskip

\subsubsection{Inference in quantile cointegration models}
\label{SuppA3}

The asymptotic distribution of $\widehat{\beta}( \uptau )$ in quantile cointegrating regressions is mixture normal. Another interesting problem in the quantile cointegration regression model is the hypothesis test on constancy of the cointegrating vector $\boldsymbol{\beta} ( \uptau ) = \bar{\boldsymbol{\beta} }$ , over $\uptau \in \mathcal{T}_{\iota}$, where $\bar{ \boldsymbol{\beta} }$ is a vector of unknown constants. A natural preliminary candidate for testing constancy of the cointegrating vector is a standardized version of $\left( \widehat{ \boldsymbol{\beta} } - \bar{\boldsymbol{\beta}} \right)$. Under the null hypothesis, we have that 
\begin{align}
n \left( \widehat{ \boldsymbol{\beta} } ( \uptau ) - \bar{\boldsymbol{\beta}} \right) \Rightarrow \frac{1}{ f_{\epsilon} \big( F_{\epsilon}^{-1} ( \uptau ) \big)} \left[ \int_0^1 \boldsymbol{B}_v^{\mu} \boldsymbol{B}_v^{\mu \prime} \right]^{-1} \times \left[ \int_0^1 \boldsymbol{B}_v^{\mu} d \boldsymbol{B}^{*}_{ \psi_{\uptau} } \right]
\end{align}
Denote with $\widehat{\boldsymbol{\beta}}$ as a preliminary estimator of  $\bar{ \boldsymbol{\beta} }$  and consider the following process
\begin{align}
\widehat{V}_n ( \uptau ) = n \left( \widehat{ \boldsymbol{\beta} } ( \uptau ) - \bar{\boldsymbol{\beta}} \right) 
\end{align} 

\newpage

Then, under the null hypothesis $\mathcal{H}_0: \boldsymbol{\beta} (\uptau) = \bar{\boldsymbol{\beta}}$ it holds that 
\begin{align}
\underset{ \uptau \in \mathcal{T}_{\iota}  }{ \mathsf{sup} } \ \left| \widehat{V}_n ( \uptau ) \right| \Rightarrow \underset{ \uptau \in \mathcal{T}_{\iota}  }{ \mathsf{sup} } \left| \frac{1}{ f_{\epsilon} \big( F_{\epsilon}^{-1} ( \uptau ) \big)} \left[ \int_0^1 \boldsymbol{B}_v^{\mu} \boldsymbol{B}_v^{\mu \prime} \right]^{-1} \times \int_0^1 \boldsymbol{B}_v^{\mu} d \big\{ B^{*}_{ \psi_{\uptau} }  - f_{\epsilon} \big( F_{\epsilon}^{-1} ( \uptau ) \big) B^{*}_{\epsilon}  \big\} \right|
\end{align}
where $B^{*}_{\epsilon}$ is the Brownian motion limit for the partial sum process of $\epsilon_t$. Therefore, we may test varying-coefficient behaviour based on the KS statistic. Notice that for the matrix $\boldsymbol{V}_{xx} = 
\displaystyle \int_{0}^{\infty} e^{r \boldsymbol{C}_p} \boldsymbol{\Omega}_{xx} e^{r \boldsymbol{C}_p} dr$, and thus by applying integration by parts we can obtain the following formula $\boldsymbol{C} \boldsymbol{V}_{xx} + \boldsymbol{V}_{xx} \boldsymbol{C} = - \boldsymbol{\Omega}_{xx}$ (see, for example expression (49) in \cite{magdalinos2009econometric}). 

In general, lets suppose we have that 
\begin{align*}
u_t (\uptau) = y_t - \alpha_0 (\uptau) - g_t(\beta)
\end{align*}
Define with 
such that $u_t(\uptau) = u_t - F^{-1}(\uptau)$. Then, for the error terms of the quantile predictive regression model the following invariance principles hold 
\begin{align}
U_n^{\psi} ( \lambda, \uptau ) 
= 
\frac{1}{ \sqrt{n} } \sum_{t=1}^{ \floor{ \lambda n } } \psi_{\uptau} \big( u_t (\uptau) \big)
\ \ \ 
V_n( \lambda ) 
=
\frac{1}{ \sqrt{n} } \sum_{t=1}^{ \floor{ \lambda n } } v_{t+1}
\end{align}
In addition, suppose that the following conditions hold: 

\begin{itemize}

\item[\textit{\textbf{(i)}}] $\bigg\{ \psi_{\uptau} \big( u_t (\uptau) \big), \mathcal{F}_t \bigg\}$ is a martingale difference sequence. 

\item[\textit{\textbf{(ii)}}] $\left\{ x_t \right\}$ is adapted to the filtration $\mathcal{F}_{t-1}$ and, for all $0 < \lambda < 1$, then the vector $\big( U_n^{\psi} ( \lambda, \uptau ) , V_n ( \lambda ) \big)$ converges weakly to a two-dimensional vector Brownian motion $\big( U^{\psi} ( \lambda, \uptau ), V ( \lambda ) \big)$ with a covariance matrix given by the following expression
\begin{align}
\lambda 
\begin{bmatrix}
\omega_{\psi}(\uptau)^2 & \omega_{\psi v}(\uptau)
\\
\omega_{v \psi}(\uptau) & \omega_{v}(\uptau)^2 
\end{bmatrix}
\end{align}
\end{itemize}
Based on the martingale difference sequence assumption on $u_t$, it follows that $U_n^{\psi} ( \lambda, \uptau ) \overset{ d }{ \to } U^{\psi} ( \lambda, \uptau )$, where $U^{\psi} ( \lambda, \uptau )$ is viewed as a Brownian motion with variance $\lambda \omega_{\psi}(\uptau)^2 = \lambda \uptau ( 1 - \uptau)$ for a fixed value of $\uptau$. Thus, for each fixed pair $( \lambda, \uptau )$, it holds that $U^{\psi} ( \lambda, \uptau )$ is distributed as $\mathcal{N} \big( 0, \lambda \omega_{\psi}(\uptau)^2  \big)$. A standard assumption in the nonstationary time series analysis is that $V_n$ converge weakly jointly with $U_n^{\psi}$ to a vector Brownian motion. Notice that for the development of the asymptotic theory we keep the argument $\uptau$ fixed (see, also \cite{cho2015quantile}). Another example is the framework proposed by \cite{cho2015quantile} who discuss quantile cointegration in ADL models. In that case, it holds
\begin{align}
\frac{1}{n} \sum_{t=1}^n \psi_{\uptau} \big( u_t(\uptau) \big) \boldsymbol{X}_t 
\Rightarrow
\int_0^1 \boldsymbol{B}_w(r) d \boldsymbol{B}_{\psi}( r, \uptau)  
\end{align}

\newpage 

\subsection{Structural Break Testing for  Quantile Regressions}

\begin{proposition}[\cite{aue2017piecewise}]
Let $F_{t-1} = \mathbb{P} \big( y_t < . | \mathcal{F}_{t-1} \big)$ be the conditional distribution function of $y_t$ given $\mathcal{F}_{t-1}$ and denote by $f_{t-1}$ its derivative. Under stationarity and if $f_{t-1}$ is uniformly integrable on $\mathcal{X} = \left\{ x : 0 < F(x) < 1 \right\}$, then 
\begin{align}
\boldsymbol{\Sigma}^{- 1 / 2} n^{1/ 2} \left[ \hat{\theta} (.) - \theta(.) \right] \overset{ \mathcal{D} }{ \to } \boldsymbol{B}_{p+1} (.), \ \ \text{as} \ n \to \infty,
\end{align}
where $\boldsymbol{\Sigma} = \boldsymbol{\Omega}_1^{-1} \boldsymbol{\Omega}_0 \boldsymbol{\Omega}_1^{-1}$ with $\boldsymbol{\Omega}_0 = \mathbb{E} \left( \boldsymbol{X}_t \boldsymbol{X}_t^{\prime} \right)$ and $\boldsymbol{\Omega}_1 =   \underset{ n \to \infty  }{ \mathsf{lim} } \sum_{t=1}^n f_{t-1} \left\{ F_{t-1}^{-1} \left( \uptau \right) \right\} \boldsymbol{X}_t \boldsymbol{X}_t^{\prime}$. Moreover, $\big( \boldsymbol{B}_{p+1}(\uptau) : \uptau \in [0,1] \big)$ is a standard $( p + 1)-$dimensional Brownian bridge. 
\end{proposition}
Therefore, one may be interested to develop a framework for detecting multiple breaks in quantile piecewise regression models or quantile cointegrated regressions. For instance if the number of break points $m$ is given, then estimating their locations and the $(m + 1)$ piecewise quantile autoregressive models at a specific quantile $\uptau \in (0,1)$ can be done via solving 
\begin{align}
\underset{ \theta ( \uptau ), \mathcal{K} }{ \mathsf{min} } \sum_{ j=1 }^{ m+1 } \sum_{ t = k_{j-1} + 1 }^{ k_j } \uprho_{ \uptau } \big( y_t - X_{j,t}^{\prime} \theta_j ( \uptau ) \big). 
\end{align} 
Further related studies on structural break testing for quantile regressions include \cite{qu2008testing} and \cite{fanchange2023} within a stationary time series environment and \cite{katsouris2023structural} within a nonstationary time series environment. Although both of these streams of literature correspond to structural break testing on quantile-dependent coefficients within a full-sample period. On the other hand, a different application would correspond to an implementation of the monitoring framework within based on a quantile regression, in the case of risk measures (see the recent study by  \cite{hoga2023monitoring}). 

The next two sections cover briefly key results from two additional relevant topics for quantile regression models, that is,  (i) estimation in high dimensional quantile regressions and, (ii) specification testing in quantile regressions. Therefore, developing powerful tests for the correct specification\footnote{Recall that an omnibus test is a statistical test for which the alternative hypothesis is a negation of the null hypothesis. For example, the particular family of tests allows to test the monotonicity of conditional moments. On the other hand, a different approach in the literature, that is, the unconditional backtesting can be decomposed into an estimation risk component and a model risk component which implies that under correct specification of the parametric VaR model, almost surely, the model risk vanishes (see, \cite{escanciano2010backtesting}). In contract, under misspecification, the model risk does not vanish and has a negligible effect on the unconditional test.} of parametric conditional quantiles over a possibly continuous range of quantiles and under general conditions on the underlying data-generating process is a good model validity practice. The main idea of regression model misspecification is explained by  \cite{stinchcombe1998consistent} as below: 

Let $\mathcal{S} :=  \left\{ f(., \theta ) : \mathbb{R}^k \to \mathbb{R} | \theta \in \Theta \right\}$, $ \Theta \subset \mathbb{R}^p, p \in \mathbb{N}$. The model is correctly specified for $\mathbb{E} \left( Y | X \right)$ when $f( X , \theta_0 )$ is a version of $\mathbb{E} \left( Y | X \right)$ for some $\theta_0 \in \Theta$. For instance, $\hat{\theta}_n$ can be a quantile regression estimator. Denote with $\epsilon:= Y - f(X, \theta^{*} )$, the correct specification of $\mathcal{S}$ for $\mathbb{E} \left( Y | X \right)$ is equivalent to $e := \mathbb{E} \left( \epsilon | X \right) = 0$, almost surely. Under suitable conditions on $Y$ and $f$, $e$ is an integrable function of $X$, that is, $e \in L^p ( X ) := L^p ( \Omega, \sigma(X), P )$ for some $p \in [1, \infty]$.


\newpage

\section{High Dimensional Quantile Regression Applications}

\subsection{High Dimensional Quantile Regression}

Following the framework of \cite{belloni2011} consider a response variable $y$ and $p-$dimensional covariates $x$ such that the $u-$th conditional quantile function of $y$ given $x$ is 
\begin{align}
F_{ y_i | x_i }^{-1} ( u| x_i ) = x^{\prime} \beta(u), \ \ \ \beta (u) \in \mathbb{R}^p \ \ \ \text{for all} \ \ u \in \mathcal{U},     
\end{align}
where $\mathcal{U} \subset (0,1)$ is a compact set of quantile indices. Recall that the $u-$th conditional quantile $F_{ y_i | x_i }^{-1} ( u| x_i )$ is the inverse of the conditional distribution function $F_{ y_i | x_i }^{-1} ( y | x_i )$ of $y_i$ given $x_i$. Furthermore, we consider the case where the dimension $p$ of the model is large, possibly much larger than the available sample size $n$, but the true model $\beta (u)$ has a sparse support
\begin{align}
T_u = \mathsf{support} \left(  \beta(u)  \right) = \left\{ j \in \left\{ 1,..., p \right\} : \left| \beta_j(u) \right| > 0 \right\}    
\end{align}
having only $s_u \leq s \leq n / \mathsf{log} ( n \cup p )$ nonzero components for all $u \in \mathcal{U}$. Therefore, the corresponding population coefficient $\beta (u)$ is known to minimize the criterion function 
\begin{align}
\mathcal{Q}_u := \mathbb{E} \big[ \rho_u \left( y - x^{\prime} \beta \right) \big]
\end{align}
In other words, given a random sample $\left\{ (y_1, x_1),..., (y_n, x_n)     \right\}$, the quantile regression estimator of $\beta(u)$ is defined as a minimizer of the empirical analoge given by 
\begin{align}
\widehat{\mathcal{Q}}_u := \mathbb{E} \left[ \rho_u \left( y - x_i^{\prime} \beta \right) \right].
\end{align}
The main challenge of the statistical problem under examination is that in high-dimensional settings, particularly when $p \geq n$, ordinary quantile regression is generally inconsistent, which motivates the use of penalization in order to remove all, or at least nearly all, regressors whose populations coefficients are zero, thereby possibly restoring consistency. Then, the $\ell_1-$penalized quantile regression estimator $\widehat{\beta}(u)$ is a solution to the following optimization problem:
\begin{align}
\underset{ \beta in \mathbb{R}^p  }{ \mathsf{min} }   \ \widehat{\mathcal{Q}}_u (\beta) + \frac{ \lambda \sqrt{u (1 - u)} }{n} \sum_{j=1}^p \widehat{\sigma}_j \left| \beta_j  \right| 
\end{align}
where $\hat{\sigma}^2_j = \mathbb{E}_n \big[ x_{ij}^2 \big]$, which ensures that the conditional variance of the error term has a bounded variance and thus excluding infitine variance cases as the related distribution function of the innovation sequence generating the data mechanism under examination. To show the result of consistency, it suffices to show that for any $\epsilon > 0$, there exists a sufficiently large $C$ such that 
\begin{align}
\mathbb{P} \left( \underset{ \norm{c}_2 = C  }{ \mathsf{inf} } \mathcal{Q}_n^{qr} \big( \beta ( \uptau ) + a_n c \big) > \mathcal{Q}_n^{qr} \big( \beta ( \uptau ) \big) \right) \geq 1 - \epsilon.
\end{align}

\newpage 

In other words, this inequality implies that with probability at least $1 - \epsilon$, there is a local minimizer $\tilde{\beta} ( \uptau )$ within the shrinking ball $\big\{  \beta ( \uptau ) + a_n c,  \norm{c}_2 = C \big\}$ such that $\norm{ \tilde{\beta} ( \uptau ) - \beta ( \uptau ) }_2 = \mathcal{O}_p ( a_n )$. Therefore, the proof can be obtained by showing that the following term is positive
\begin{align*}
\mathcal{Q}_n^{qr} \big( \beta ( \uptau ) + a_n c \big) - \mathcal{Q}_n^{qr} \big( \beta ( \uptau ) \big) = \sum_{t=1}^n \rho_{ \uptau } \big( u_t ( \uptau ) - x_{t-1}^{\prime} a_n c \big) - \sum_{t=1}^n \rho_{ \uptau } \big( u_t ( \uptau ) \big)
\end{align*}

\paragraph{Proof.}

\begin{align*}
Var ( I_3 ) 
&=
Var \left( \sum_{t=1}^n \int_0^{ x_{t-1}^{\prime} a_n c } \bigg[ \mathbf{1} \big( u_t ( \uptau ) \leq s \big) -  \mathbf{1} \big( u_t ( \uptau ) \leq 0 \big) \bigg] ds \right)
\\
&\leq 
\mathbb{E} \left[ \sum_{t=1}^n \int_0^{ x_{t-1}^{\prime} a_n c } \bigg[ \mathbf{1} \big( u_t ( \uptau ) \leq s \big) -  \mathbf{1} \big( u_t ( \uptau ) \leq 0 \big) \bigg] ds  \right]
\\
&= 
a_n^2 \sum_{t=1}^n c^{\prime} \mathbb{E} \big[ x_{t-1} x_{t-1}^{\prime} \big] c + 2 a_n^2 \sum_{t=2}^n \sum_{k=1}^{t-1} \mathbb{E} \big[ \left| x_{t-1}^{\prime} c \right| \left| x_{t-1}^{\prime} c \right| \big]
\\
&\equiv V_{3,1} + V_{3,2}.
\end{align*}
Therefore, using the similar arguments in $I_1$, we have that
\begin{align}
V_{3,1} \leq a_n^2 C^2 \frac{ \bar{c}_{A(n,p)} }{ c_f } \frac{ n(n+1) }{2} = \mathcal{O}_p \left( a_n^2 \bar{c}_{A(n,p)} n^2 \right).
\end{align}
Consequently, by the Cauchy-Schwarz inequality, and $t > k$, we have that 
\begin{align}
V_{3,2} \leq 2 a_n^2 \sum_{t=2}^n \sum_{k=1}^{ t - 1 } \sqrt{ \mathbb{E} \left[ \left( x_{t-1}^{\prime} x_{t-1}  \right) \right] } \sqrt{ \mathbb{E} \left[ \left( x_{t-1}^{\prime} x_{t-1}  \right) \right] }
\end{align}
\begin{align*}
\mathbb{E}  \left| I_1 \right| ^2 
&= 
a_n^2 \sum_{t=1}^n c^{\prime} M_n^{\prime} \mathbb{E} \left[ \psi_{ \uptau } \big( u_t ( \uptau ) \big)^2 \tilde{x}_{t-1} \tilde{x}_{t-1}^{\prime}  \right] M_n c 
+ 
2 a_n^2 \sum_{t=2}^n \sum_{k=1}^{t-1} c^{\prime} M_n^{\prime} \mathbb{E} \left[ \psi_{ \uptau } \big( u_t ( \uptau ) \big)^2 \tilde{x}_{t-1} \tilde{x}_{t-1}^{\prime}  \right] M_n c  
\\
&=
a_n^2 \sum_{t=1}^n c^{\prime} M_n^{\prime} \mathbb{E} \left[ \psi_{ \uptau } \big( u_t ( \uptau ) \big)^2 \tilde{x}_{t-1} \tilde{x}_{t-1}^{\prime}  \right] M_n c
\\
&=
a_n^2 c^{\prime} M_n^{\prime} \sum_{t=1}^n \mathbb{E} \left[ \psi_{ \uptau } \big( u_t ( \uptau ) \big)^2 \tilde{x}_{t-1} \tilde{x}_{t-1}^{\prime}  \right] M_n c
\\
&\leq 
a_n^2  n^2_{ \bar{c}_{ B(n,p) }  } C^2  
\end{align*}
Further applications of high dimensional quantile time series regression models are studied by \cite{belloni2023high} while in a cross-sectional setting relevant frameworks are proposed by \cite{he2013quantile}, \cite{carlier2016vector}, \cite{he2021smoothed}, \cite{lee2023complete} and  \cite{zhang2023bootstrap}.

\newpage

\begin{example}[\cite{lee2023complete}]
We consider the framework of  \cite{lee2023complete} who develops asymptotic theory for Complete Subset Averaging in quantile regressions. 

\medskip

Denote with (see, page 11 in \cite{lee2023complete})
\begin{align}
z_{(m,k)} = x^{\prime}_{(m,k)} \Theta^{*}_{(m,k)} - \mathbb{E} \left[   x^{\prime}_{(m,k)} \Theta^{*}_{(m,k)} \right]
\end{align}
Depending on the dependence structure of $z_{(m,k)}$ a corresponding uniform law of large numbers hold. For example, consider the following maximal inequality for $\delta > 0$,
\begin{align*}
\mathbb{P} \left( \underset{ 1 \leq k \leq K }{ \text{max} } \left| M^{-1} \sum_{m=1}^M z_{(m,k)} \right| > \delta \right) 
&\leq 
K \underset{ 1 \leq k \leq K }{ \text{max} } \ \mathbb{P} \left(  \left| M^{-1} \sum_{m=1}^M z_{(m,k)} \right| > \delta \right)
\\
&\leq
\frac{K}{M} \underset{ 1 \leq k \leq K }{ \text{max} }  \frac{ \displaystyle \mathbb{E} \left[ \sum_{m=1}^M z_{m,k} \right]^2 }{ \displaystyle  M \delta^2 }
\end{align*}
where the second line holds from the Markov inequality. Moreover, since $K / M = o_p(1)$, a sufficient condition for the uniform convergence is $\text{max}_{ 1 \leq k \leq K} \mathbb{E} \left[ \sum_{m=1}^M z_{m,k} \right]^2 / M = O_p(1)$. Therefore, if $z_{(m,k)}$ is covariance stationary over $m$ for all $k$, then the sufficient condition becomes the absolute summability condition $\text{max}_{ 1 \leq k \leq K} \sum_{ j = 0}^{ \infty } \big| \mathbb{E} \left[ z_{m,k} z_{m+j,k} \right] \big| < \infty$. 

Denote with $v_{i(m,k)} := \norm{ x_{i(m,k)} } - \mathbb{E} \left[ x_{i(m,k)} \right]$. Notice that $\text{Var} \left( v_{i(m,k)}  \right) \leq CK$ for some generic constant $C > 0$. Let $e_n := \left( nMK^2 \right)^{1 / 4}$. We have that 
\begin{align*}
\mathbb{P} \left( A_2 \geq 2 \epsilon \right) 
&= 
\mathbb{P} \left(   \underset{ 1 \leq k \leq K }{ \text{max} } \  \underset{ 1 \leq m \leq M }{ \text{max} } \ \left| \frac{1}{n} \sum_{i=1}^n v_{i(m,k)}   \right| \geq 2 \epsilon \right)
\\
&\leq 
\mathbb{P} \left(   \underset{ 1 \leq k \leq K }{ \text{max} } \  \underset{ 1 \leq m \leq M }{ \text{max} } \  \frac{1}{n} \sum_{i=1}^n \left| v_{i(m,k)}   \right| \geq 2 \epsilon \right)
\\
&\leq 
\mathbb{P} \left(   \underset{ 1 \leq k \leq K }{ \text{max} } \  \underset{ 1 \leq m \leq M }{ \text{max} } \  \frac{1}{n} \sum_{i=1}^n \left| v_{i(m,k)} \right| \mathbf{1} \left\{ \left| v_{i(m,k)} \right| \leq e_n \right\} \geq \epsilon \right)
\\
&+ \mathbb{P} \left(   \underset{ 1 \leq k \leq K }{ \text{max} } \  \underset{ 1 \leq m \leq M }{ \text{max} } \  \frac{1}{n} \sum_{i=1}^n \left| v_{i(m,k)} \right| \mathbf{1} \left\{ \left| v_{i(m,k)} \right| > e_n \right\} \geq \epsilon \right)
\equiv A_{21} + A_{22}.
\end{align*}
By Boole's Bernestein inequalities we have that 
\begin{align*}
A_{21} 
&\leq 
KM \underset{ 1 \leq k \leq K }{ \text{max} } \ \underset{ 1 \leq m \leq M }{ \text{max} } \mathbb{P} \left(   \frac{1}{n} \sum_{i=1}^n \left| v_{i(m,k)} \right| \mathbf{1} \left\{ \left| v_{i(m,k)} \right| \leq e_n \right\} \geq \epsilon \right)
\\
&\leq 2KM \text{exp} \left\{ - \frac{n \epsilon^2 }{ 2CK + 2 \epsilon e_n / 3} \right\} 
\end{align*}

\end{example}
Some further applications of quantile time series regressions are discussed in \cite{felix2023some}.

\newpage

\subsection{Smoothed Quantile Regression with Large-Scale Inference}

Following the framework of \cite{he2021smoothed}, consider a univariate response variable $y \in \mathbb{R}$ and a $p-$dimensional covariate vector, the primary goal here is to learn the effect of $\boldsymbol{x}$ on the distribution of $y$. Let $F_{y| \boldsymbol{x}}$ be the conditional distribution function of y given  $\boldsymbol{x}$. The dependence between $y$ and $\boldsymbol{x}$ is then fully characterized by the conditional quantile functions of $y$ given $\boldsymbol{x}$, denoted as $F^{-1}_{y| \boldsymbol{x}}$, for $0 < \uptau < 1$. 

Consider a linear quantile regression model at a given $\uptau \in (0,1)$, that is, the $\uptau-$the conditional quantile function is
\begin{align}
F^{-1}_{y| \boldsymbol{x}} = \langle \boldsymbol{x}, \boldsymbol{\beta}_0(\uptau) \rangle
\end{align}
where $\boldsymbol{\beta}_0(\uptau) = \big( \beta_1^{*},..., \beta_p^{*} \big)^{\prime} \in \mathbb{R}^p$ is the true quantile regression coefficient. 

\subsubsection{Smoothed Estimation Equation and Convolution-type smoothing}

Let $Q( \boldsymbol{\beta} ) = \mathbb{E} \big[ \widehat{Q}(\boldsymbol{\beta}) \big]$ be the population quantile loss function.  Under mild conditions, $Q(.)$ is twice differentiable and strongly convex in a neighbourhood of $\boldsymbol{\beta}$ with Hessian matrix such that
\begin{align}
\boldsymbol{J} := \nabla^2 Q( \boldsymbol{\beta}^{*} ) = \mathbb{E} \big[ f_{\boldsymbol{\epsilon | \boldsymbol{x} } }(0) \boldsymbol{x}  \boldsymbol{x}^{\prime} \big], \ \ \text{where} \ \ \epsilon = y - \langle \boldsymbol{x}, \boldsymbol{\beta}^{*} \rangle
\end{align} 
is the random noise and $f_{\boldsymbol{\epsilon | \boldsymbol{x} } }(.)$ is the conditional density of $\epsilon$ given $\boldsymbol{x}$. 

On the other hand, by the first-order condition the population parameter $\boldsymbol{0}$ satisfies the moment condition 
\begin{align}
\nabla Q( \boldsymbol{\beta} ) = \mathbb{E} \big[ \uptau - \mathbf{1} \left\{ y < \boldsymbol{x}^{\prime} \boldsymbol{\beta} \right\} \big] \big|_{ \boldsymbol{\beta} = \boldsymbol{\beta}_0} 
\end{align}

\begin{example}[see, \cite{he2021smoothed}]
For a standard QR with fixed design, Theorem 2 in Belloni et al (2019) we write below: 
\begin{align}
\sqrt{n} \left( \widehat{\boldsymbol{\beta}}_h - \boldsymbol{\beta}_0   \right)
= 
\boldsymbol{J}_h^{-1} \boldsymbol{U} + \boldsymbol{r}
\end{align}
where
\begin{align}
\boldsymbol{U} 
&= 
\frac{1}{ \sqrt{n} } \sum_{i=1}^n \big[ \uptau - \mathcal{K}_h ( - \epsilon_i ) \big] \boldsymbol{x}_i 
- 
\mathbb{E} \big[ \uptau - \mathcal{K}_h ( - \epsilon_i ) \big] \boldsymbol{x}_i
\\
\norm{ \boldsymbol{r} }_2 &= \mathcal{O}_p \left( p^{3/4} \zeta_p ( \mathsf{log} n )^{1/2} n^{- 1 / 4} \right)
\end{align}
such that $\zeta_p = \underset{ \boldsymbol{x} \in \mathcal{X} }{ \mathsf{sup} } \norm{ \boldsymbol{x} }_2$. From an asymptotic perspective, the QR estimator has the advantage of being (conditionally) pivotal asymptotically. 
\end{example}

\newpage 

Moreover, the Bahadur representation can be used to establish the limiting distribution of the estimator or its functionals. We consider a fundamental statistical inference problem for testing the linear hypothesis $\mathcal{H}_0: \langle \boldsymbol{a}, \boldsymbol{b}_0 \rangle = 0$, where $\boldsymbol{a} \in \mathbb{R}^p$ is a deterministic vector that satisfies that defines a linear functional of interest. It is then natural to consider a test statistic that depends on $\sqrt{n} \langle \boldsymbol{a}, \widehat{\boldsymbol{b}}_h \rangle$. Based on the non-asymptotic result (finite sample) of the following theorem, we establish a Berry-Esseen bound for the linear projection of the conquer estimator. 

\medskip

\begin{theorem}[\cite{he2021smoothed}]
Assume that the conditions in Theorem 3.2 hold. Then, 
\begin{align*}
\Delta_{n,p}(h) 
:= 
\underset{ x \in \mathbb{R}, \boldsymbol{a} \in \mathbb{R}^p }{ \mathsf{sup} } \ \left| \mathbb{P} \left( \sqrt{n} \sigma_h^{-1} \langle \boldsymbol{a},  \widehat{\boldsymbol{b}}_h -  \boldsymbol{b}_0 \rangle \leq x \right) - \Phi(x) \right| \leq \frac{ p + \mathsf{log}(n) }{ (nh)^{1/2} } + n^{1/2} h^2, 
\end{align*}
\end{theorem}
where $\sigma_h^2 = \sigma_h^2(\boldsymbol{a}) = \boldsymbol{a}^{\prime} \boldsymbol{J}_h^{-1} \mathbb{E} \big[ \big\{ \mathcal{K}_h ( - \epsilon ) - \uptau \big\}^2 \boldsymbol{x} \boldsymbol{x}^{\prime} \big] \boldsymbol{J}_h^{-1} \boldsymbol{a}$, where $\Phi(.)$ denotes the standard normal distribution function. Moreover, it holds that 
\begin{align}
\underset{ \boldsymbol{a} \in \mathbb{R}^p }{ \mathsf{sup} } \ \left|  \frac{ \sigma_h^2 (\boldsymbol{a}) }{ \boldsymbol{a}^{\prime} \boldsymbol{J}_h^{-1} \boldsymbol{\Sigma} \boldsymbol{J}_h^{-1} \boldsymbol{a}^{\prime} } - \uptau( 1 - \uptau) \right| 
= \mathcal{O}(h), \ \ \text{as} \ h \to 0.
\end{align}
An advantage of convolution smoothing is that it facilitates conditional density estimation for the quantile regression process. Assume that $Q_y ( \uptau | \boldsymbol{\mathcal{X}} ) = F^{-1}_{y| \boldsymbol{\mathcal{X}} } = \langle \boldsymbol{\mathcal{X}}, \boldsymbol{\beta}_0 (\uptau) \rangle$ for all $\uptau \in (0,1)$. Under mild regularity conditions, 
\begin{align}
q_y \big( \uptau | \boldsymbol{\mathcal{X}} \big) 
:= 
\displaystyle \frac{ \partial  }{ \partial \uptau } Q_y ( \uptau | \boldsymbol{\mathcal{X}} ) 
= 
\frac{1}{ f_{ y | \boldsymbol{\mathcal{X}} } \big( \langle \boldsymbol{\mathcal{X}}, \boldsymbol{\beta}_0 (\uptau) \rangle     \big) }
= \boldsymbol{0}.
\end{align}
In particular, the inverse conditional density function plays an important role in, for example, the study of quantile treatment effects through modelling inverse prosperity scores. Therefore, by the linear conditional quantile model assumption, we have that 
\begin{align}
\displaystyle \frac{ \partial }{ \partial \uptau } Q_y ( \uptau | \boldsymbol{\mathcal{X}} ) 
=
\langle \boldsymbol{\mathcal{X}}, \frac{ \partial }{ \partial \uptau } \boldsymbol{\beta}_0 (\uptau) \rangle 
\end{align}
Recall that the conquer estimator $\hat{\boldsymbol{\beta}}_h(\uptau)$ satisfies the first-order condition such that $\nabla \widehat{Q}_h \big( \widehat{\boldsymbol{\beta}}_h(\uptau) \big) = \boldsymbol{0}$.  Therefore, taking the partial derivative with respect to $\uptau$ on both sides, it follows by the chain rule that 
\begin{align}
\frac{ \partial }{ \partial \uptau } \widehat{\boldsymbol{\beta}}_h(\uptau) 
&= 
\bigg\{ \nabla^2 \widehat{Q}_h \big( \widehat{\boldsymbol{\beta}}_h(\uptau) \big) \bigg\}^{-1} \frac{1}{n} \sum_{i=1}^n \boldsymbol{\mathcal{X}}_i
\nonumber
\\
&=
\left\{ \frac{1}{n} \sum_{i=1}^n K_h \big( y_i - \boldsymbol{\mathcal{X}}_i^{\prime} \widehat{\boldsymbol{\beta}}_h(\uptau) \big) \boldsymbol{\mathcal{X}}_i  \boldsymbol{\mathcal{X}}_i^{\prime}  \right\}^{-1} \frac{1}{n} \sum_{i=1}^n \boldsymbol{\mathcal{X}}_i 
\end{align}
Consequently, the inverse functions $1 / f_{y_i | \boldsymbol{x}_i} \big( \boldsymbol{x}_i \boldsymbol{\beta}_0(\uptau) \big)$ can be directly estimated by $\boldsymbol{x}_i^{\prime} \frac{ \partial \widehat{\boldsymbol{\beta}}_h(\uptau) }{  \partial \uptau }$.

\newpage 

\section{Specification Testing in Quantile Regression Models}

\subsection{Specification testing for parametric conditional quantile functional form}

In this section, we follow carefully the framework proposed by \cite{escanciano2010specification} which corresponds to the class of consistent parametric specification tests for dynamic conditional quantiles. A recent relevant framework is proposed by \cite{horvath2022consistent}. 
The conditional moment restriction test approach considers directly the model estimates (for the corresponding model functional form we impose), which is a more flexible method allowing to consider e.g., different estimation methods (parametric versus nonparametric such as the NW estimation method or the nonparametric series estimation).

Let $W_{t-1} = \left( Y_{t-1}, Y_{t-2}, Z^{\prime}_{t-1}, Z^{\prime}_{t-2}   \right)$ be the set of variables included in the quantile regression model for estimating the CoVaR and VaR risk measures. Assuming that the conditional distribution of $Y_t$ given $W_{t-1}$ is continuous, we define the $a-$th conditional VaR of $Y_t$ given $W_{t-1}$ as the $\mathcal{F}_{t-1}-$measurable function $q_{\alpha}\left( W_{t-1} \right)$ satisfying the following probability statement
\begin{align}
\label{statement1}
\mathbb{P} \bigg(  q_{\alpha}\left( W_{t-1} \right) \big| W_{t-1} \bigg) = \alpha \ \ \textit{almost surely} \ , \alpha \in (0,1) \ \forall \ t \in \mathbb{Z}.
\end{align}   
For parametric VaR inference, one assumes the existence of a parametric family of functions 
\begin{align*}
\mathcal{M} = \bigg\{ m_{\alpha} \left( ., \theta \right) : \theta \in \Theta \subset \mathbb{R}^p \bigg\}
\end{align*}
and proceeds to make VaR forecasts using the model $\mathcal{M}$. Statistical inference is based on the crucial assumption that $q_{\alpha} \in \mathcal{M}$, which implies that there exists some $\theta_0 \in \Theta$ such that $m_{\alpha} \left( W_{t-1}, \theta_0 \right) = q_{\alpha} \left( W_{t-1} \right)$ \textit{almost surely}. For instance, in parametric model the nuisance parameter $\theta_0$ belongs to $\Theta$, with $\Theta$ a compact set in a Euclidean space $\mathbb{R}^p$. More specifically, parametric VaR models are a commonly used methodology in the literature, since the the functional form $m_{\alpha} \left( W_{t-1}, \theta_0 \right)$ along with the parameter $\theta_0$ provides a parsimonious representation of the VaR risk measure based on the available information set of the decision maker at the particular time period.     
 
\begin{remark}
According to \cite{escanciano2010specification} a direct implication of the \textit{almost surely} convergence statement given by expression \eqref{statement1} is that a parametric VaR model $m_{\alpha} \left( W_{t-1}, \theta_0 \right)$ is correctly specified if and only if the following condition holds:
\begin{align}
\label{statement2}
\mathbf{E} \big( \mathcal{E}_{t, \alpha} \left( \theta_0 \right) \big| W_{t-1}   \big) = \alpha \ \ \textit{almost surely} \ \ \text{for some} \ \theta_0 \in \Theta,
\end{align}  
where we define $\mathcal{E}_{t, \alpha} \left( \theta \right) := \mathbf{1} \left\{ Y_t \leq  m_{\alpha} \left( ., \theta \right) \right\}$, with $\theta \in \Theta$ and $\textbf{1} (A)$ is the indicator function implying that $\textbf{1} \left\{ A \right\} \equiv 1$ , if the event A occurs and 0 otherwise. In other words, the conditional moment condition implying testing that the \textit{almost surely} condition given by expression \eqref{statement2} holds, which is the main purpose of using specification testing for statistical validity of conditional quantiles as the ones proposed by  \cite{escanciano2010specification} and \cite{horvath2022consistent} (see, also \cite{delgado2010distribution}).    
\end{remark}

\newpage

\subsubsection{Tests statistics and Asymptotic Theory}

Here we present some examples which demonstrate the choice of the family of models required for the parametric estimation of the VaR. A suitable model is the linear quantile regression (LQR)
\begin{align}
m \big( \mathcal{I}_{t-1}, \theta_0( \alpha ) \big) 
\equiv 
m \big(  W_{t-1}, \theta_0( \alpha ) \big)   
\end{align}
Following the framework of \cite{escanciano2010specification} we can obtain some insights regarding the specification testing methodology. In particular, the proposed test statistics are based on the fact that $q \in \mathcal{M}$ is characterized by the infinite set of conditional moment restrictions given as below
\begin{align}
\label{statement3}
\mathbb{E} \bigg[ \mathbf{1}\bigg\{ Y_t \leq m \big( \mathcal{I}_{t-1}, \theta_0 (\alpha) \big) \bigg\}  - \alpha \bigg| \mathcal{I}_{t-1} \bigg] 
= 0 \ \ \textit{almost surely} \ \text{for} \ \ \theta_0 \in \Theta \subset \mathbb{R}^p.  
\end{align}   
Note that the moment condition given by expression \eqref{statement3} is equivalent to the condition of expression \eqref{statement2}. Therefore, the test statistic and hypothesis testing of interest is
\begin{align}
\mathbb{H}_0 : \mathbb{E} \bigg[ \Psi_{\alpha} \bigg( Y_t \leq m \big( \mathcal{I}_{t-1}, \theta_0 (\alpha) \big) \bigg)  - \alpha \bigg| \mathcal{I}_{t-1} \bigg] 
= 0 \ \ \textit{a.s} \ \text{for some} \ \theta_0
\end{align}
against the nonparametric alternatives given by 
\begin{align}
\mathbb{H}_1: \mathbb{P} \bigg( \mathbb{E} \bigg[ \Psi_{\alpha} \bigg( Y_t \leq m \big( \mathcal{I}_{t-1}, \theta_0 (\alpha) \big) \bigg)  - \alpha \bigg| \mathcal{I}_{t-1} \bigg] \neq 0  \bigg) > 0 \ \ \textit{a.s} \ \forall \ \theta \left( \alpha \right) \in \Theta \subset \mathbb{R}^p
\end{align}
where $\Psi \left( \epsilon \right) = \mathbf{1} \left( \epsilon \leq 0 \right) - \alpha$. Therefore, to simplify notation we denote with $\Psi_{ \alpha, t } \left( \theta  \right) \equiv \Psi_{ \alpha } \big( Y_t - m \left( \mathcal{I}_{t-1}, \theta  \right)  \big)$ and $m_{t-1} \left( \theta \right) \equiv m \left( \mathcal{I}_{t-1}, \theta \right)$. Thus, under the null hypothesis, and assuming that a continuity condition for $m( . )$ holds, then $m_{t-1} \left( \theta_0 \right)$ is identified as the $\alpha-$th quantile of the conditional distribution of $Y_t$ given $\mathcal{I}_{t-1}$, for all $\alpha$. Then, to characterize $\mathbb{H}_0$ by the infinite number of unconditional moment restrictions we write
\begin{align}
\mathbb{H}_0 : \mathbb{E} \bigg[ \Psi_{\alpha,t} \left( \theta_0 \right) \text{exp} \left( i x^{\prime} \mathcal{I}_{t-1} \right) \bigg] = 0 \ \ \forall x \in \mathbb{R}^d, \ \text{for some} \ \theta_0 \in \Theta.
\end{align}
Therefore, given a sample $\left\{ \left( Y_t, \mathcal{I}_{t-1}^{\prime} \right)^{\prime}: 1 \leq t \leq n \right\}$ and a parameter value $\theta$, one can consider the following quantile-marked empirical process indexed by $x \in \mathbb{R}^d$
\begin{align}
S_n \left( x, \alpha, \theta \right) := n^{-1/2} \sum_{t=1}^n \Psi_{\alpha, t} \ \text{exp} \left( i x^{\prime} \mathcal{I}_{t-1} \right).
\end{align}
in order to construct a formal specification testing procedure based on the functional form proposed by \cite{escanciano2010specification}. 
A commonly used estimator for the unknown parameter $\theta_0$ is the quantile regression estimator, which can be obtained as the solution of the corresponding minimization problem with $\rho_{\alpha} \left( \epsilon \right) = - \Psi_{ \alpha } \left( \epsilon \right) \epsilon $.

\newpage

\begin{remark}
Some key points here include the fact that the parametric approach although is considered to be a robust estimation methodology when one utilizes the correct distributional assumption, it also requires careful use of complex analysis arguments to establish the asymptotic theory of the conditional quantile specification test. On the other hand, a specification testing procedure based on a semiparametric estimation of moment conditions as in the seminal paper of  \cite{newey2004efficient} could be possible but in that case additional regularity conditions are necessary for estimation and inference purposes (see, for example \cite{wang2016conditional}). In particular, \cite{horvath2022consistent} use a direct estimation approach for the conditional expectation function estimated by series regression in order to construct a consistent conditional quantile specification test (via the classical Bierens test) based on an empirical process for a continuum of instrumented unconditional moments. 
\end{remark}

\subsubsection{Asymptotic null distribution}

We briefly discuss the main results of the asymptotic theory proposed by \cite{escanciano2010specification}. More precisely, in order to establish the limit distribution of the quantile-marked empirical process, under the null hypothesis, $\mathbb{H}_0$, one needs the following notation. Define the family of conditional distributions $F_x( y ) : = \mathbb{P} \big( Y_t \leq y | \mathcal{I}_{t-1} = x \big)$. Moreover, let 
\begin{align}
\epsilon_{t, \alpha} := Y_t - q_{\alpha} \left( \mathcal{I}_{t-1} \right) \ \ \ \text{and} \ \ \ e_t \big( \theta \big) \equiv e_t \left( \theta (\alpha ) \right) := Y_t - m \big( \mathcal{I}_{t-1}, \theta(\alpha) \big)
\end{align}
where $\left\{ ( Y_t, Z_t^{\prime} )^{\prime} : t \in \mathbb{Z} \right\}$ is a strictly stationary and ergodic process. Then, under the null hypothesis, $\mathbb{H}_0$, we have that $\left\{ \Psi_{\alpha, t} \left( \theta_0  \right), \mathcal{F}_t \right\}$ is a \textit{mds} \ $\forall$ \ $\alpha \in \Pi$, such that, $\Pi \subset [0,1]$. Furthermore, it holds that the finite-dimensional distribution of $R_n$ converge to those of a multivariate normal distribution with mean a zero mean vector and variance-covariance matrix given by the covariance function as expressed below
\begin{align}
K_{ \infty } \left( v_1, v_2 \right) = ( \alpha_1 \land \alpha_2 - \alpha_1 \alpha_2 ) \mathbb{E} \big[ \text{exp} \left(  i (x_1 - x_2 ) \mathcal{I}_{0} \right) \big]
\end{align}
where $v_1 = ( x_1^{\prime}, \alpha_1 )$ and $v_2 = ( x_2^{\prime}, \alpha_1 )$ represent generic elements of $\Pi$.
  
\begin{corollary}
The estimator $\theta_n$ satisfies that $\mathbb{P} \left( \theta_n \in \mathcal{B} \right) \to 1$, as $n \to \infty$, and the following asymptotic expansion under $\mathbb{H}_0$, uniformly in $\alpha \in \mathcal{G}$ and $\theta_0 \in \mathcal{B}$, which implies 
\begin{align}
Q_n ( \alpha ) 
&= \sqrt{n} \left( \theta_n ( \alpha ) - \theta_0 ( \alpha ) \right)
= \frac{1}{ \sqrt{n} } \sum_{t=1}^n \mathcal{L}_{\alpha} \left( Y_{t}, \mathcal{I}_{t-1}, \theta_0 ( \alpha ) \right) + o_p(1).
\end{align}   
where $l_{ \alpha }( . )$ is such that $\mathbb{E} \big[ l_{ \alpha } \big( Y_1, I_0, \theta_0 ( \alpha ) \big) \big] = 0$ and
\begin{align}
\mathcal{L}_{ \alpha } \left( \theta_0 ( \alpha ) \right) = \mathbb{E} \bigg[ l_{ \alpha } \big( Y_1, I_0, \theta_0 ( \alpha ) \big) l_{ \alpha }^{\prime} \big( Y_1, I_0, \theta_0 ( \alpha ) \big) \bigg]
\end{align}
exists and is positive definite. 
\end{corollary}

\newpage

Moreover, the following condition holds
\begin{align}
\mathbb{E} \bigg[ l_{ \alpha } \bigg( Y_t, I_{t-1}, \theta_0 ( \alpha ) \bigg) \Psi_{\alpha} \bigg( Y_s - m ( \mathcal{I}_{s-1}, \theta_0 ( \alpha ) ) \bigg) \bigg] = 0, \ \ \text{if} \ \ t \neq s.
\end{align}
Then, $Q_n ( \alpha )$ converges weakly to a Gaussian process $Q(.)$ with zero mean and covariance function
\begin{align}
K_Q ( \alpha_1, \alpha_2 ) = \underset{ n \to \infty }{ \text{lim} } \ \frac{1}{n} \sum_{t=1}^n \sum_{s=1}^n \ \mathbb{E} \bigg[ l_{ \alpha_1 } \big( Y_1, \mathcal{I}_{t-1}, \theta_0 ( \alpha_1 ) \big) \times l_{ \alpha_2 } \big( Y_s, \mathcal{I}_{s-1}, \theta_0 ( \alpha_2 ) \big) \bigg]
\end{align}

\begin{corollary}
The estimator $\theta_n ( \alpha )$ satisfies the following asymptotic expansion under $\mathbb{H}_{1,n}$, uniformly in $\alpha$, such that 
\begin{align}
\sqrt{n} \big( \theta_n ( \alpha ) - \theta_0 ( \alpha ) \big) = \xi_{ \alpha } ( \alpha ) + \frac{1}{ \sqrt{n} } \sum_{t=1}^n l_{ \alpha } \big( Y_t, \mathcal{I}_{t-1}, \theta_0 ( \alpha ) \big) + o_p( 1 ).
\end{align} 
\end{corollary}

\subsubsection{CvM test statistic}

Denote with $\mathcal{T}_m := \left\{ \alpha_j \right\}_{j = 1}^{ \alpha}$ the points in the grid which represent a set of different quantile levels, with $\alpha_1 < ... < \alpha_m$. Let $W_{\text{exp}}$ be the $n \times n$ matrix with $w_{\text{exp}, t, s} = \text{exp} \left( -\frac{1}{2} \left| \mathcal{I}_{t-1} - \mathcal{I}_{s-1} \right|^2 \right)$ and let $\psi$ be the $n \times n$ matrix with typical elements given by $\psi_{ij} = \psi_{\alpha j } \left(  Y_i - m \left( \mathcal{I}_{i-1}, \theta_n  \right) \right)$. Hence, the CvM test statistic is computed as below
\begin{align}
\label{cvm.statistic}
\text{CvM}_n = \frac{1}{mn} \sum_{j=1}^m \psi_{.j}^{\prime} W_{\text{exp}} \psi_{.j}
\end{align}
where $\psi_{.j}$ denotes the $j-$th column of $\Psi$, with $m$ fixed points such as $\left\{ \alpha_j \right\}_{j = 1}^{ \alpha}$ are deterministic. 

\begin{remark}
The CvM test proposed by \cite{escanciano2010specification} is constructed by first constructing the $\Psi$ matrix with the quantile residuals of the model (for a general functional form) on each column, over different quantile levels alpha. That is, we evaluate for example the quantile regression model with the quantile level being in the interval e.g., $[0.1, 0.9]$. Then, the CvM statistic given by \eqref{cvm.statistic} which is a quadratic form functional with covariance kernel defined by $W_{\text{exp}}$, which is constructed based on scanning the information set (e.g. the predictors we include in the model). This is the unconditional moment restriction approach proposed in the literature. In particular, estimation of the CoVaR risk measure require to use a generated regressor. Testing for the correct specification of a dynamic quantile process is crucial for the robust identification of the region which represents risk measures such as the CoVaR (see, \cite{tobias2016covar}). Correct specification of a conditional quantile model implies that a particular conditional moment is equal to zero. Regarding the inclusion of nearly-integrated predictors, relevant research questions include: \textit{(i).} How are the particular fluctuation residual based test is affected when the predictors are generated separately with an AR(1) model which allows additional features such as different types of persistence, serial correlation, heteroscedasticity, and \textit{(ii).} How does the $W_{\text{exp}}$ function changes when we allow for the AR(1) specification. 
\end{remark}

\newpage

\subsection{Misspecification testing in Quantile Regression}

A different perspective is presented in the framework proposed by \cite{firpo2022gmm} which in practise it corresponds to a misspecification testing methodology in conditional quantile regression models based on a GMM estimation approach. Specifically, consider the linear quantile regression model (QR) as
\begin{align}
Q_{Y|X} ( \tau | X ) = X^{\top} \beta( \tau ),      
\end{align}
where $Y \in \mathcal{Y} \subset \mathbb{R}$ is the outcome variable, $X \in \mathcal{X} \subset \mathbb{R}^K$ is a $K-$dimensional vector of covariates and $Q_{Y|X}$ is the conditional quantile function of $Y$ given $X$ and $\tau \in (0,1)$ is the fixed quantile. Thus, a quantile functional form such as the slope vector or a subset of its components can be written as below
\begin{align}
\beta(\tau) = \mathsf{g} ( \theta, \tau )    
\end{align}
Moreover, by imposing the above restriction on the first $K_1$ components of $\beta$ such that 
\begin{align}
Q_{Y|X} ( \tau | X ) = X_1^{\top} \mathsf{g}_1 (\theta, \tau) + X_2^{\top} \mathsf{\beta}_2 \mathsf{g}_2 (\theta, \tau)     
\end{align}
where $X_1 \in \mathbb{R}^{K_1}$ are the first components of $X$ and $X_2 \in \mathbb{R}^{K_2}$ are the remaining components. Therefore, if $X = ( X_1, X_2^{\top} )^{\top}$, where $X_1$ is a scalar variable, we can compute $\widehat{\theta}$ by minimizing the weighted distance between the QR estimator $\beta_1 (\tau)$ and parametrization $\mathsf{g}_1(\theta, \tau)$ using the following objective function
\begin{align}
\int_{ \mathcal{T} } \norm{ \widehat{\beta}_1 (\tau) - \mathsf{g}_1(\theta, \tau ) }^2 \widehat{W}_1 (\tau) d \tau.     
\end{align}
Then, \cite{firpo2022gmm} impose the following assumptions to study the limiting properties of $\widehat{\theta}_{GMM}$. 
\begin{itemize}
    
\item[\textbf{A1.}] The data $\left\{ Y_i, X_i \right\}_{i=1}^n$ are independent and identically distributed. 

\item[\textbf{A2.}] $Y|X=x$ has a density with respect to the Lebesgue measure for all $x \in \mathcal{X}$, the support of $X$. Moreover, it holds that, the matrix $ J(\tau)$ is invertible, where 
\begin{align}
J(\tau) = \mathbb{E} \left[ f_{Y|X} \left( X^{\top} \beta(\tau) | X \right) X X^{\top} \right]
\end{align}
\end{itemize}
First we show the uniform convergence of $\widehat{Q}_n (\theta)$ to $Q(\theta)$. Notice that 
\begin{align*}
\underset{ \theta \in \Theta }{ \mathsf{sup} } \left|  \widehat{Q}_n( \theta ) - Q( \theta ) \right| &\leq \underset{ \theta \in \Theta }{ \mathsf{sup} } \ \int_{ \mathcal{T} } \left| \widehat{W}(\tau) - W(\tau) \right| \left| \widehat{\beta}(\tau) - \mathsf{g}( \theta, \tau ) \right|^2 d \tau  
\\
&+
\underset{ \theta \in \Theta }{ \mathsf{sup} } \ \int_{\mathcal{T} } \left| W(\tau) \right|
\left|  \big| \widehat{\beta}(\tau) - \mathsf{g}( \theta, \tau ) \big|^2 - \big| \beta(\tau) - \mathsf{g} (\theta, \tau ) \big|^2  \right| d \tau
\\
&\leq 
\underset{ \tau \in \mathcal{T} }{ \mathsf{sup} } \ \left| \widehat{W}(\tau) - W(\tau) \right| \underset{ ( \theta, \tau ) \in \Theta \times \mathcal{T} }{ \mathsf{sup} } \big| \widehat{\beta}(\tau) - \mathsf{g}( \theta, \tau ) \big|^2 \int_{\mathcal{T} } \boldsymbol{1} . d \tau
\\
&+
\int_{ \mathcal{T} } \left| W(\tau) \right| d \tau \ \times \underset{ \tau \in \mathcal{T} }{ \mathsf{sup} }  \left| \widehat{\beta}(\tau) - \beta (\tau) \right| \left(  \underset{ \tau \in \mathcal{T} }{ \mathsf{sup} }  \left| \widehat{\beta}(\tau) - \beta (\tau) \right| + \underset{ ( \theta, \tau ) \in \Theta \times \mathcal{T} }{ \mathsf{sup} } \left| \mathsf{g} (\theta_0, \tau ) - \mathsf{g} ( \theta, \tau ) \right| \right).
\end{align*}

\newpage

Next, we establish the asymptotic normality of the MD-QR estimator. In particular, by the first order condition, the estimator $\widehat{\theta}_{MD}$ satisfies the following condition
\begin{align}
\int_{ \mathcal{T} } \mathsf{g} \left( \widehat{\theta}_{MD}, \tau \right)^{\top} \widehat{W}(\tau) \left( \widehat{\beta}(\tau) - \mathsf{g} \left( \widehat{\theta}_{MD}, \tau \right) \right) d \tau = 0.  
\end{align}
Moreover, using a Taylor expansion of $\mathsf{g} \left( \widehat{\theta}_{MD}, \tau \right)$ around $\theta_0$ we obtain
\begin{align*}
\int_{ \mathcal{T} } \mathsf{g}_{\theta} \left( \widehat{\theta}_{MD}, \tau \right)^{\top} \widehat{W}(\tau) \left( \widehat{\beta}(\tau) - \mathsf{g} \left( \theta_0, \tau \right) \right) d \tau  
-
\left( \int_{ \mathcal{T} } \mathsf{g}_{\theta} \left( \widehat{\theta}_{MD}, \tau \right)^{\top} \widehat{W}(\tau) \mathsf{g}_{\theta} \left( \bar{\theta}, \tau \right) d \tau \right) \left( \widehat{\theta}_{MD} - \theta_0 \right) = 0.
\end{align*}
where $\bar{\theta}$ is a line segment between $\widehat{\theta}_{MD}$ and $\theta_0$. Therefore, we obtain that
\begin{align*}
\widehat{\theta}_{MD} - \theta_0 =  \left( \int_{ \mathcal{T} } \mathsf{g}_{\theta} \left( \widehat{\theta}_{MD}, \tau \right)^{\top} \widehat{W}(\tau) \mathsf{g}_{\theta} \left( \bar{\theta}, \tau \right) d \tau \right)^{-1} \left( \int_{ \mathcal{T} } \mathsf{g}_{\theta} \left( \widehat{\theta}_{MD}, \tau \right)^{\top} \widehat{W}(\tau) \left( \widehat{\beta}(\tau) - \mathsf{g} \left( \theta_0, \tau \right) \right) d \tau    \right)    
\end{align*}
By Theorem 3 in \cite{angrist2006quantile}, it follows that $\sqrt{n} \left( \widehat{\beta}(\tau) - \beta (\tau) \right) \overset{d}{\to} \boldsymbol{Z}_{\beta}( \tau)$, is a mean-zero Gaussian process in $\ell^{\infty} ( \mathcal{T}, \boldsymbol{R}^K )$ which is the set of bounded functions that map $\mathcal{T}$ into $\mathbb{R}^K$ with covariance kernel $\Sigma ( \tau, \tau^{\prime} )$. Therefore, the parametric rate $\sqrt{n}$ times the above term, converges in distribution to 
\begin{align}
\int_{ \mathcal{T} } \mathsf{g}_{\theta} \left( \widehat{\theta}_{MD}, \tau \right)^{\top} \widehat{W}(\tau) \textcolor{blue}{ \sqrt{n} } \left( \widehat{\beta}(\tau) - \mathsf{g} \left( \theta_0, \tau \right) \right) d \tau \overset{d}{\to} \int_{\mathcal{T}} \mathsf{g}_{\theta} (\theta_0, \tau )^{\top} W(\tau) \boldsymbol{Z}_{\beta}(\tau) d \tau \equiv \boldsymbol{Z},    
\end{align}
by the continuous mapping theorem for functionals. Since $\boldsymbol{Z}$ is a linear functional of $\boldsymbol{Z}_{\beta}$, a mean-zero Gaussian process, is also mean-zero and Gaussian. Its variance is equal to
\begin{align}
\mathbb{E} \big[ \boldsymbol{Z} \boldsymbol{Z}^{\top} \big] = \int_{\mathcal{T}^2 } \mathsf{g}_{\theta} (\theta_0, \tau )^{\top} W(\tau) \mathbb{E} \left[ \boldsymbol{Z}_{\beta}(\tau) \boldsymbol{Z}_{\beta}(\tau^{\prime})^{\top} \right] W(\tau^{\prime} )  \mathsf{g}_{\theta} (\theta_0, \tau ) d \tau d \tau^{\prime} =: V_{MD} (\theta_0)  
\end{align}

\subsubsection{Minimum Distance QR Estimator}

Moreover, \cite{firpo2022gmm} propose a minimum distance quantile regression (MD-QR) estimator as an alternative to GMM estimators. The MD-QR is defined by minimizing the integrated weighted distance, over a subset of quantiles $\mathcal{T}$, between the standard QR estimator and the parametric model for some of the coefficients of interest. Let $\widehat{\beta}(\tau)$ denote the standard QR estimator. The MD-QR estimator of $\theta$
\begin{align}
\widehat{\theta}_{MD} = \underset{ \theta \in \Theta }{ \mathsf{arg min} } \ \int_{\mathcal{T} } \left( \widehat{\beta}(\tau) - \mathsf{g}(\theta, \tau) \right)^{\top} \widehat{W}(\tau) \left( \widehat{\beta}(\tau) - \mathsf{g}(\theta, \tau) \right) d \tau,  
\end{align}
where $\widehat{W}(\tau) \in \mathbb{R}^{K \times K}$ is a weight function that converges in probability to a positive-definite valued function $W(\tau)$. Notice that if one is interested about the correct specification of the function $\mathsf{g} (.,.)$ which is used to model some of the components of $\beta(\tau)$. Since we are using the GMM objective function, hence a simple specification test can be applied.

\newpage 

\begin{theorem}[\cite{firpo2022gmm}]
As $n \to \infty$, $\widehat{\theta}_{MD}$ is a consistent estimator of $\theta_0$and it holds that 
\begin{align}
\sqrt{n} \left( \widehat{\theta}_{MD} - \theta_0 \right) &\overset{d }{\to} \mathcal{N} \big( 0, H(\theta_0)^{-1} V_{MD}(\theta_0) H(\theta_0)^{-1} \big)    
\\
H(\theta_0) &:= \int_{ \mathcal{\theta}_0 } \mathsf{g} (\theta_0, \tau )^{\top} W(\tau) \mathsf{g} (\theta_0, \tau ) d \tau    
\end{align}
\end{theorem}
\textit{Proof.}
\begin{align}
\textcolor{blue}{ \sqrt{n} } 
\begin{pmatrix}
\widehat{\beta}(\tau) - \beta(\tau) 
\\
\widehat{\theta}_{MD} - \theta_0
\end{pmatrix}
\overset{d}{\to}
\begin{pmatrix}
\boldsymbol{Z}_{\beta}(\tau)
\\
\displaystyle H(\theta_0)^{-1} \int_{\mathcal{T} } \mathsf{g}_{\theta}(\theta_0, u)^{\top} W(u) \boldsymbol{Z}_{\beta}(\tau) du
\end{pmatrix},
\end{align}
and by the delta method, we have that 
\begin{align}
\textcolor{blue}{ \sqrt{n} } 
\begin{pmatrix}
\widehat{\beta}(\tau) - \beta(\tau) 
\\
\mathsf{g} \left( \widehat{\theta}_{MD}, \tau \right)  -  \mathsf{g} \left( \theta_0, \tau \right) 
\end{pmatrix}
\overset{d}{\to}
\begin{pmatrix}
\boldsymbol{Z}_{\beta}(\tau)
\\
\displaystyle \mathsf{g} \left( \theta_0, \tau \right)^{\top} H(\theta_0)^{-1} \int_{\mathcal{T} } \mathsf{g}_{\theta}(\theta_0, u)^{\top} W(u) \boldsymbol{Z}_{\beta}(\tau) du
\end{pmatrix}.
\end{align}
Therefore, we can write the normalized objective function as below
\begin{align*}
&n \widehat{Q}_n \left( \widehat{\theta}_{MD} \right) 
= 
\int_{ \mathcal{T} } \left[ \textcolor{red}{\sqrt{n}} \left( \widehat{\beta}(\tau) - \mathsf{g} \left( \widehat{\theta}_{MD}, \tau \right) \right)  \right]^{\top} \widehat{W}(\tau) \left[ \textcolor{red}{\sqrt{n}} \left( \widehat{\beta}(\tau) - \mathsf{g} \left( \widehat{\theta}_{MD}, \tau \right) \right)  \right]   
\\
&= 
\int_{ \mathcal{T} } \left[ \textcolor{red}{\sqrt{n}} \left( \widehat{\beta}(\tau) - \mathsf{g} \left( \widehat{\theta}_{MD}, \tau \right) \right) - \textcolor{blue}{ \big( \beta(\tau) - \mathsf{g} (\theta_0, \tau)  \big) } \right]^{\top} \widehat{W}(\tau) \left[ \textcolor{red}{\sqrt{n}} \left( \widehat{\beta}(\tau) - \mathsf{g} \left( \widehat{\theta}_{MD}, \tau \right) \right) - \textcolor{blue}{ \big( \beta(\tau) - \mathsf{g} (\theta_0, \tau)  \big) } \right] d \tau  
\end{align*}
since under correct specification it holds that $\textcolor{blue}{\beta(\tau) \equiv \mathsf{g}(\theta_0, \tau)}$ for all $\tau \in \mathcal{T}$. Therefore, by an application of the functional delta method and by the consistency of $\widehat{W}(\tau)$ for $W(\tau)$, we obtain 
\begin{align*}
n \widehat{Q}_n \left( \widehat{\theta}_{MD} \right) 
&\overset{d}{\to} 
\int_{\mathcal{T}} \left( \boldsymbol{Z}_{\beta}(\tau) - \mathsf{g}_{\theta}(\theta_0, \tau)^{\top} H(\theta_0)^{-1} \int_{\mathcal{T}} \mathsf{g}_{\theta} (\theta_0, u )^{\top} W(u) \boldsymbol{Z}_{\beta}(u) du \right)^{\top} W(\tau)   
\\
&\ \ \ \ \ \times \left( \boldsymbol{Z}_{\beta}(\tau) - \mathsf{g}_{\theta}(\theta_0, \tau)^{\top} H(\theta_0)^{-1} \int_{\mathcal{T}} \mathsf{g}_{\theta} (\theta_0, u )^{\top} W(u) \boldsymbol{Z}_{\beta}(u) du \right) d \tau
\\
&= \int_{\mathcal{T} } \widetilde{\boldsymbol{Z}}(\tau)^{\top} W(\tau) \widetilde{\boldsymbol{Z}}(\tau). 
\end{align*}

\medskip

\begin{remark}
Recall that correct specification means correct (dynamic) specification of the conditional mean, that is, $\mathbb{E} \big[ Y_t | X_t = x_t \big] = m_t \big( x_t, \theta_0 \big)$, for some $\theta_0 \in \Theta$. More specifically, a conditional mean test is designed to detect departures from the hypothesis $\mathbb{H}_0: \mathbb{E} \big[ Y_t | X_t  \big] = m_t \big( X_t, \theta_0 \big)$. Define with $u_t ( \theta ) \equiv Y_t - m_t ( \Theta )$ be the $1 \times K$ residual function, and let $u_t^{0} = u_t ( \theta_0 )$ be the true residuals under the null hypothesis, $\mathbb{H}_0$. Suppose that the random vector $( Y_t, X_{t1}, X_{t2}, ... , X_{tk} )$ is continuously distributed with joint probability density $f_{Y,X} (.,.)$ and conditional probability density $f_{Y|X} (.|x)$ for $Y_t$ given $X_t = x$. Then, the conditional quantile $q_{\theta} ( Y_t | X_t )$ satisfies $\displaystyle \int_{- \infty}^{ q_{\theta ( Y_t | X_t )}} f_{ Y | X } ( y | X_t ) dy - \theta = 0$ (see, \cite{kim2003estimation}).
\end{remark}

\newpage

\subsection{Specification testing in Quantile Cointegrating Regressions}

A relevant framework is proposed by \cite{tu2022nonparametric}. In particular, their model specification testing approach is constructed by a comparison of the nonparametric estimator of $\mathsf{g} ( x,z )$ and the corresponding parametric estimator of $\mathsf{g}_0 ( x, z ; \theta_0 )$ under the null hypothesis. The following test statistic is used 
\begin{align}
R_n = \int \int \left[ \widehat{\mathsf{g}}(x,z) - \mathsf{g}_0 (x,z; \widehat{\theta}_n )      \right]^2 \pi_1 (x) \pi_2 (z) p_n(x,z) dx dz      
\end{align}
where $\pi_1 (x)$ and $\pi_2 (z)$ are positive integrable weight functions and the parametric quantile estimator $\hat{\theta}_n$
\begin{align}
\hat{\theta}_n = \underset{ \theta \in \Theta }{ \mathsf{arg min} } \sum_{t=1}^n \rho_{\tau} \left[ y_t - \mathsf{g}_0 \left( x_t, z_t ; \theta \right) \right]   
\end{align}
and $p_n( x, z )$ is a weighting function. Furthermore, due to the presence of nonstationary regressors as well as the kernel density estimators involved makes deriving the asymptotic theory quite challenging. Therefore, we consider an alternative test statistic to $R_n$ which is a corresponding smooth version. Recall that the quantile estimator $\widehat{\mathsf{g}}(x,z)$ of $\mathsf{g}(x,z)$ such that 
\begin{align}
\widehat{\mathsf{g}}(x,z) = \mathsf{g}(x,z) + \frac{ \displaystyle \sum_{t=1}^n K_1 \left( \frac{x_t - x}{h_1} \right) K_2 \left( \frac{z_t - z}{h_2} \right) \psi_{\tau} (u_t) f^{-1} (0) }{  \displaystyle \sum_{t=1}^n K_1 \left( \frac{x_t - x}{h_1} \right) K_2 \left( \frac{z_t - z}{h_2} \right) } + o_p \left( \nu_n^{-1} \right),   
\end{align}
where $\psi_{\tau} = \tau - \mathbf{1} \left\{ t < 0 \right\}$, and there exists a smooth version $\widetilde{\mathsf{g}}_0 \left( x, z; \theta \right)$ of $\mathsf{g}_0 \left( x, z; \theta \right)$ such that 
\begin{align}
\widetilde{\mathsf{g}}_0 \left( x, z; \theta \right) \equiv \frac{ \displaystyle \sum_{t=1}^n K_1 \left( \frac{x_t - x}{h_1} \right) K_2 \left( \frac{z_t - z}{h_2} \right) \mathsf{g}_0 \left( x_t, z_t; \theta \right) }{  \displaystyle \sum_{t=1}^n K_1 \left( \frac{x_t - x}{h_1} \right) K_2 \left( \frac{z_t - z}{h_2} \right) } = \mathsf{g}_0 \left( x_t, z_t; \theta \right) + o_p \left( \nu_n^{-1} \right)   
\end{align}
Therefore, a smoothed version of $R_n$, with  $\mathsf{g}_0 (x,z; \widehat{\theta}_n )$ replaced by $\widetilde{\mathsf{g}}_0 (x,z; \widehat{\theta}_n )$ is
\begin{align}
R_n = \int \int \left[ \widehat{\mathsf{g}}(x,z) - \widetilde{\mathsf{g}}_0 (x,z; \widehat{\theta}_n )      \right]^2 \pi_1 (x) \pi_2 (z) p_n(x,z) dx dz      
\end{align}
and, if we let $p_n( x,z ) = \left\{ \sum_{t=1}^n K_1 \left( \frac{x_t - x}{h_1} \right) K_2 \left( \frac{z_t - z}{h_2} \right) \right\}^2$. Define with 
\begin{align}
\mathcal{S}_n := \int_{- \infty}^{ + \infty } \int_{- \infty}^{ + \infty } \left\{  \sum_{k=1}^n \psi_{\tau} (u_k) K_1 \left( \frac{x_k - x}{ h_1 } \right) K_2 \left( \frac{z_k - z}{ h_2 } \right)  \right\}^2 \pi_1(x) \pi_2(z) dxdz,      
\end{align}
Thus, $\frac{ d_n }{n h_1 h_2} \mathcal{S}_n$ converges in distribution to a local time of an Ornstein-Uhlenbeck process.

\newpage 

\subsection{Specification testing in Nonparametric Cointegrating Regressions}

In this Section, we present the main results of the framework proposed by \cite{wang2012specification} who construct a statistical methodology for parametric specification testing in cointegrating models (see, also  \cite{kasparis2012dynamic}). 

\subsubsection{Specification test in Nonlinear Nonstationary Cointegrating Regression}

Consider the nonlinear cointegrating regression model given by
\begin{align}
y_{t+1} = f(x_t) + u_{t+1}, \ \ \ t= \left\{ 1,...,n \right\}.
\end{align}
where $u_t$ is a stationary error process, and $x_t$ is a nonstationary regressor. The testing hypothesis of interest is expressed as below 
\begin{align}
\mathbb{H}_0: f(x) = f(x, \theta), \ \ \theta \in \Omega_0
\end{align} 
To test the null hypothesis, we use the following kernel-smoothed test statistic
\begin{align}
S_n = \sum_{s,t=1, s \neq t}^n \hat{u}_{t+1} \hat{u}_{s+1} K  \left( \frac{x_t - x_s}{h} \right)
\end{align}
involving the parametric regression residuals $\hat{u}_{t+1} = y_{t-1} - f( x_t, \hat{\theta} )$, where $K(x)$ is a nonnegative real kernel function and $h$ the window bandwidth satisfying $h \equiv h_n \to 0$ as $n \to \infty$. Note that the difficulty for developing an asymptotic theory for $S_n$ stems from the presence of the kernel weights $K \left( (x_t - x_s ) / h \right)$. The behaviour of these weights depends on the self intersection properties of $x_t$ in the sample. Therefore, to establish asymptotics for $S_n$, we need to account for the related limit theory which involves self-intersection local time of a Gaussian process. 
\begin{assumption}[\cite{wang2012specification}]
Let $\left\{ \epsilon_t \right\}_{ t \in \mathbb{Z} }$ be a sequence of i.i.d continuous random variables with $E \epsilon_0 = 0$ and $E \epsilon^2_0 = 1$, and with a characteristic function $\phi(t)$. Then, the regressor $x_t$
\begin{align}
x_t = \rho x_{t-1} + \eta_t, \ \ \ x_0 = 0, \ \ \rho = \left( 1 + \frac{\kappa}{n} \right), \ \ \ 1 \leq t \leq n,
\end{align} 
where $\kappa$ is some constant and $\eta_t = \sum_{k=0}^{\infty} \phi_k \epsilon_{t-k}$, with $\phi \equiv \sum_{j=0}^{ \infty } \phi_j \neq 0$ and the following summability condition $\sum_{j=0}^{ \infty } j^{1+ \delta} | \phi_j | < \infty$ holds for some $\delta > 0$.
\end{assumption}  

\begin{assumption}[\cite{wang2012specification}]
The following conditions hold:

\begin{itemize}

\item[\textit{(i).}] Let $\left\{ u_t, \mathcal{F}_{t} \right\}_{ t \geq 1 }$, where $\mathcal{F}_{t}$ is a sequence of increasing $\sigma-$fields which is independent of $\epsilon_j$, $j \geq t + 1$, forms a martingale difference sequence satisfying  $E \left( u^2_{t+1} | \mathcal{F}_{t} \right) \to \sigma^2 > 0$, as $t \to \infty$ and $\text{sup}_{t \geq 1} E \left( | u_{t+1}|^4 | \mathcal{F}_{t} \right) < \infty$.

\newpage

\item[\textit{(ii).}]$x_t$ is adapted to $\mathcal{F}_t$, and there exists a correlated vector Brownian motion $(W,V)$ such that the following joint weakly convergence holds
\begin{align}
\left( \frac{1}{\sqrt{n} } \sum_{j=1}^{\floor{nt} } \epsilon_j , \frac{1}{\sqrt{n} } \sum_{j=1}^{\floor{nt} } u_{j+1}   \right) \Rightarrow \big( W(t), V(t) \big)
\end{align}
on $\mathcal{D}[0,1]^2$ as $n \to \infty$.

\end{itemize}
\end{assumption}

\subsubsection{Specification testing in Nonparametric  Predictive Regression models}

These class of specification testing methodologies as in the study of \cite{duffy2021estimation} (among others) compare the functional forms of the parametric against the alternative hypothesis of a correctly specified nonparametric functional form (currently the literature corresponds to the conditional mean functional form but extensions to a conditional quantile functional form can potentially be possible).  

\paragraph{Parametric Regression}

Consider the OLS estimator of $\left( \mu, \gamma \right)$ of the regression model given by
\begin{align}
\label{parametric}
y_t = \mu + \gamma f \left( x_{t-1} \right) + u_t
\end{align}
where the unknown model parameters are obtained as below: 
\begin{align}
\left( \hat{\mu}, \hat{\gamma} \right ) 
:= 
\underset{ \left( \mu ,  \gamma  \right) }{ \text{ argmin}} \sum_{t=1}^n \big[ y_t - \mu - \gamma f \left( x_{t-1} \right) \big]^2
\end{align}
where $f$ has a known nonlinear functional form. The single regressor included in the nonparametric predictive regression model is assumed to be predetermined or $\mathcal{F}_{t-1}-$measurable. Then, \cite{duffy2021estimation} show that if $x_t$ is stationary, the OLS estimator is asymptotically normal, whereas in the case that $x_t(n)$ exhibits strong dependence structure, then the OLS estimator has a non-standard limiting distribution. Specifically, the authors assume a linear array of the form 
\begin{align}
x_t(n) = \sum_{j=0}^{t-1} \phi_j(n) v_{t-j}, \ \ \ \text{where} \ \ \ v_t = \sum_{i=0}^{\infty} c_i \xi_{t-i}
\end{align}
where $\left\{ \xi_t \right\}_{ t \in \mathbb{Z} }$ satisfies the following assumption. 

\begin{assumption}[\cite{duffy2021estimation}]
Let $\left\{ \xi_t \right\}_{ t \in \mathbb{Z} }$ be a sequence of innovation terms, then  
\begin{itemize}
\item[(i)] $\xi_t$ is \textit{i.i.d} with $\mathbf{E} ( \xi_1 ) = 0$ and Var$(\xi_1) = \sigma^2_{\xi} < \infty$.
\item[(ii)] $\xi_1$ has an \textit{absolute continuous distribution}, and a \textit{characteristic function} $\psi_{\xi} ( \lambda )$ that satisfies $\int_{ \mathbb{R}} | \psi_{\xi} ( \lambda )^{\theta}| d \lambda < \infty$, for some $\theta \in \mathbb{N}$. 
\end{itemize}
\end{assumption}
Furthermore, based on the local-to-unity parametrization  \cite{duffy2021estimation} consider in their framework two distinct classes of persistence processes, that is, nearly integrated (NI) processes and mildly integrated (MI) processes. Both MI and NI processes can be defined in terms of triangular arrays  
\begin{align}
x_t(n) = \left( 1 - \frac{1}{\kappa_n}  \right) x_{t-1} (n) + v_t, \ x_0(n) = 0.
\end{align}
where $v_t$ is a stationary process and $\kappa_n > 0$ with $\kappa_n < \infty$, so that the autoregressive coefficient becomes approximates to unity as $n$ grows.
Both NI and MI processes describe highly persistent autoregressive processes, which have a root in the vicinity of unity. As a matter of fact both processes have been extensively used to study the behaviour of inferential procedures under local departures from unit roots and robust inferential procedures (e.g., see \cite{phillips2007limit} and \cite{kostakis2015Robust}). The crucial difference between NI and MI processes concerns the assumed growth rate of the sequence $\kappa_n$. In particular, NI processes are defined by $\kappa_n / n \to c \neq 0$, with the consequence that $n^{-1 / 2 } x_{ \floor{ n r}}$ converges weakly to an Ornstein-Uhlenbeck procees. On the other hand, MI processes have $\kappa_n / n \to 0$, which implies that the process $x_t(n)$ moves closer to the neighbourhood of stationarity, which implies that \textit{FCLT} cannot be used to derive the asymptotics of functionals of these processes \citep{duffy2021estimation}.

\paragraph{Nonparametric regression}

Consider the nonparametric estimation of the predictive regression
\begin{align}
y_t = m \left( x_{t-1} \right) + u_t
\end{align}  
The unknown function is estimated nonparametrically via the Nadaraya-Watson (NW) estimator
\begin{align}
\hat{m} ( x ) 
:=
\frac{ \displaystyle \sum_{t=2}^n K_{th} (x) y_t }{ \displaystyle \sum_{t=2}^n K_{th} (x)  } 
\end{align}
and the corresponding local linear (LL) estimator is expressed as
\begin{align}
\begin{bmatrix}
\tilde{\mu} (x) 
\\
\tilde{\gamma} (x)
\end{bmatrix}
:= 
\underset{ ( \mu, \gamma ) }{ \text{argmin}  } \sum_{t=2}^n \big[ y_t - \mu - \gamma \left( x_{t-1} - x \right)  \big]^2 K_{th} \left( x \right)
\end{align}
where $K_{th} (x) := K \left( \frac{x_{t-1} - x}{ h_n }  \right)$. Then, the testing hypothesis of interest is that the true regression function $m$ belongs to a certain parametric family. More specifically, the null hypothesis is formulated as 
\begin{align}
\label{hypothesis}
\mathbb{H}_0: m(x) = \mu + \gamma f(x), \ \ \text{for some} \ \ ( \mu, \gamma) \in \mathbb{R}^2 \ \ \text{and all} \ \ x \in \mathbb{R}
\end{align}
where $f$ is a known function. The alternative hypothesis implies that no such $\mu$ and $\gamma$ exist.

\newpage 

In other words, the specification test proposed by \cite{duffy2021estimation} allows to test the null hypothesis of parametric fit under the presence of NI regressor, which closely resembles that of \cite{gao2009specification}, who propose a studentized U-statistic formed of kernel-weighted OLS regression residuals. The specification test proposed by \cite{duffy2021estimation} tests the null hypothesis $\mathbb{H}_0$ by comparing the parametric OLS and the corresponding nonparametric (kernel) estimates of the $m$ function. In particular, the model specified in \eqref{hypothesis} can be estimated nonparametrically at each $x$ as below
\begin{align}
\widetilde{m}_f( x) := \text{argmin} \ \underset{ \gamma \in \mathbb{R}}{ \text{min} } \sum_{t=1}^n \bigg\{  y_t  - \mu - \gamma \big[ g(x_{t-1} - g(x)  \big] \bigg\}^2 K_{th } (x)
\end{align}
Denote with $\left( \hat{\mu}, \hat{\gamma} \right)$ the OLS estimates of $\left( \mu, \gamma \right)$, then the statistical distance corresponds to a comparison between the fit provided by the parametric and local nonparametric estimates of the model via an ensemble of $t$statistics of the following form
\begin{align}
\tilde{t} ( x; \hat{\mu}, \hat{\gamma} ) := \left[  \frac{ \sum_{t=2}^n K_{th} (x)  }{  \hat{\sigma}_u^2 (x) Q_{11} } \right]^{1 / 2} \bigg[   \widetilde{m}_f( x) - \hat{\mu} - \hat{\gamma} f(x) \bigg]
\end{align}
where $\hat{\sigma}_u^2 (x): = \bigg(  \sum_{t=2}^n K_{th} (x) \bigg)^{-1} \sum_{t=2}^n \big[ y_t - \hat{\mu} - \hat{\gamma} f(x_{t-1}  \big] K_{th} (x)$. Therefore, under the null hypothesis of no misspecification both the parametric and nonparametric estimators converge to identical limits and so for each $x \in \mathbb{R}$, the following limit is obtained
\begin{align}
\tilde{t} \big( x; \hat{\mu}, \hat{\gamma} \big) :=  \tilde{t} \big( x; \mu, \gamma \big) + o_p(1) \overset{ d }{ \to } \mathcal{N} (0,1).
\end{align}

\subsubsection{Specification testing in Nonlinear Cointegrating Regression with nonlinear dependence}

Following \cite{escanciano2009econometrics} who consider specification testing in the context of nonlinear cointegrating regressions with nonlinear dependence. Specifically, the near-epoch dependence condition when imposed on the error term of the regression model is considered to impose a form of nonlinear dependence in time series regressions models. 

\medskip

\begin{definition}(Strong Mixing) Let $\left\{ v_t \right\}$ be a sequence of random variables. Let $\mathcal{F}_s^t \equiv \sigma \left( \nu_s,...., \nu_t \right)$ be the generated $\sigma-$algebra. Define the $\alpha-$mixing coefficients such that
\begin{align}
\alpha_m = \underset{ t }{ \mathsf{sup} }  \left|  P( G \cup F ) - P (G) P(F) \right|  
\end{align}    
The process $\left\{ v_t \right\}$ is said to be strong mixing if $\alpha_m \to 0$ as $m \to \infty$. Moreover, if $\alpha_m \leq m^{- \alpha}$, we say that $\left\{ v_t \right\}$ is a strong mixing of size $- \alpha$.
\end{definition}

\newpage

\begin{definition}[NED] Let $\left\{ w_t \right\}$ be a sequence of random variables with $\mathbb{E} \left[ w_t^2 \right] < \infty$ for all $t$. Then, it is said that $\left\{ w_t \right\}$ is NED on the underlying sequence $\left\{ v_t \right\}$ of size-$\alpha$ if $\phi(n)$ is of size $-\alpha$, where $\phi(n)$ is 
\begin{align}
\mathsf{sup} \norm{ w_t - \mathbb{E}_{t-n}^{t+n} (w_t) }_2 \equiv \phi(n)     
\end{align}
where $\mathbb{E}_{t-n}^{t+n} (w_t) = \mathbb{E} \left( w_t | \nu_{t-n},..., \nu_{t+n} \right)$ and $\norm{.}_2$ is the $L_2$ norm of a random vector.     
\end{definition}

\medskip

\begin{definition}[Weak-Dependence based definition of $I(0)$] A sequence $\left\{ w_t \right\}$ is $I(0)$ if it is NED on an underlying $\alpha-$mixing sequence $\left\{ v_t \right\}$ but the sequence $\left\{ x_t \right\}$ given by $x_t = \sum_{s=1}^t w_s$ is not NED on $\left\{ v_t \right\}$. In this case, we say that $x_t$ is $I(1)$.  
\end{definition}

\medskip

\begin{definition}
Two $I(1)$ sequences $\left\{ y_t \right\}$ and $\left\{ x_t \right\}$ are (linearly) cointegrated with cointegrating vector $[ 1, - \beta ]$, if $y_t - \beta x_t$ is NED on a particular $\alpha-$mixing sequence but $y_t - \delta_{12} x_t$ is not NED for $\delta_{12} \neq \beta$.      
\end{definition}

\medskip

\begin{theorem}[Granger's Representation Theorem] Consider the nonlinear correction model (NEC) for the vector $X_t = \left( y_t, x_t \right)^{\prime}$, given by the following expression 
\begin{align}
\Delta X_t = \Psi_1 \Delta X_{t-1} + F ( X_{t-1} ) + v_t
\end{align}
Assume that 
\begin{itemize}
    
\item[(a)] Let $v_t = \left( v_{y_t},  v_{x_t} \right)^{\prime}$ is $\alpha-$mixing of size $-s/(s-2)$ for $s > 2$.

\item[(b)] $\sum_t v_t$ is not NED on $\alpha-$mixing sequence. 
     
\end{itemize}
   
\end{theorem}
In particular, the above conditions ensure that 
\begin{itemize}
    
    \item[(i)] $\Delta X_t$ and $Z_t$ are simultaneously NED on the $\alpha-$mixing sequence $( v_t, u_t )$, where $u_t = v_{y,x} - \beta^{\prime} v_{x,t}$, 

    \item[(ii)] $X_t$ is $I(1)$.
    
\end{itemize}
According to \cite{escanciano2009econometrics}, the above theorem provides sufficient conditions for cointegrated variables to be generated by a nonlinear error correction model. Furthermore, the fully modified OLS and FM-OLS estimation approach involves a two-step procedure in which in the first step the following econometric specification $\Delta y_t = \psi_1 \Delta x_t + \gamma z_{t-1} + v_t$, is estimated by OLS. Moreover, in the second step semiparametric corrections are made for the serial correlations of the residuals $z_t$ and for the endogeneity of the $x-$regressors. Under general conditions the fully modified estimator is asymptotically efficient. A particular case of the more general form of a nonlinear parametric cointegration is of the form $y_t = f( x_t, \beta ) + v_t$, where $x_t$ is a $( p \times 1 )$ vector of $I(1)$ regressors, $v_t$ is a zero-mean stationary error term and $f( x_t, \beta )$ a smooth function of the process $x_t$, known up to the finite-dimensional parameter vector $\beta$. Thus, in the classical asymptotic theory, the properties of estimators of $\beta$, i.e., the nonlinear least squares estimator (NLSE), depend on the specific class of functions where $f( x_t, \beta )$ belongs. The rate of convergence of the NLSE is class-specific and in some cases involves random scaling. 

\newpage

\newpage 

\appendix
\numberwithin{equation}{section}
\makeatletter

\newpage 

\section{Appendix}

\subsection{Elements on Gaussian Approximation of Suprema of Empirical Processes}

Following the framework of \cite{chernozhukov2014gaussian} who consider the problem of approximating suprema of empirical processes by a sequence of suprema of Gaussian processes. To formulate the problem, let $X_1,..., X_n$ be an \textit{i.i.d} random variables taking values in a measurable space $( S, \mathbb{S} )$ with common distribution $\mathcal{P}$. Suppose there is a sequence $\mathcal{F}_n$ of classes of measurable functions $S \mapsto \mathbb{R}$, and consider the empirical process indexed by $\mathcal{F}_n$:
\begin{align}
\label{problem1}
\mathbb{G}_n f = \frac{1}{ \sqrt{n} } \sum_{i=1}^n \big( f(X_i) - \mathbb{E} \left[ f(X_i) \right]  \big), \ \ f \in \mathcal{F}_n.
\end{align}
More specifically, \cite{chernozhukov2014gaussian} consider the problem of approximating $Z_n = \mathsf{sup}_{f \in \mathcal{F}_n } \mathbb{G}_n f$ by a sequence of random variables $\widetilde{Z}_n$ equal in distribution to $\mathsf{sup}_{f \in \mathcal{F}_n } B_n f$, where each $B_n$ is a centered Gaussian process indexed by $\mathcal{F}_n$ with covariance function $\mathbb{E} \left[ B_n(f) B_n(g) \right] = \mathsf{Cov} \left( f(X_1), g(X_1) \right)$ for all $f,g \in \mathcal{F}_n$.  
We look for conditions under which there exists a sequence of such random variables $\widetilde{Z}_n$ with 
\begin{align}
\left| Z_n - \widetilde{Z}_n \right| = \mathcal{O}_{\mathbb{P} } ( r_n)
\end{align}
where $r_n \to 0$ as $n \to \infty$ is a sequence of constants.

A related but different problem is that of approximating whole empirical processes by a sequence of Gaussian processes in the sup-norm. This problem is more difficult than \eqref{problem1}. Indeed,  \eqref{problem1} is implied if there exists a sequence of versions of $B_n$ such that
\begin{align}
\label{problem2}
\norm{ \mathbb{G}_n - B_n }_{ \mathcal{F}_n } := \underset{ f \in \mathcal{F}_n }{ \mathsf{sup} } \left| ( \mathbb{G}_n - B_n ) f    \right| = \mathcal{O}_{ \mathbb{P } } (r_n).
\end{align}
There is a large literature on the latter problem. Notably, \cite{komlos1975approximation} (henceforth, abbreviated as KMT) proved that $\norm{ \mathbb{G}_n - B_n }_{ \mathcal{F} } = \mathcal{O}_{a.s} \left( n^{- 1 / 2} \text{log} n \right)$. The KMT construction is a powerful tool in addressing the problem \eqref{problem2}, but when applied to general empirical processes, it typically requires strong conditions on classes of functions and distributions.

\begin{definition}[Asymptotic Tightness]
There exists $c_1, c_2 > 0$ such that for every small $\epsilon > 0$ there exists $\delta (\epsilon) > 0$ such that for all $h > 0$ and all $s \in T$, it holds that
\begin{align}
\mathbb{P} \left( \underset{ t \in T: \rho(s,t) \leq \delta ( \epsilon ) }{ \mathsf{sup} } \ U(t) \leq h e^{ c_1 \epsilon^2 } \ \bigg| \ \underset{ t \in T: \rho(s,t) \leq \delta ( \epsilon ) }{ \mathsf{inf} } \ U(t) \leq h    \right) \geq 1 - c_2 \epsilon^2.
\end{align}
Moreover, if $N_{T, \delta}$ is the minimal number of closed $\delta-$balls needed to cover $T$, then it holds that
\begin{align*}
\int_0^1 \sqrt{ \mathsf{log} N_{T, \delta ( \epsilon )  }  } < \infty.
\end{align*}
\end{definition} 

\newpage 

\subsection{Elements on Stochastic Processes under Nonlinear Dependence}

In the econometrics and statistics literature the main ways to define the dependence structure of innovation sequences in time series regression models is using martigale differences and mixing processes. However, a function of a mixing sequence, which depends on an infinite number of lags and/or leads of the sequence, is not generally mixing (see, \cite{qiu2004functional}). In particular, a concept which reflects weak dependence is the near-epoch dependence. The main idea is that although $\left\{ X_t \right\}$ may not be mixing, if it depends almost entirely on the near epoch of a mixing sequence $\left\{ V_t \right\}$, it will often have properties permitting the application of limit theorems. 
\color{black}

\medskip

\begin{definition}
\label{def1}
Let $p > 0$ and $\left\{ X_n, n \geq 1 \right\}$ will be called an $L_p-$near-epoch dependent $L_p-$NED on $\left\{ V_n, n \geq 1 \right\}$, if there exist sequence $\left\{ d_n \right\}$ and $\nu (m)$ of non-negative constants, where $\nu (m) \to 0$ as $m \to \infty$ such that for $n \geq 1$ and $m \geq 0$, 
\begin{align}
 \norm{ X_n - \mathbb{E}_{n-m}^{n + m} X_n }_p \leq \nu (m) d_n.    
\end{align}    
\end{definition}
hen, we can $\left\{ X_n, n \geq 1 \right\}$ an $L_p-$NED sequence of size $-\lambda$ if $\nu (m)$ in Definition \ref{def1} is of size $- \lambda$, that is, $\nu (m) = \mathcal{O} \left( m^{-\lambda-\epsilon} \right)$, for some $\epsilon > 0$. 

\begin{remark}
Notice that Definition \ref{def1} above can be adjusted accordingly to correspond to the innovation sequences of the structural cointegrating regression model. In the literature a maximal inequality on $p-$th moment, $p > 2$, under weaker dependent sizes was established. Moreover, applying this inequality, we focus on deriving a CLT for strong $L_p-$NED under the same moment conditions but weaker size requirement for dependence (\cite{qiu2004functional})).  
\end{remark}

\begin{definition} Let $p > 0$ and $\left\{ X_n, n \geq 1 \right\}$ will be called an $L_p-$near-epoch dependent $L_p-$NED on $\left\{ V_n, n \geq 1 \right\}$, if there exist sequence $\left\{ d_n \right\}$ and $\nu (m)$ of non-negative constants, where $\nu (m) \to 0$ as $m \to \infty$ such that for all $k > 0$, $n \geq 1$ and $m \geq 0$, it holds that  
\begin{align*}
 \norm{ \sum_{t=k+1}^{k+n} X_t  - \mathbb{E} \left[ \sum_{t=k+1}^{k+n} X_t \bigg| \mathcal{F}_{k+n+m} \right] }_p 
 &\equiv 
 \left\{ \left| \mathbb{E} \left[  \sum_{t=k+1}^{n+k} X_t  - \mathbb{E} \left( \sum_{t=k+1}^{k+n} X_t \bigg| \mathcal{F}_{k+n+m} \right) \right] \right| \right\}^{1/p} 
 \\
 &\leq \nu (m) \times \left(  \sum_{j=1}^n d_{k+j}^2 \right)^{1/2}.    
\end{align*}    
\end{definition}
In order to facilitate the development of the asymptotic theory in our study our aim is to impose the NED conditions given by \cite{lu2007local} which incorporates a more general dependence structure in time series regression models than the classical $\text{i.i.d}$ assumption of error terms, to the innovation structure of the nonparametric cointegrating regression model such that 
\begin{align}
\nu_p(m) = \mathbb{E} \left[ y_t - y_t^{(m)} \right] + \left\{ \mathbb{E} \left[ x_t - x_t^{(m)} \right] \right\}^p.
\end{align}
Furthermore, a fundamental result which is essential for developing the asymptotic behaviour of estimators and test statistics is to establish the weak convergence of partial sum processes of stationary sequences (i.e., the innovation sequences of the structural cointegrating regression) under the conditions of near-epoch dependence to  functionals of Brownian motion. Specifically, \cite{davidson1992central} establishes conditions for the FCLT in nonlinear and semiparametric time series processes to hold, that is, demonstrates that such sums of random variables under NED have a Gaussian limit distribution. The main idea here is that since we can prove that under the condition of NED the partial sum processes of the corresponding stochastic process weakly converges into a Brownian motion functional, then our aim is to investigate whether a similar result applies for the weakly convergence of the joint functional of partial sum processes under NED conditions. Although NED conditions can be considered as a form of $m-$dependence (e.g., see \cite{berkes2011split}), the main challenging task in our setting is that we consider cointegrating regression models with nonstationary regressors and this requires to determine the exact form auxiliary limit functionals will have. 

\subsubsection{Adaptation}

A natural case to consider when the sequence $\left\{ X_t, \mathcal{F}_t \right\}$ is adapted, such that $\mathbb{E} \left[ X_t | \mathcal{F}_t \right] = X_t$ almost surely. This has the simple implication that the process is causal, if the progression of the sequence is identified with the passage of time, then under the adaptation assumption it is driven by present and past shocks alone. This is a natural assumption in the context of econometric time series modelling, implying that knowledge of future shocks cannot improve the prediction of $X_t$ once its history is known. In the mixingale context, it means that the norm is equal to zero for every $m \geq 0$, whereas the NED-defining inequality becomes 
\begin{align}
\norm{ X_t - \mathbb{E} \left( X_t | \mathcal{F}_{t-m}^t \right) } \leq d_t \nu_m   
\end{align}

\begin{proof}
Uniform integrability implies that 
\begin{align}
\underset{ n,t }{ \mathsf{sup} } \ \mathbb{E} \left[  \left| \frac{X_{nt} }{ c_{nt} } \right|^p \mathbf{1}_{ \left\{  \left| \frac{X_{nt}}{ c_{nt} } \right| > M  \right\} } \right] \to 0 \ \ \ \text{as} \ \ M \to \infty.
\end{align}
Therefore, for $\varepsilon > 0$ there is a constant $B_{\varepsilon} < \infty$ such that 
\begin{align}
\underset{ n,t }{ \mathsf{sup} } \  \left\{ \frac{ \norm{ X_{nt} \mathbf{1}_{ \left\{ \left| X_{nt} \right| > B_{\varepsilon} c_{nt} \right\} } }_p }{ c_{nt} } \right\} \leq \varepsilon.
\end{align}
Define with $Y_{nt} = X_{nt} \mathbf{1}_{ \left\{ \left| X_{nt} \right| \leq B_{\varepsilon} c_{nt} \right\} }$ and $Z_{nt} = X_{nt} - Y_{nt}$. Then, since $\mathbb{E} \left( X_{nt} | \mathcal{F}_{n,t-1} \right) = 0$, 
\begin{align}
X_{nt} = Y_{nt} - \mathbb{E} \left[ Y_{nt} | \mathcal{F}_{n,t-1} \right] + Z_{nt} - \mathbb{E} \left[ Z_{nt} | \mathcal{F}_{n,t-1} \right].    
\end{align}

\newpage

By the Minkowski inequality it holds that 
\begin{align}
\norm{ \sum_{t=1}^{k_n} X_{nt} }_p \leq \norm{ \sum_{t=1}^{k_n} \big( Y_{nt} - \mathbb{E} \left[  Y_{nt} | \mathcal{F}_{n,t-1} \right]  \big) }_p     
+
\norm{ \sum_{t=1}^n \big( Z_{nt} - \mathbb{E} \left[ Z_{nt} |  \mathcal{F}_{n,t-1} \right] \big) }_p.
\end{align}
\end{proof}

\subsubsection{FCLT for Strong NED Random Variables}

\begin{proposition}[\cite{qiu2004functional}]
Let $\left\{ X_n, n \geq 1 \right\}$ be a sequence of random variables with $\mathbb{E} \left[ X_n \right] = 0$ and $\mathbb{E} \left[ \left| X_n \right|^p \right] \leq M < \infty$ , for some $p > 2$. Moreover, suppose that $\left\{ X_n, n \geq 1 \right\}$ is strong $L_p-$NED on a $\phi-$mixing sequence $\left\{ V_n, n \geq 1 \right\}$ with 
\begin{align}
\phi(n) = \mathcal{O} \left( \left( \mathsf{log}(n) \right) ^{ \frac{1}{ p ( 1 + \delta/ 2 )} } \right)   
\end{align}
and $\left\{ d_n \right\}$ and $\nu(m)$ satisfying 
\begin{align}
\underset{ n \to \infty }{ \mathsf{lim sum  } }  \  \underset{ k \geq 0 }{ \mathsf{sup} } \ \left\{ \frac{1}{n} \sum_{j=1}^n d^2_{j+k} \right\} = B 
\end{align}
and $\sigma_n^2 := \mathbb{E} \left[ S_n^2 \right] \to \infty$, as $n \to \infty$, where $S_n = \sum_{j=1}^n X_j$. Then it holds that 
\begin{align}
S_n / \sigma_n \to \mathcal{N} (0,1).     
\end{align}
\end{proposition}

Moreover, we consider the limit behaviour of the partial sums of $X_{nt}$ where
\begin{align}
X_n(r) = \sum_{t=1}^{ \floor{nr} } X_{nt}, \ \ \ \text{for} \ r \in [0,1],    
\end{align}
where $X_{nt} = X_t / \sigma_n$. 
\begin{lemma}
Let $\left\{ X_n, n \in \mathbb{Z} \right\}$ be a $\phi-$mixing sequence and $X \in L_p \left( \mathcal{F}_{-\infty}^k \right)$ and $Y \in L_q \left( \mathcal{F}_{k + n}^{\infty} \right)$, with $p > 1$ and $q > 1$. Then, 
\begin{align}
 \big| \mathbb{E} \left( XY \right) - \mathbb{E} (X) \mathbb{E} (Y)  \big| \leq 2 \phi (n)^{1/p} \norm{X}_p \norm{Y}_q   
\end{align}
\end{lemma}

\begin{theorem}
Let $\left\{ X_n, n \geq 1 \right\}$ and assume that 
\begin{align}
\psi ( r ) := \underset{ n \to \infty }{ \mathsf{lim} } \mathbb{E} \left[ X_n(r) \right]^2    
\end{align}
exists for all $r \in [0,1]$. 
Then $X_n ( . ) \to X(.)$, where $X(.)$ is a Gaussian process having almost surely continuous sample paths and independent increments (see, \cite{qiu2004functional}).
\end{theorem}

\newpage

Then, the main step for proving the particular result is to obtain 
\begin{align}
\norm{  \sum_{ s = k_n( \theta_{j-1} ) + 1 }^{ k_n ( \theta_j + \delta ) } X_{ns} }_2 \leq   CD \big( ... \big) \sigma_n^{-2} 
\end{align}
Therefore, by the Cauchy-Schwarz inequality and the definition of strong $L_p-$NED, we apply the Remark 2 such that we can bound the following expression 
\begin{align}
I_2 \leq \left\{ \mathbb{E} \bigg[ Z_{ m_{i-1} } ( n_i ) \right] - \mathbb{E}_{ m_{i-1} + 1 - m }^{ m_{i-1} + n_i + m } \left[  Z_{ m_{i-1} } ( n_i ) \right]  \bigg\} \times \left\{ \mathbb{E} \left[  Z_{ m_{j-1}  (\delta ) } \big( n_j ( \delta ) \big) \right]^2 \right\}^{1/2}  
\end{align}

\begin{remark}
Some remarks based on the limit results established above are: 

\begin{itemize}
    \item[\textbf{(i).}] The basic idea is to consider that we approximate the pair $\left\{ \left( y_t, x_t \right)  \right\}$ by the $\alpha-$mixing process $\left\{ \left( y_t^{(m)}, x_t^{(m)} \right) \right\}$. Notice that this approximation even though is not changing the nonstationary properties as these are captured by the structural regression model, it does affect the asymptotic theory when considering the local linear fitting estimation approach. In particular, we consider that the structural cointegrating regression is generated by an $\alpha-$mixing $m-$dependent pair of stochastic processes under NED conditions (see,  \cite{katsouris2023specification}). 

    \item[\textbf{(ii).}] Thus, in order to demonstrate that the joint weak convergence holds for these functionals under NED one needs to extend the result given by Theorem 1 in the paper of Jin and Zhengyan, which corresponds only to one coordinate, to the case in which we have joint weak convergence of these functionals for both coordinates. Notice that these two coordinates do not necessarily have to correspond to both random variables of the $m-$dependent pair of random variables $( y_t, x_t)$. 
    
\end{itemize}
    
\end{remark}

Roughly speaking, a structural cointegrating regression under conditions of fading memory; since the NED condition is assumed to be a form of weak dependence such that the underline stochastic processes have fading long-term memory due to less interactions (e.g., see \cite{schennach2018long}) and weakening dependence which can affect the degree of cointegration in the long-run. Exactly for this reason, the local linear fitting method (see, \cite{lu2007local}, \cite{linton2005estimating}) is a suitable approach for estimation and inference purposes within such a setting. 

\begin{example}
Consider the process below
\begin{align}
h_t^m = \sum_{j= -m}^m \theta_j V_{t-j}    
\end{align}
Then, the function above is different from  $E( X_t | \mathcal{F}_{t-m}^{t+m} )$ unless $\left\{ V_t \right\}$ is an independent process, but is also an $L_p-$approximator for $X_t$ since
\begin{align}
\norm{ X_t - h_t^m }_p = \norm{  \sum_{j=m+1}^{\infty} \big( \theta_j V_{t-j} + \theta_{-j} V_{t+j} \big) }_p \leq d_t \nu_m   
\end{align}
\end{example}

\newpage

Therefore, in contrast to the property of "long memory"; which implies that shocks have persistent effect resulting to a slow rate of decay of the autocorrelation function with lag or by a divergence of the power spectrum near the origin - NED conditions its a form of weak dependence with memory that fades away. Therefore, in order to estimate the cointegrating regression model under NED we consider the local linear fitting approach. The main challenge however is to employ correctly adjusted invariance principles for NED sequences to the corresponding brownian motion functionals.  
\begin{definition}
The stationary process $\left\{ \right( Y_t, \boldsymbol{X}_t \left) \right\}_{t=1}^n$ is said to be $L_{p}-$NED on $\left\{ \varepsilon_t \right\}$ if  
\begin{align}
\nu_{p}(m) = \mathbb{E} \left|  Y_t - Y_t^{(m)} \right|^{p}  +   \mathbb{E} \norm{ \boldsymbol{X}_t - \boldsymbol{X}_t^{(m)} }^{p}  \to 0,    
\end{align}
as $m \to \infty$. Then, the $\nu_{p}(m)$ are called the stability coefficients of order $p$ for the joint stochastic process $\left\{ \left( Y_t, \boldsymbol{X}_t \right) \right\}_{t=1}^n$ under the presence of Near-epoch dependence (NED).
\end{definition}

\begin{definition}
Let $\mathcal{F}_{- \infty}^n$ and $\mathcal{F}_{\infty}^{n + k}$ be the $\sigma-$fields generated by $\left\{ \varepsilon_t, t \leq n \right\}$ and $\left\{ \varepsilon_t, t > n + k \right\}$, respectively. Moreover, define with 
\begin{align}
\alpha (k) = \underset{ A \in \mathcal{F}_{- \infty}^n, B \in \mathcal{F}_{\infty}^{n + k} }{ \mathsf{sup} } \ \left|  \mathbb{P} \left( AB \right) - \mathbb{P} \left(A \right) \mathbb{P} \left( B \right) \right| \to 0, \ \ \text{as} \ n \to \infty.
\end{align}
Then, the stationary sequence $\left\{ \varepsilon_t, t = 0,  \pm 1,...  \right\}$ is said to be $\alpha-$mixing and the $\alpha(k)$ is termed mixing coefficient. Thus, $\left\{ \right( Y_t^{(m)}, \boldsymbol{X}_t^{(m)} \left) \right\}$ is an $\alpha-$mixing process with mixing coefficient 
\begin{align}
\alpha_m(k) \leq 
\begin{cases}
 \alpha (k - m) &, k \geq m + 1,   
 \\
 1 & k \geq m.
\end{cases}
\end{align}
\end{definition}

\medskip

\begin{remark}
In other words NED is a property of the mapping rather of the random variables themselves (see, \cite{davidson2002establishing} and \cite{qiu2006variance}). The NED condition fully determines the restrictions on the memory of the observed process. In other words, the econometrician can assume independent shocks, although this is principally for tractability of asymptotic theory analysis derivations. Notice that the class of NED functions of mixing processes includes the cases in which infinite distributed lag functions of strong-mixing processes are strong mixing  such as ARMA processes with mixing innovations. Moreover, a strong approximation result for heterogenous martingale difference arrays is given by  \cite{davidson1992central} and \cite{de1997central}. The particular structure generalizes the independent data assumption towards a general setting with serial dependency and heterogeneity.  
\end{remark}

Furthermore, as per the discussion of \cite{kuan1994adaptive}, when the stochastic process $\left\{ X_t \right\}$ is a mixingale and $X_t$ is measurable with respect to the filtration $\mathcal{F}_t$, we have an adapted mixingale and the condition of $\norm{ X_t - \mathbb{E} \left( X_t | \mathcal{F}_{t+m}  \right) }_2$ holds. The mixingale condition implies that the term $\norm{ \mathbb{E} \left( X_t | \mathcal{F}_{t+m} \right) }_2$ converges to zero as $m \to \infty$, which is a form of asymptotic martigale difference property.

\newpage

Special cases of mixingale processes include independent sequences, $\phi-$, $\rho-$ and $\alpha-$mixing processes, finite and certain infinite order moving average processes and sequences of NED functions of infinite histories of mixing processes. Therefore, an invariance principle for $\left\{ X_n \right\}$ to hold the degree of dependence among $X_n$'s is controlled. In other words, the controlling sequence $\psi_k$ has to decay to zero sufficiently fact. Let $S_{n,k} = \sum_{j=1}^k X_{n,j}$.

\medskip 

Then, $W_n(t)$ as defined below corresponds to a standard Brownian motion. 
\begin{align}
W_{n}(t) = \frac{ S_{n, \floor{nt} } }{  \sqrt{n} \sigma }    
\end{align}
Similarly, an invariance principle for triangular arrays of mixingales holds.     

\medskip

\begin{theorem}
Let $\left\{ \left( X_{n,j}, \mathcal{F}_{n,j} \right) \right\}$ be an array of mixingales satisfying the NED conditions with $\left\{  d_m \right\}$ of size $- \frac{1}{2}$. If $\left\{ X_{n,j}, j = 1,2,...  \right\}$ is uniformly integrable, then $\left\{ W_n \right\}$ is tight and 
\begin{align}
\norm{ \mathbb{E}_{k-m} \left[  \frac{ \left( S_{n,k+d} - S_{n,k} \right)^2 }{ d} - \sigma^2    \right] }_1 \to 0    
\end{align}
as $\mathsf{min} \left( m, k, n, d \right) \to \infty$, then $W_n$ converges weakly to $W$ on $D[0,1]$. 
\end{theorem}

Various studies in the literature have proposed invariance principles (functional central limit theorems) for partial sum processes of sequences of dependent random variables such as \cite{Phillips1987time, Phillips1988regression, Phillips1987towards}. Moreover, more general dependence structures are given by \cite{davidson1992central, davidson2002establishing}. Following the aforementioned literature and a natural extension when considering econometric models with near-epoch dependence, is to prove the weak convergence of partial sum processes of NED sequences to the $J_1$ topology. Notice that although $X_t$  may not be mixing, if it depends almost entirely on the near epoch of $\left\{ \boldsymbol{V}_t  \right\}$ it will often have properties permitting the application of limit theorems, of which the mixingale property is the most important. Furthermore, note that near-epoch dependence is not an alternative to a mixing assumption, is a property of the mapping from $\left\{ \boldsymbol{V}_t \right\}$ to $\left\{ X_t \right\}$, not of the random variables themselves. In particular, the NED property is important when $\left\{ \boldsymbol{V}_t  \right\}$ is a mixing process, because then $\left\{ X_t \right\}$ inherits certain useful characteristics. Notice that $\mathbb{E} \left( X_t | \mathcal{F}_{t-m}^{t+m}     \right)$ is a finite-lag, then the corresponding information set is considered to be a measurable function of a mixing process and hence is also mixing. Thus, NED implies that $\left\{ X_t \right\}$ is approximately mixing in the sense of being well approximated by a mixing process. Therefore, a NED function of a mixing process can be mixingale subject to suitable restrictions on the distribution moments. From the point of view of applications, NED captures the characteristics of a stable dynamic econometric model in which a dependent variable $X_t$ depends mainly on the recent histories of a collection of explanatory variables or shock processes $\boldsymbol{V}_t$, which might be assumed to be mixing.  

\newpage 

\begin{proposition}[M-dependence]
If $\left\{ X_j \right\}_{ j \in \mathbb{Z}}$ is a stationary $m-$ dependent Gaussian sequence with $\mathbb{E} X_0 = \mu$, then there exists $\left\{ a_k \right\}_{k=0}^m$ such that 
\begin{align}
\left\{ X_j \right\}_{ j \in \mathbb{Z}} = \left\{ \mu + \sum_{k=0}^{\infty} \alpha_j \eta_{j - k }  \right\}_{ j \in \mathbb{Z}}
\end{align}
where $\left\{ \eta_{j} \right\}_{ j \in \mathbb{Z}}$ are \textit{i.i.d} standard normal random variables. 
\end{proposition}

\begin{proof}
Let $R(j) = \text{Cov} \left( X_0, X_j \right)$, for $j \in \mathbb{R}$. $R(j)$ is a positive define function on $\mathbb{Z}$ and $R(j) = 0$ for $|j| > m$. The Riesz factorization lemma implies that there exists $\left\{ a_k \right\}_{ k = 0}^{\infty}$ such that 
\begin{align}
\sum_{ j = -  m }^m R(j) e^{ij \theta} = \bigg|  \sum_{j=0}^m a_j e^{ij \theta} \bigg|^2.
\end{align}
By expanding the square and grouping the coefficients of $e^{ij \theta}$, one sees that 
\begin{align}
\label{covariance}
R(j) = \sum_{ k = 0 }^{m - j} a_k a_{ k + j}, \ \ \text{for} \ \ 0 \leq j \leq m.
\end{align}
The RHS of \eqref{covariance} is the covariance function of the finite moving average sequence $\left\{ \sum_{k=0}^m a_k \eta_{j - k }  \right\}_{ j \in \mathbb{Z} }$, then the proposition is verified.
\end{proof}

\begin{remark}
Although all $m-$dependent stationary Gaussian sequences are finite moving average sequences, it is next shown that there are $m-$dependent stationary infinitely divisible sequences that do not have the same distribution as any finite moving average of $\textit{i.i.d}$ infinitely divisible random variables. 
\end{remark}

\subsection{Approximations of NED Processes}

Understanding the underline mechanism for handling NED processes and their corresponding approximations is important especially when econometric identification is concerned with spatial equilibrium dynamics based on cointegrating relations. In other words, the existence of such phenomenan can be proved by rejecting the null hypothesis of cointegration with no NED versus cointegration under NED conditions. For example, rejecting the null hypothesis of no presence of spatial generated cointegration dynamics with NED conditions, implies that economic agents decide to cointegrate only under NED and spatial separability. Next, we follow the derivations in the framework of \cite{lu2007local}.  

Assume that the stationary process $Y_t$ and $\boldsymbol{X}_t$ are $\mathbb{R}^1$ and $\mathbb{R}^p-$valued random fields, respectively and 
\begin{align}
Y_t &:= \Psi_Y \left(  \varepsilon_t, \varepsilon_{t-1}, \varepsilon_{t-2},... \right)   
\\
\boldsymbol{X}_t &:= \Psi_Y \left(  \varepsilon_t, \varepsilon_{t-1}, \varepsilon_{t-2},... \right),  
\end{align}
where $\boldsymbol{X}_t = \left( X_{t1},..., X_{td} \right)^{\prime} \in \mathbb{R}^d$. 

\newpage 

The idea of the approximations for an NED process as proposed by \cite{lu2007local} is to be able to express an NED process using an $\alpha-$mixing process approximation.

In particular, the $\alpha-$mixing process approximation is given by the following expression 
\begin{align}
Y_t &= Y_t^{(m)} + \left( Y_t - Y_t^{(m)} \right) := Y_t^{(m)} + \delta_{Y,t}^{(m)}
\\
\boldsymbol{X}_t &= \boldsymbol{X}_t^{(m)} + \left( \boldsymbol{X}_t - \boldsymbol{X}_t^{(m)} \right)   := \boldsymbol{X}_t^{(m)} + \delta_{X,t}^{(m)}
\end{align}
where
\begin{align}
\mathbb{E} \left[ \delta_{Y,t}^{(m)} \right] = \mathcal{O} \left( \nu_2 (m) \right) \ \ \ \text{and} \ \ \ \mathbb{E} \left[ \delta_{X,t}^{(m)} \right] = \mathcal{O} \left( \nu_2 (m) \right), \ \ \text{as} \ \ m \to \infty.
\end{align}
where the $\alpha-$mixing coefficient satisfies the following condition
\begin{align}
\alpha_m(k) \leq 1 \ \ \ \text{for} \ \ k = 0,1,...,m \ \ \text{and} \ \ \alpha_m(k) = \alpha (k-m) \ \ \text{for} \ \ k \geq m + 1
\end{align}
Therefore, in order to be able to estimate the local linear quantile estimator under NED when regressors are integrated we need to modify the approximation based on the I(1) reprentation of the process. Furthermore, the error term of the process is asymptotically equivalent in the case we impose the assumption of a linear process representation for the innovation sequence. 
\begin{remark}
Notice that the main intuition of this step is that the near-epoch dependence can be considered as a type of $m-$dependence processes which means that even if the data around the neighborhood of the block size are dependent, asymptotically as $m \to \infty$, then an asymptotic independence assumption holds which can facilitate the development of the asymptotic theory without necessitating to consider finite-sample covariance matrix corrections when constructing test statistics. 
\end{remark}

Let $m = m_T$ be a positive integer tending to $\infty$ as $T \to \infty$. Then, it holds that 
\begin{align}
\mathbb{E} \left[ Y_i^{(m)} K \left(  \frac{ \boldsymbol{x} - \boldsymbol{X}_i^{(m)}}{b_T} \right) \right] 
=
\mathbb{E} \left[ Y_i K \left(  \frac{ \boldsymbol{x} - \boldsymbol{X}_i }{b_T} \right) \right] + \mathcal{O}_p \left( b_T^{-1} L \nu^{1/2} (m) \right)
\end{align}
which in practice means that the convergence rate is a function of the block size. Thus the larger the block size the slower the convergence rate to the true parameter space. Therefore, in order to consider the efficiency bound of the approximation which depends both on the block size as well as on the $\alpha-$mixing property we take the expected value difference between the stochastic process under near-epoch dependence and the corresponding stochastic process without the near-epoch dependence as given by the following expression elememtwise: 
\begin{align}
\mathcal{I}:= \mathbb{E} \left|  \left( \boldsymbol{W}_T^{(m)} \right)_i -  \left( \boldsymbol{W}_T \right)_i \right| 
&=
\mathbb{E} \left| \left( T b_T^d \right)^{-1} \sum_{j=1}^T \left\{ Z_j^{(m)} K_i \left( \frac{ \boldsymbol{X}_j^{(m)} - \boldsymbol{x} }{ b_T } \right) - Z_j K_i \left( \frac{ \boldsymbol{X}_j - \boldsymbol{x} }{ b_T } \right) \right\} \right|
\end{align}

\newpage

\begin{align*}
\mathcal{I}
&\leq 
 \left( T b_T^d \right)^{-1} \sum_{j=1}^T \mathbb{E} \left| Z_j^{(m)} K_i \left( \frac{ \boldsymbol{X}_j^{(m)} - \boldsymbol{x} }{ b_T } \right) - Z_j K_i \left( \frac{ \boldsymbol{X}_j - \boldsymbol{x} }{ b_T } \right) \right|
\\
&\leq 
 b_T^{-d} \mathbb{E} \left| Z_j^{(m)} - Z_j \right| K_i \left( \frac{ \boldsymbol{X}_j^{(m)} - \boldsymbol{x} }{ b_T } \right) 
+ 
b_T^{-d} \mathbb{E} \bigg| Z_j \bigg| \mathbf{1}\left\{  |Z_j | \leq L \right\} 
 + \mathbf{1}\left\{  |Z_j | \geq L \right\} \bigg| K_i \left( \frac{ \boldsymbol{X}^{(m)}_j - \boldsymbol{x} }{ b_T } \right) -  K_i \left( \frac{ \boldsymbol{X}_j - \boldsymbol{x} }{ b_T } \right) \bigg| 
\\
&= 
\mathcal{O}_p \left( b_T^{-d} \sqrt{ \nu_2(m)} \right) +  \mathcal{O}_p \left( b_T^{- (d+1) } \sqrt{ \nu_2(m)} \right) +  \mathcal{O}_p \left( b_T^{-d} L^{ -( \delta + 1 )}  \right)
\end{align*}
In other words, we need to identify which of the three terms of the above expansion stochastically dominates the other. Then, if we find that the particular term is bounded in probability we have proved the desired result. The next step is to employ the $\alpha-$mixing approximation to the above expression. In other words, the fundamental idea to prove the asymptotic normality of $\boldsymbol{W}_T (\boldsymbol{x})$ is to divide the expression for $\boldsymbol{W}_T (\boldsymbol{x})$ into two parts such that $m = m_T \to \infty$ which implies that we can write the following expression withour worrying that the block size can affect the asymptotic approximation result.
\begin{align}
\boldsymbol{W}_T (\boldsymbol{x}) = \boldsymbol{W}_T^{(m)} (\boldsymbol{x}) + \big[ \boldsymbol{W}_T (\boldsymbol{x}) - \boldsymbol{W}_T^{(m)} (\boldsymbol{x})  \big]
\end{align}
Therefore, considering the main result of the lemma that we need to prove, we have that
\begin{align}
\left( T b_T^{d} \right)^{1/2} \mathbb{E} \big[ \boldsymbol{W}_T (\boldsymbol{x}) - \boldsymbol{W}_T^{(m)} (\boldsymbol{x}) \big] \to 0.    
\end{align}
which implies that
\begin{align}
\left( T b_T^{d} \right)^{1/2} \left( \boldsymbol{c}^{\prime} \frac{ \boldsymbol{W}_T (\boldsymbol{x}) - \mathbb{E} \big[ \boldsymbol{W}_T^{(m)} (\boldsymbol{x})  \big]  }{ \sigma } \right)
\end{align}
is asymptotically standard normal as $T \to \infty$. Recalling that 
\begin{align}
\eta_j^{(m)} ( \boldsymbol{x} ) := Z_j^{(m)} K_c \left( \boldsymbol{x} - \boldsymbol{X}_j^{(m)} \right) \ \ \ \text{and} \ \ \  \Delta_j^{(m)} ( \boldsymbol{x} )  := \eta_j^{(m)} ( \boldsymbol{x} )  - \mathbb{E} \left[ \eta_j^{(m)} ( \boldsymbol{x} ) \right]
\end{align}

\begin{proof}
Using the Bahadur representation for the  conditional quantile functional form as in \cite{ren2020local}, the proof is similar to the arguments for the conditional mean functional form case as in \cite{lu2007local}. Suppose that 
\begin{align}
W_n :=
\begin{bmatrix}
w_{n0}
\\
w_{n1}
\end{bmatrix},
\ \ \ ( W_n )_j := 
\left( n h_n^p  \right)^{-1} \sum_{i=1}^n \psi_{\uptau} \left(  Y_i^* \right) \left( \frac{X_i - x}{ h_n } \right)_j K \left( \frac{X_i - x}{ h_n } \right), \ \ \ \text{for} \ \ j = 0,...,p,
\end{align}
Furthermore, we employ the Cramer-Wold device. For all $c := \left( c_0, c_1^{\prime} \right)^{\prime} \in \mathbb{R}^{ p+1}$, let 
\begin{align}
A_n := \left( n h_n^p  \right)^{1/2} c^{\prime} W_n = \phi_{\uptau} (x) \frac{1}{ \sqrt{n h_n^p } } \sum_{i=1}^n \psi_{\uptau} \left(  Y_i^* \right) K \left( \frac{X_i - x}{ h_n } \right)
\end{align}

\newpage

A second-order Taylor expansion of $g(u)$ about x gives
\begin{align}
\mathbb{E} \left[ Y_{in} | X_{in} = u \right] = g(x) + g^{\prime}(x)^{\top} (u-x) + \frac{1}{2} ( u- x )^{\prime} g^{\prime \prime} \big( x + \theta (u-x) \big) (u-x)   
\end{align}
for some scalar $| \theta | < 1$, where $u$ is an open neighborhood of $x$. 

Therefore, using the properties of the kernel $K_c(.)$ we get the following expression:
\begin{align*}
\mathbb{E} &\left[ \eta_j - \eta_j^{(m)} \right]^2 
\leq \mathbb{E} \left[ v_j - v_j^{(m)} \right]^2
\\
&= 
\mathbb{E} \left\{ \left[ \psi_{\uptau} (Y_i^{*}) - \psi_{\uptau}(Y_j^{*(m)}) \right] K_n \left(  \frac{ X_j^{(m)} - x }{ h_n } \right)  + \psi_{\uptau}(Y_j^{*}) \left[ K_n \left( \frac{ X_j - x }{ h_n } \right) - K_n \left(  \frac{ X_j^{(m)} - x }{ h_n } \right) \right]   \right\}^2 
\\
&\leq 
2 \left\{  \mathbb{E} \left[ \psi_{\uptau} (Y_i^{*}) - \psi_{\uptau}(Y_j^{*(m)}) \right]^2 K_n^2 \left(  \frac{ X_j^{(m)} - x }{ h_n } \right)  + \mathbb{E}\ \left[ \psi_{\uptau}(Y_j^{*}) \right] \left[ K_n \left( \frac{ X_j - x }{ h_n } \right) - K_n \left(  \frac{ X_j^{(m)} - x }{ h_n } \right) \right]^2 \right\} 
\end{align*}
Therefore, it can be proved that 
\begin{align}
\nu_{n23} := ( n h_n^p )^{-1} \sum_{ 1 \leq i \leq j \leq n  } \mathbb{E} \left[ \eta_j - \eta_j^{(m)} \right]^2 \to 0.    
\end{align}
\end{proof}
Next we focus on the prove of Theorem 3.2 in \cite{ren2020local} as given below. 
\begin{proof}
We only give the proof of the main argument used in the proof of Theorem 3.2 as below: 
\begin{align*}
\widehat{q}_{\uptau}(x)  - q_{\uptau}(x) 
&= 
\left( n h_n^p \right)^{-1} \phi_{\uptau} \sum_{i=1}^n \psi_{\uptau} (Y_i^{*}) K_i + o_p \left( \frac{1}{H_n} \right)  
\\
&=
\left( n h_n^p \right)^{-1} \phi_{\uptau} \sum_{i=1}^n  \big\{ \psi_{\uptau}(Y_i^{*}) - \mathbb{E}\left[ \psi_{\uptau}(Y_i^{*}) \right] K_i \big\} + \left( n h_n^p \right)^{-1} \phi_{\uptau} \sum_{i=1}^n \mathbb{E}\left[ \psi_{\uptau}(Y_i^{*}) \right] K_i + o_p \left( \frac{1}{H_n} \right).  
\end{align*}
Furthermore, we can that $c = ( 1, 0)$ therefore we obtain the following analytical expression 
\begin{align}
\mathbb{E} \left[ Q_{n1}^2 \right] = \left( n h_n^p \right)^{-1} \psi_{\uptau}^2 \mathbb{E} \big[ c^{\prime} V_n(0) - c^{\prime} \mathbb{E} V_n (0) \big]^2 = \mathcal{O}_p \left(  \left( nh_n^p \right)^{-1} \right) = \mathcal{O}_p \left( H_n^{-2} \right).     
\end{align}
In other words, the key step of this proof is to consider the local Bahadur representation. 
\end{proof}

\newpage

Define with
\begin{align}
\eta_i^{(m)}(x)
&:= \psi_{\uptau} \left( Y_i^{*(m)} \right) K_c \left( \frac{ X_i^{(m)} - x}{h_n} \right)
\\
\zeta_{ni}^{(m)}(x)&:= h_n^{-p/2} \left( \eta_i^{(m)}(x) - \mathbb{E} \left[ \eta_i^{(m)}(x) \right] \right), 
\end{align}
and let $S_n^{(m)} := \sum_{i=1}^n \zeta_{ni}^{(m)}$. Then, it holds that
\begin{align}
n^{-1/2} S_n^{(m)} = \left( n h_n^p \right)^{1/2} c^{\prime} \left( W_n^{(m)}(x) - \mathbb{E} \left[ W_n^{(m)}(x) \right] \right) = A_n^{(m)} - \mathbb{E} \left[ A_n^{(m)} \right]    
\end{align}
Therefore, by decomposing $n^{-1/2} S_n^{(m)}$ into smaller pieces involving "large" and "small" blocks. More specifically, consider 
\begin{align}
U^{(m)} (1,n,x,k) := \sum_{j = k(p^{*} + q^{*}) + p^{*} + 1}^{(k+1)(p^{*} + q^{*})} \zeta_{nj}^{(m)}(x)    
\end{align}
The general approach consists in showing that, as $n \to \infty$, 
\begin{align}
Q_1^{(m)} &:= \left| \mathbb{E} \mathsf{exp} \left[ i u Y^{(m)}(n,x,1) \right] - \prod_{j=0}^{r-1} \mathbb{E} \left[ \mathsf{exp} \left\{ i u Y^{(m)}(n,x,1) \right\} \right] \right| \to 0
\\
Q_2^{(m)} &:= n^{-1} \mathbb{E} \left[ Y^{(m)} (n,x,2) \right]^2 \to 0,
\end{align}

\begin{lemma}
Assume that the probability space $\left( \Omega_n, \mathcal{F}_n, P_n \right)$ supports square integrable random variables $S_{n,1}, S_{n,2},..., S_{n,k_n}$, and that $S_{n,t}$ are adapted to $\sigma-$algebras $\mathcal{F}_{m,t}, 1 \leq t \leq k_n$, where 
\begin{align}
\mathcal{F}_{n,1} \subset \mathcal{F}_{n,2} \subset ... \subset \mathcal{F}_{n,k_n} \subset \mathcal{F}_n.    
\end{align}
Let $X_{n,t} = S_{n,t} - S_{n,t-1}, S_{n,0} = 0$ and $U_{n,t}^2 \sum_{s=1}^t X_{n,s}^2$. 
\end{lemma}

\begin{definition}[Uniform Integrability]
A family of random variables $\left\{ X_i, i \in I \right\}$ is said to be uniformly integrable if the following holds $\displaystyle 
\underset{ a \to \infty }{ \mathsf{lim} } \ \underset{ i \in I }{ \mathsf{sup} } \int_{| X_i| \leq a } | X_i | dP = 0$.
\end{definition}
Then, it suffices to show that for all $\delta > 0$, $\sum_{t=2}^T \mathbb{E} \left[ Y_{Tt}^2 \mathbf{1} \left\{ | Y_{Tt} | > \delta \right\} | \Omega_{T,t-1} \right]  \overset{p}{\to} 0$,such that 
\begin{align*}
\sum_{t=2}^T \mathbb{E} \left[ Y_{Tt} | \mathcal{F}_{T,t-1} \right] 
= 
\sum_{t=2}^{P(T)} Y_{Tt} + \sum_{t=P(T) + 1 }^T \mathbb{E} \left[ Y_{Tt} | \mathcal{F}_{T,t-1} \right] 
=  
\sum_{t=P(T) + 1 }^T Y_{Tt} \overset{p}{\to} 0, 
\end{align*}
Therefore, it can be proved that 
\begin{align}
\underset{ \delta \to 0 }{ \mathsf{lim} }  \underset{ T \to \infty } \ { \mathsf{lim \ sup} } \mathbb{P} \ \left(  \frac{ \tilde{\sigma}_T }{ \sigma_T } > \delta \right) = 1.      
\end{align}

\newpage

\bibliographystyle{apalike}
\bibliography{myreferences1}

\addcontentsline{toc}{section}{References}

\newpage

\end{document}